\documentclass[11pt,oneside,english]{article}
\usepackage{amssymb,amsmath,amsfonts, amsmath,mathtools,amsthm,euscript,mathrsfs,MnSymbol,verbatim,enumerate,multirow,bbding,slashed,multicol,color,array, esint,babel, tikz,tikz-cd, tikz-3dplot,tkz-graph,pgfplots,ytableau,graphicx,float,rotating,hyperref,geometry,mathdots,savesym,cite, wasysym,amscd,graphicx,pifont,float,setspace,multirow,wrapfig,picture,subfigure, amsthm, enumitem}
\usepackage[normalem]{ulem} 
\usepackage[utf8]{inputenc}
\usepackage{prettyref}
\usepackage{datetime}
\usepackage{tabularx}
\allowdisplaybreaks
\numberwithin{equation}{section}
\usetikzlibrary{arrows,positioning,decorations.pathmorphing,   decorations.markings, matrix, patterns}
\hypersetup{colorlinks=true}
\hypersetup{linkcolor=black}
\hypersetup{citecolor=black}
\hypersetup{urlcolor=black}
\theoremstyle{plain}
\newtheorem*{thm*}{Theorem}
\theoremstyle{plain}
\newtheorem{thm}{Theorem}[section]

\newtheorem{lem}[thm]{Lemma}

\theoremstyle{definition}

\newtheorem{rem}[thm]{Remark}

\newcommand{\susu}{$\text{SU}(2)\!\times\!\text{SU}(3)$}
\newcommand{\sug}{$\text{SU}(2)\!\times\!\text{G}_2$}
\newcommand{\tr}{{\bf tr} }

\numberwithin{equation}{section}

\tikzset{
  big arrow/.style={
    decoration={markings,mark=at position 1 with {\arrow[scale=1.5,#1]{>}}},
    postaction={decorate},
    shorten >=0.4pt},
  big arrow/.default=black}
\begin{document}
\begin{titlepage}
\begin{center}
\vspace{2cm}
{\Huge\bfseries  48 Crepant Paths to \susu \\  }
\vspace{2cm}
{\Large
Mboyo Esole$^{\heartsuit}$, Ravi Jagadeesan$^{\spadesuit, \diamondsuit}$, and Monica Jinwoo Kang$^\clubsuit$\\}
\vspace{.8cm}
{\large $^{\heartsuit}$ Department of Mathematics, Northeastern University}\par
{\large  360 Huttington Avenue, Boston, MA 02115, USA}\par
\vspace{.4cm}
{\large $^\spadesuit$ 
Harvard Business School}\par
{Wyss Hall, Soldiers Field,  Boston, MA 02163, USA}\par
\vspace{.4cm}
{\large $^\diamondsuit$ 
Department of Economics, Harvard University}\par
{1805 Cambridge St, Cambridge, MA 02138, USA}\par
\vspace{.4cm}

{\large $^\clubsuit$ Department of Physics,  Harvard University}\par
{ 17 Oxford Street, Cambridge, MA 02138, U.S.A}\par
\vspace{.4cm}
 \scalebox{.95}{\tt  j.esole@northeastern.edu,  ravi.jagadeesan@gmail.com, jkang@fas.harvard.edu }\par
\vspace{2cm}
{ \bf{Abstract}}\\
\end{center}
We study crepant resolutions of Weierstrass models of \susu-models, whose gauge group describes the non-abelian sector of the Standard Model.  The \susu-models are elliptic fibrations characterized by the collision of two Kodaira fibers with dual graphs that are  affine Dynkin diagrams of type  $\widetilde{\text{A}}_1$ and $\widetilde{\text{A}}_2$. 
Once we eliminate those collisions that do not  have crepant resolutions, we are left with six distinct collisions that are related to each other by deformations. 
Each of these six collisions  has eight distinct crepant resolutions whose flop diagram is a hexagon with two legs attached to two adjacent nodes. 
Hence, we consider 48 distinct resolutions that are connected to each other by deformations and flops. 
We determine topological invariants---such as Euler characteristics, Hodge numbers, and triple intersections of fibral divisors---for each of the crepant resolutions.
We analyze  the physics of these fibrations when used as compactifications of M-theory and F-theory on Calabi--Yau threefolds yielding 5d  ${\mathcal N}=1$ and 6d ${\mathcal N}=(1,0)$ supergravity theories respectively. We study the 5d prepotential in the Coulomb branch of the theory and check that the six-dimensional theory is anomaly-free and compatible with a 6d uplift from a 5d theory.

\vfill 

\end{titlepage}

\begin{spacing}{0.6}
\tableofcontents
\end{spacing}
\newpage
\section{Introduction}\label{Sec:Intro}

One of the most powerful consequences of string theory dynamics is that gauge theories can be geometrically engineered by singularities \cite{Witten:1995ex,Gimon:1996rq,Anderson:2017zfm}. In F-theory, the singularities are given by the degenerations of the fiber of an elliptic  fibration \cite{Morrison:1996pp,Bershadsky:1996nh}. The F-theory picture naturally associates to an elliptic fibration  a reductive Lie group $G$, a Lie algebra $\mathfrak{g}=\text{Lie}(G)$, and a representation $\mathbf{R}$ of $G$. Such an elliptic fibration is called a {\em $G$-model}. The dual graphs  of the singular fibers over the generic points of the discriminant locus of the elliptic fibration determine the gauge algebra $\mathfrak{g}$. 

The Lie group $G$ is semi-simple when the discriminant of the elliptic fibration contains at least two irreducible components $\Delta_1$ and $\Delta_2$ such that the dual graph of the singular fiber over the generic point of $\Delta_i$  ($i$ = 1, 2) is reducible. These are called {\em collisions of singularities} and were first studied by Bershadsky and Johanson \cite{Bershadsky:1996nu}. The gauge group $G$ depends on the gauge algebra and  the Mordell--Weil group of the elliptic fibration \cite{Mayrhofer:2014opa}.   When the Mordell--Weil group is trivial, the gauge group $G$ is the unique compact, connected,  and  simply-connected Lie group with Lie algebra $\mathfrak{g}$.
  The compact simply connected semi-simple Lie groups with rank $2$ or $3$ are 
$$
\text{SU}(2)\times\text{SU}(2), \quad\text{SU}(2)\times \text{G}_2, \quad \text{SU}(2)\times \text{Sp}(4), \quad \text{and}\quad \text{SU}(2)\times\text{SU}(3).
$$
The $\text{SU}(2)\times\text{SU}(2)$-, $\text{SU}(2)\times \text{G}_2$-, and \susu-models  could be realized by non-Higgsable clusters \cite{Morrison:2012np} as in \cite{SO4,SU2G2,Grassi:2014zxa}. There are subtleties in realizing an SU($3$) as a non-Higgsable group as discussed in Remark \ref{Rem:NHC}.
The $\text{SU}(2)\times\text{SU}(2)$-model, the $\text{SU}(2)\times \text{G}_2$-model, and the  $\text{Sp}(4)$-model are studied respectively in \cite{SO4}, \cite{SU2G2}, and \cite{EKY2}.
The individual SU($2$) and SU($3$)-models are studied in \cite{ESY1,ES} and  the \susu-model  (realized by the collision III+IV$^s$) has been studied 
from  the point of view of string junctions in \cite{Grassi:2014zxa}. 
While the  group \susu\    is famously the non-Abelian gauge sector of the Standard Model of particle physics \cite[\S 2.4]{Baez:2009dj}, 
 the \susu-model has never been constructed explicitly as a nonsingular variety.

The purpose of this paper is to study the \susu-model  with  associated Lie algebra 
$$\text{E}_3=\text{A}_1\oplus\text{A}_2.$$
  We define elliptic fibrations with collisions of singularities corresponding to a Lie algebra of type E$_3$, we study their geometry and topology, and explore the physics of  compactifications of M-theory and F-theory on such varieties when the elliptic fibration is a Calabi--Yau threefold. By definition, an \susu-model is an elliptic fibration with a  trivial Mordell--Weil group and a discriminant locus containing two irreducible smooth components $S$ and $T$ such that the generic fiber over $S$ and $T$ have respectively dual graphs of affine Dynkin type  $\widetilde{\text{A}}_1$ and $\widetilde{\text{A}}_2$, while the Kodaira type of the fiber over the generic point of any  other irreducible component of the discriminant locus is of type I$_1$ or II. 
The \susu-models examined in this paper are defined by singular Weierstrass models, given  in Section \ref{sec:ManyFaces},  for which we construct explicit crepant resolutions in Section \ref{sec:crepres}.
       Weierstrass models provide convenient canonical birational models for elliptic fibrations since  any elliptic fibration over a smooth base  is birational to a possibly singular Weierstrass model \cite{Deligne.Formulaire,MumfordSuominen,Esole:2017csj}.

We show that there are six distinct types of {\em collisions of singularities} that define Weierstrass models for \susu-models with crepant resolutions, namely:
$$ 
\text{I}_2^{\text{s}}+\text{I}_3^{\text{s}}, \quad \text{I}_2^{\text{ns}}+\text{I}_3^{\text{s}}, \quad \text{III}+\text{I}_3^{\text{s}}, 
\quad 
\text{I}_2^{\text{s}}+\text{IV}^{\text{s}}, \quad \text{I}_2^{\text{ns}}+\text{IV}^{\text{s}}, \quad \text{III}+\text{IV}^{\text{s}}.
$$
 We show that each of the corresponding Weierstrass models has eight distinct crepant resolutions.
 
 These six \susu-models are deformation of each other, where the deformations commute with the resolutions in a way such that the same eight sequences of three blowups are used for each of the six realizations of the \susu-model.  In total, we have a network of 48 distinct elliptic fibrations connected by deformations and flops. All the crepant resolutions of the  \susu-models are listed in Section \eqref{sec:crepres}. 

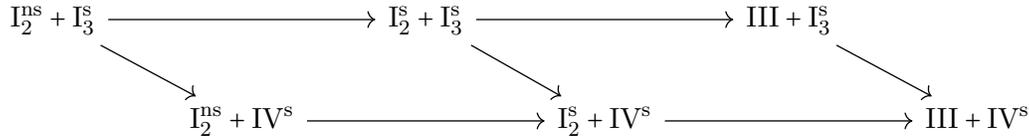
\begin{figure}[H]
\begin{center}
 \begin{tikzcd}[column sep=normal,scale=1.1]
 \text{I}_2^{\text{ns}}+\text{I}_3^{\text{s}}  \arrow[rightarrow]{rr}   \arrow[rightarrow]{dr}  &&  \text{I}_2^{\text{s}}+\text{I}_3^{\text{s}} \arrow[rightarrow]{rd}  \arrow[rightarrow]{rr}  &&  \text{III}+\text{I}_3^{\text{s}}  \arrow[rightarrow]{rd}\\
 & \text{I}_2^{\text{ns}}+\text{IV}^{\text{s}} \arrow[rightarrow]{rr}  && \text{I}_2^{\text{s}}+\text{IV}^{\text{s}} \arrow[rightarrow]{rr}  && \text{III}+\text{IV}^{\text{s}} 
 \end{tikzcd}
\end{center}
\caption{The six collisions of singularities that  give Weierstrass models whose crepant resolutions are smooth \susu-models. In this graph, each arrow indicates a specialization of a Weierstrass model to another. 
The only collision that correspond to a non-Higgsable model is $ \text{III}+\text{IV}^{\text{s}}$. 
\label{Pic:AllModels} }
\end{figure}

\subsection{Matter representation and connection to the Standard Model}\label{sec:sm}

The gauge group of the \susu-model is the non-Abelian sector of the Standard Model of particle physics  \cite[\S 2.4]{Baez:2009dj}.  In this paper, we consider this gauge theory in the context of five and six-dimensional supergravity theories with minimal supersymmetry resulting from a compactification of M-theory or F-theory on a Calabi--Yau threefold $Y$ that corresponds to an elliptic fibration giving an \susu-model. The representation $\mathbf{R}$ of the resulting 5d and 6d theory is reminiscent of the representations of fermions of the Standard Model once we ignore the Abelian sector. 
 
The weights of vertical curves over codimension-one loci of the discriminant locus determine a representation $\mathbf{R}$ of $\mathfrak{g}$ called the {\em matter representation}. See \cite{G2} and reference therein for details on Weierstrass models, weights of curves, dual graphs, and $G$-models.
We show that \susu-models, defined over a base of dimension-two or higher, have vertical rational curves carrying the weights of  the following quaternionic representation $\mathbf{R}$ reminiscent of the representation  $\mathbf{F}$ of fermions of the Standard Model transforming non-trivially under \susu:
$$
\begin{aligned}
\bold{R}  &=(\bold{2},\bold{1})\oplus(\bold{1},\bold{3})\oplus(\bold{1},\bold{\bar{3}})\oplus(\bold{2},\bold{3})\oplus(\bold{2},\bold{\bar{3}})\oplus (\bold{3},\bold{1})\oplus (\bold{1},\bold{8}),\\
\bold{F}& =(\bold{2},\bold{1})\oplus(\bold{1},\bold{3})\oplus(\bold{1},\bold{\bar{3}})\oplus(\bold{2},\bold{3})\oplus(\bold{2},\bold{\bar{3}}),
\end{aligned}
$$
where $(\bold{r}_1,\bold{r_2})$ is the product of the representation $\bf{r}_1$  of the Lie algebra of type $\text{A}_1$ and the representation $\bf{r}_2$  of $\text{A}_2$, and $\overline{\bf{r}}$ is the complex conjugate representation of $\bf{r}$.\footnote{In particular, $(\bold{3},\bold{1})$ is the adjoint representation of $\text{A}_1$, $(\bold{1},\bold{8})$ is the adjoint representation of $\text{A}_2$,  $(\bold{2},\bold{1})$ is the fundamental representation of $\text{A}_1$,  $(\bold{1},\bold{3})$ is the fundamental representation of $\text{A}_2$,  and $(\bold{2},\bold{3})$ is the bifundamental representation of $\text{A}_1\oplus\text{A}_2$.}

In the Standard Model, left-handed leptons transform in the representation $(\bold{2},\bold{1})$ of \susu, left-handed quarks transform in the representation $(\bold{2},\bold{3})$, and right-handed up and down quarks transform in the representation $(\bold{1},\bold{3})$. 
Right-handed leptons are neutral under \susu\ in the Standard Model, and we also have $n_H^0=h^{2,1}(Y)+1$ neutral hypermultiplets in the \susu-model (see equation \eqref{eq:NH0}).
 	 
There are also differences: 
the \susu-model has hypermultiplets transforming under the adjoints representations $(\bold{3},\bold{1})$ and $(\bold{1},\bold{8})$ while the Standard Model does not have fermions transforming in  adjoint representations.

\subsection{Geography of crepant resolutions of $G$-models and hyperplane arrangements}

A crepant resolution is a resolution of singularities that preserves the canonical class. In the case of surfaces, a crepant resolution is necessarily unique and always exists for du Val singularities. Starting from dimension-three, crepant resolutions (when they exist) are not necessarily  unique. For normal threefolds with canonical singularities, the number of crepant resolutions is finite \cite{KMFinite} and two crepant resolutions are connected by flops \cite{Kawamata.flops}. A substitute for crepant resolutions are terminal varieties and they always exist for varieties with canonical singularities \cite{BCHM}.  In light of the result of Kawamata \cite{Kawamata.flops}, it follows from the celebrated results of \cite{BCHM} that minimal models are connected by flops. 
$\mathbb{Q}$-factorial terminal singularities are obstructions to the existence of crepant resolutions as $\mathbb{Q}$-factoriality implies that the exceptional set of any birational map is a divisor, and being terminal implies that all the discrepancies are positive \cite{Kovacs}. 
Canonical singularities,  $\mathbb{Q}$-factorial singularities, terminal singularities, crepant resolutions, and flops  are defined for example in \cite{KM,Matsuki.book}.

 For $G$-models, the geography of flops can be described by the hyperplane arrangement I$(\mathfrak{g},\mathbf{R})$  defined inside the dual fundamental Weyl chamber of $\mathfrak{g}$ and whose hyperplanes are the kernel of the weights of the representation $\mathbf{R}$ \cite{Hayashi:2014kca,EJJN1,EJJN2}.  In stringy geometry, this is a conjecture motivated by the structure of the Coulomb branches of a five-dimensional ${\mathcal N}=1$ supergravity theory with a gauge group $G$ and matter transforming in the representation $\mathbf{R}$ \cite{IMS}. As each Coulomb phase corresponds to a unique crepant resolution of the Weierstrass model, flops corresponds to phase transitions \cite{Witten}, and the different Coulomb phases correspond to the different connected domains in which the prepotential is differentiable; the study of the structure of the 5d Coulomb branches boils down to identifying the chambers of a hyperplane arrangement (see Section \ref{sec:5d}). Mathematically, this fact can be understood by subdivision of a relative movable cone of any of the crepant resolution over the Weierstrass model of a $G$-model into nef cones of each particular crepant resolution in the spirit of \cite{Matsuki.Weyl,Kawamata,Kawamata.CY} and \cite[Theorem 12-2-7]{Matsuki.book}. 

\begin{figure}[H]
\begin{center}
{\begin{tikzpicture}[scale=0.35]
\coordinate (A1) at (90:10cm) {};
\coordinate (A2) at (210:10cm) {};
\coordinate (A3) at (330:10cm) {};
\coordinate (A4) at (0,-5cm) {};
\coordinate (A5) at (-5.768cm,0) {};
\coordinate (A6) at (5.772cm,0) {};
\coordinate (A7) at (-2.89cm, 5cm) {};
\coordinate (A8) at (2.89cm, 5cm) {};
\draw[thick] (A1)--(A2)--(A3)--(A1);
\draw[thick] (A1)--(A4);
\draw[thick] (A5)--(A4);
\draw[thick] (A6)--(A4);
\draw[thick] (A5)--(A8);
\draw[thick] (A6)--(A7);
\node at (-1.2cm, 6cm) {$1'$};
\node at (1.4cm, 6cm) {$1$};
\node at (-3.1cm, 2.6cm) {$2'$};
\node at (3.3cm, 2.6cm) {$2$};
\node at (-1.9cm, -0.5cm) {$3'$};
\node at (2.1cm, -0.5cm) {$3$};
\node at (-5cm, -2.8cm) {$4'$};
\node at (5.2cm, -2.8cm) {$4$};
\end{tikzpicture}}
 \quad\quad
\scalebox{.65}{
\begin{tikzpicture}[every node/.style={circle,draw,thick, minimum size= 10 mm},  scale=.5]
\node (A1) at (60*2-60:6cm) {$1$};
 \node (A2)at (60-60:6cm) {$2$};
 \node (A3) at ($(0-60:6cm)$) {$3$};
 \node (A4) at (-60-60:6cm) {$3'$};
 \node (A5) at (-120-60:6cm) {$2'$};
 \node (A6) at (180-60:6cm) {$1'$};
 \draw[thick] (A1)--(A2)--(A3)--(A4)--(A5)--(A6)--(A1);
 \node (A7) at (-60:12 cm)  {$4$};
 \node (A8) at (-60-60:12cm) {$4'$};
 \draw[thick]  (A3)--(A7);
 \draw[thick] (A4)--(A8);
\end{tikzpicture}}

\caption{The chamber structure of I$(\text{E}_3=\text{A}_1\oplus\text{A}_2, (\mathbf{1},\mathbf{3})\oplus (\mathbf{2},\mathbf{3}))$ and its  adjacency graph. This also represents the structure of the extended K\"ahler cone of an \susu-model.  See Figure \ref{2DChambersMatch} for more details on the walls separating the chambers. }
\label{fig:chambers}
\end{center}
\end{figure}
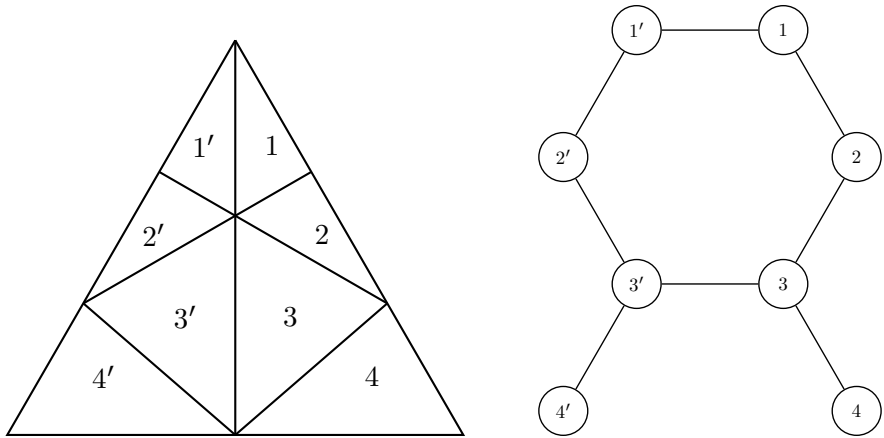

The network of flops  of the \susu-model matches the  hyperplane arrangement I($\text{E}_3=\text{A}_1\oplus\text{A}_2,(\mathbf{1},\mathbf{3})\oplus (\mathbf{2},\mathbf{3})$), where $(\mathbf{1},\mathbf{3})\oplus (\mathbf{2},\mathbf{3})$ is the direct sum of the bifundamental representation   of \susu\  and the fundamental representation of $\text{SU}(3)$. The chamber structure of this hyperplane arrangement is  represented in Figure \ref{fig:chambers}.

\subsection{Topological invariants} 

The Euler characteristic of the elliptically-fibered Calabi--Yau threefold $Y$ is instrumental in the discussion of the gravitational anomalies of the six-dimensional supergravity theory when F-theory is compactified on such $Y$ \cite{GM1}. 
The Hodge numbers of $Y$ determines the number of vector multiplets and neutral hypermultiplets in an M-theory compactification on $Y$ \cite{Cadavid:1995bk}.
In the case of Calabi--Yau fourfolds, the Euler characteristic is relevant to the cancellation of the D3-brane tadpole \cite{Sethi:1996es,CDE,AE1,AE2}. 
The triple intersection numbers of divisors of a Calabi--Yau threefold  $Y$ give the Chern--Simons terms of an M-theory compactification on $Y$.

While  the Euler characteristic and the Betti numbers, are independent of the choice of a crepant resolution as shown by Batyrev \cite{Batyrev.Betti}, the triple intersection numbers are not preserved by flops and depend on a  choice of a crepant resolution \cite{IMS}. 
 We compute  the Euler characteristic of an \susu-model over a base of arbitrary dimension in Theorem \ref{Thm:Eulerchar} in the spirit of \cite{Euler,CharMW,Char}. 
 Denoting by  $B$ the base of the fibration, by $S$ and $T$ the divisors supporting respectively the fiber with dual graphs  $\widetilde{\text{A}}_1$ and $\widetilde{\text{A}}_2$, and by $L$ the first Chern class of the line bundle $\mathscr{L}$ of the Weierstrass model (see Theorem \ref{Thm:PushH}), we have
\begin{align}\nonumber
\chi (Y)=6 \ \frac{S^2-2 L-3 S L+2 (S^2-3 S L+S-2 L) T+(3 S+2) T^2}{(1+S) (1+T) (-1-6 L+2 S+3 T)} \ c(TB) .
\end{align}

We also determine the Hodge numbers of an \susu-model  and the triple intersection numbers   in the case of the Calabi--Yau threefolds. 
Since the triple intersection numbers depend on a choice of a resolution, we compute them for each of the crepant resolutions in Theorem \ref{sec:triple}. 

$$
h^{1,1}(Y)=14-K^2, \quad h^{2,1}(Y)=29 K^2+15 K S+24 K T+3 S^2+6 S T+6 T^2+14,
$$
where $K$ is the canonical class of the base $B$.

\subsection{Compactification of M-theory to $5d$ ${\cal N}=1$ supergravity}

We analyze the physics of the compactifications of M-theory and F-theory on elliptically-fibered Calabi--Yau threefolds corresponding to  \susu-models. 
These give five and six-dimensional  gauge theories coupled to supergravity theories with eight supercharges respectively. For the five-dimensional supergravity theory, we compute the one-loop prepotential in the Coulomb branch (see Theorem \ref{Thm:Prepotentials}), and determine the Chern--Simons couplings, the number of vector multiplets, tensor multiplets, and hypermultiplets. The Chern--Simons couplings are computed geometrically as triple intersection numbers of fibral divisors  in each Coulomb chamber (see Theorem \ref{Thm:TripleInt}). 

The Chern--Simons levels are given by triple intersection numbers of divisors of the Calabi--Yau threefolds and determine the prepotential of a five-dimensional supergravity theory from an M-theory compactification on such Calabi--Yau threefolds \cite{Cadavid:1995bk}. Using such relation, we match the triple intersection polynomial with the one-loop prepotential in each Coulomb chamber to obtain constraints on the number of charged hypermultiplets (see equation \eqref{eq:numbers}). In many cases, such a method will completely fix the number of multiplets \cite{ES,G2,F4,SO4,  EKY2}; however for the \susu-model, we get linear constraints that does not fully determine the number of multiplets.\footnote{When we consider a model with a semisimple Lie algebra and a trivial Mordell--Weil group, we always get linear constraints that does not fully determine the number of multiplets by matching the prepotential and the triple intersection polynomial \cite{SO4,  EKY2,SU2G2,SO4}.} However, they can completely be fixed by using Witten's genus formula, which is a five-dimensional result. In the five-dimensional theory, we also determine the structure of the Coulomb chambers. Each chamber corresponds to a specific crepant resolution that we determine explicitly.

If we denote the number of hypermultiplets transforming in the irreducible representation $\bf{R}$ as $n_{\bf{R}}$ and denote the curve supporting the gauge group SU($2$) and SU($3$) by $S$ and $T$ respectively, and write $K$ for the canonical class of the base of the elliptic threefold,\footnote{The genus $g(S)$ of a supporting curve of $S$ in a surface of canonical class $K$ satisfies $2-2g(S)=-S\cdot K-S^2$. The same holds for the case of $T$ as well.} we have:
$$
\begin{aligned}
n_{\bf{3,1}} &= \frac{1}{2}(K\cdot S+S^2+2),  \quad   && \quad \quad n_{\bf{1,8}} = \frac{1}{2}(K\cdot T+T^2+2), \\
n_{\bf{2,1}} & =-S\cdot (8 K +2S+3T), \quad   &&n_{\bf{1,3}}+n_{\bf{1,\bar{3}}} = -T\cdot (9K+3 T+2S), \quad  && n_{\bf{2,3}} + n_{\bf{2,\bar{3}}}= S\cdot T.
\end{aligned}
$$ 
Interpreting the number of charged hypermultiplets with the genus of the supporting curves of $S$ and $T$, denoted as $g(S)$ and $g(T)$ respectively, we get
$$
\begin{aligned}
n_{\bf{3,1}} &= g(S), \quad   && \quad \quad n_{\bf{1,8}} =g(T), \\
n_{\bf{2,1}} & =16(1-g(S)) + 6 S^2 -3S\cdot T, \quad   &&n_{\bf{1,3}}+n_{\bf{1,\bar{3}}} =  18 (1-g(T)) + 6 T^2 -2 S\cdot T,\quad n_{\bf{2,3}} &+ n_{\bf{2,\bar{3}}}= S\cdot T. 
\end{aligned}
$$

\subsection{Compactification of F-theory to $6d$ $\mathcal{N}=(1,0)$ supergravity and  cancellation of anomalies} 

We also study in detail the anomaly cancellations of the six-dimensional theory. Considering the set of equations from requiring the local anomalies in the six-dimensional theory to cancel, we determine the unique solution of the number of charged hypermultiplets in each irreducible representation. Using Sadov's techniques \cite{Sadov:1996zm}, we check that anomalies are canceled explicitly by the Green-Schwarz mechanism. The number of hypermultiplets in each irreducible representations for an anomaly-free six-dimensional theory matches with those we get from the five-dimensional theories (see Section \ref{sec:anomaly} for the discussion). This ensures that the five-dimensional theory can be uplifted to an anomaly-free six-dimensional theory.

Furthermore, in this paper we show that the matter content of the 6d gauge theory is anomaly-free. This theory can be still anomaly-free after an additional compactification on a Riemann surface to four-dimensions if we impose the condition $n_{\mathbf{3}}=n_{\mathbf{\bar{3}}}$, which will naturally follow from the CPT condition. For local gauge anomalies in four-dimensional spacetime, SU($2$) is always anomaly-free as  all its representation are (pseudo)-real whereas the triangle diagram of SU($3$) is anomaly-free only when the number of matter transforming in the $\mathbf{3}$ and $\mathbf{\overline{3}}$ are equal to each other, which will naturally be satisfied if the  CPT invariance is respected. 

For global anomalies in four spacetime dimensions, one has to worry about SU($2$) as it is subject to Witten anomaly \cite{Witten:1982fp}. This is a direct consequence of the fact that the fourth homotopy group for SU($N$) is zero for $N$ greater than two whereas for SU($2$) it is a $\mathbb{Z}_2$. Thus, Witten anomalies from SU($2$) only vanishes when there is an even number of SU($2$) doublet. 
Similarly, the global anomaly contributions from SU($2$) and SU($3$) in six-dimensions can be discussed by looking into the sixth homotopy group for SU($2$) and SU($3$), which are given by $\mathbb{Z}_{12}$ and $\mathbb{Z}_6$ respectively. Bershadsky and Vafa has shown that this yields the linear constraints on the number of hypermultiplets \cite{Bershadsky:1997sb}.
Considering these conditions only give a constraint on the genus of the curve $S$ supporting the SU($2$) group: 
$
g(S)=0 \ \mod 3,
$
as derived in Section \ref{sec:6dAnomaly}.

\subsection{Organization of this paper}

The rest of the paper is structured as follows. 
In Section \ref{sec:ManyFaces}, we introduce the possible collisions of fibers that produce  \susu-models and define their Weierstrass models.
In Section \ref{SU2SU3Collision}, we describe the eight crepant resolutions of each of the  six Weierstrass models introduced in Section \ref{sec:ManyFaces}. We also  compute the Euler characteristic of the crepant resolutions and the triple intersection of the fibral divisors. 
In the case of Calabi--Yau threefolds, we also compute the Hodge numbers. 
In Section \ref{sec:RepFlops}, we discuss the matter representations of the five and six-dimensional supergravity theories with a gauge group \susu , compute the adjacency graph of the hyperplane arrangement associated with an \susu-model, and match the structure of the hyperplane arrangement with the flopping curves of the crepant resolutions. 
In Sections \ref{sec:I2sI3s} and \ref{sec:IIIIVs},  we study in detail the crepant resolutions of the  I$_2^{\text{s}}+$I$_3^{\text{s}}$-models and the  $\text{III}+\text{IV}^{\text{s}}$-models respectively. 
In Section \ref{sec:Physics}, we study the consequences of our geometric results for the physics of F-theory and M-theory compactified on an \susu-model. We first explain the number of multiplets of the five-dimensional theory using geometric data. We further derive the unique number of charged hypers in each irreducible representations in the uplifted six-dimensional theory canceling the local anomalies and show that it matches that of the five-dimensional theory. We also study the gobal anomaly contributions of the gauge group \susu \cite{Bershadsky:1997sb}.

\section{The many faces of the  \susu-model  }\label{sec:ManyFaces}

In algebraic geometry, the notion of an irreducible variety depends on the choice of the underlying field. 
A variety is said to be {\em geometrically irreducible} when it stays irreducible after any field extension. 
Kodaira has introduced symbols to classify the type of the singular fibers of a minimal elliptic surface \cite{Kodaira}. 
Kodaira symbols can be used more generally to classify the geometric fibers over generic points of the discriminant locus of an elliptic fibration of arbitrary dimension. 
 The types of the  fibers over generic points of the discriminant locus of an elliptic fibration are classified by decorated Kodaira fibers, as the Kodaira type classifies only the geometric fiber \cite[\S 7.4.2]{Esole:2017csj}.
The decoration tracks the minimal field extension under which all  fiber components of the generic fiber become geometrically irreducible. 
Following \cite{Bershadsky:1996nh},  we use the decoration  {\it non-split} (``ns"),  {\it semi-split} (``ss"), and {\it split}  (``s").  
The fibers of type  I$_1$, II, III, II$^*$, and III$^*$ are always well-defined without the need of a field extension. 
The fiber of type  I$_n$ ($n\geq 2$), IV, IV$^*$, and I$_n^{*}$ can be split or non-split. 
The Kodaira fiber of type I$_0^*$ can come from three distinct type:  I$_0^{*s}$, I$_0^{*ss}$, and I$_0^{*ns}$ with respective  dual graphs $\widetilde{\text{D}}_4$, $\widetilde{\text{B}}_3$, and $\widetilde{\text{G}}_2$. The fiber I$_0^{*ns}$ comes in two types as the minimal field extension can be the cyclic group $\mathbb{Z}_3$ or the permutation group of three elements $S_3$ \cite{G2}.

There is a subtlety about the fiber I$_2$:  the fiber I$_2^{\text{s}}$ and  I$_2^{\text{ns}}$ both have two geometrically irreducible components intersecting at two geometric points forming a divisor of degree two on the curve. The fiber I$_2^{\text{ns}}$ corresponds to the case where the two distinct points  are well defined  defined  only after a  a quadratic field extension while the fiber of type I$_2^{\text{s}}$ does not require a field extension to define the two points of intersection.

There are several inequivalent decorated Kodaira fibers that have the same dual graph.
The \susu-model involves the only two root systems that are dual graphs of several distinct Kodaira fibers. 
The dual graph of type  $\widetilde{\text{A}}_1$ is shared by  two different Kodaira fibers, while there are five different Kodaira fibers with the dual graph $\widetilde{\text{A}}_2$, which is presented in Table \ref{table:fibersA1A2}.

\begin{table}[H]
\begin{center}
$$
\scalebox{1.1}{
\begin{tabular}{|c|c|}
\hline 
Fibers & Dual graph \\
\hline 
$\text{I}_2^{\text{s}},\, \text{I}_2^{\text{ns}},\, \text{I}_3^{\text{ns}},\, \text{III},\, \text{IV}^{\text{ns}}$ &$\widetilde{\text{A}}_1$\\
\hline
$\text{I}_3^{\text{s}},\, \text{IV}^{\text{s}} $ &$\widetilde{\text{A}}_2$\\
\hline
\end{tabular}}
$$
\caption{ Decorated Kodaira fibers with  dual graph of $\widetilde{\text{A}}_1$  or $\widetilde{\text{A}}_2$.}
\label{table:fibersA1A2}
\end{center}
\end{table}

The affine root system $\widetilde{\text{A}}_1 \oplus \widetilde{\text{A}}_2$ is the dual graph of ten different pairs of decorated Kodaira fibers.  
 It follows that there are ten distinct ways to realize an \susu-model.   
 Using Tate's algorithm, we realize these  collisions  as singular Weierstrass models by introducing specific valuations of the coefficients with respect to the divisors over which the singular fibers are defined. The minimal multiplicities for the coefficients $a_i$ are reproduced in Table \ref{table:TateSU2SU3}.\\

\begin{table}[htb]
\begin{center}
\begin{tabular}{|c|c|c|c|c|c||c|}
\cline{2-7} 
 \multicolumn{1}{c|}{}& $a_1$ & $a_2$ & $a_3$ & $a_4$ & $a_6$ & $\Delta$ \\
\hline
I$_2 ^{\text{s}}$ & 0 & 1 & 1 & 1 & 2 & 2\\
\hline 
I$_2 ^{\text{ns}}$ & 0 & 0 & 1 & 1 & 2 & 2\\
\hline 
I$_3 ^{\text{ns}}$ & 0 & 0 & 2 & 2 & 3 & 3 \\
\hline
III & 1 & 1 & 1& 1 & 2  & 3  \\
\hline 
IV$^{\text{ns}}$ & 1 & 1 & 1 & 2 & 2 & 4 \\
\hline 
\hline 
I$_3 ^{\text{s}}$ & 0 & 1 & 1 & 2 & 3&3 \\
\hline 
IV$^{\text{s}}$ & 1 & 1 & 1 & 2 & 3 & 4 \\
\hline 
\end{tabular}
\end{center}
\caption{ Valuations of the coefficients of the  Weierstrass models used in this paper for the Kodaira fibers of type I$_2 ^{\text{s}}$, I$_2 ^{\text{ns}}$, I$_3 ^{\text{ns}}$, III, IV$^{\text{ns}}$, I$_3 ^{\text{s}}$, and IV$^{\text{s}}$ to define SU($2$) and SU($3$)-models \cite{Tate.Alg,Bershadsky:1996nh,Katz:2011qp,Esole:2017csj}.  The valuation of the discriminant $\Delta$ follows from the other valuations. }
\label{table:TateSU2SU3}
\end{table}
 
 However, not all corresponding Weierstrass models have a crepant resolution.  In this paper, we only consider  those Weierstrass models that have crepant resolutions.  In particular, we do not use $\text{I}_3^{\text{ns}}$ or $\text{IV}^{\text{ns}}$ to realize $\widetilde{\text{A}}_1$, as the Kodaira fibers of type  $\text{I}_3^{\text{ns}}$ and $\text{IV}^{\text{ns}}$ have $\mathbb{Q}$-factorial terminal singularities. 
Such singularities are obstructions for the existence of crepant resolutions\footnote{ The possible physical relevance of $\mathbb{Q}$-factorial terminal singularities in F-theory is explored in \cite{Arras:2016evy,Grassi:2018rva}.} (this follows, for example, from \cite[Lemma 3.6.2]{BCHM}).
Thus, we consider only the following six cases of collisions:
$$
\text{I}_2^{\text{s}}+\text{I}_3^{\text{s}}, \quad \text{I}_2^{\text{ns}}+\text{I}_3^{\text{s}}, \quad \text{III}+\text{I}_3^{\text{s}}, 
\quad 
\text{I}_2^{\text{s}}+\text{IV}^{\text{s}}, \quad \text{I}_2^{\text{ns}}+\text{IV}^{\text{s}}, \quad \text{III}+\text{IV}^{\text{s}}.
$$
The corresponding Weierstrass models are listed in  equation \eqref{Eq:Weierstrass}:
\begin{eqnarray}\label{Eq:Weierstrass}
\begin{aligned}
\text{I}_2^{\text{ns}}+\text{I}_3^{\text{s}} & : & y(y+a_{1}x+\widetilde{a}_{3}s t)=x^{3}+\widetilde{a}_{2}tx^{2}+\widetilde{a}_{4}st^{2}x+\widetilde{a}_{6}s^{2}t^{3}, \\[5pt] 
\text{I}_2^{\text{s}}+\text{I}_3^{\text{s}} & : & y(y+a_{1}x+\widetilde{a}_{3}st)=x^{3}+\widetilde{a}_{2}stx^{2}+\widetilde{a}_{4}st^{2}x+\widetilde{a}_{6}s^{2}t^{3}, \label{I2sI3s} \\[5pt] 
\text{I}_2^{\text{s}}+\text{IV}^{\text{s}} & : & y(y+\widetilde{a}_{1}tx+\widetilde{a}_{3}st)=x^{3}+\widetilde{a}_{2}stx^{2}+\widetilde{a}_{4}st^{2}x+\widetilde{a}_{6}s^{2}t^{3}, \\[5pt] 
\text{I}_2^{\text{ns}}+\text{IV}^{\text{s}} & : & y(y+\widetilde{a}_{1}tx+\widetilde{a}_{3}st)=x^{3}+\widetilde{a}_{2}tx^{2}+\widetilde{a}_{4}st^{2}x+\widetilde{a}_{6}s^{2}t^{3},\\[5pt] 
\text{III}+\text{I}_3^{\text{s}} & : & y(y+\widetilde{a}_{1}sx+\widetilde{a}_{3}st)=x^{3}+\widetilde{a}_{2}stx^{2}+\widetilde{a}_{4}st^{2}x+\widetilde{a}_{6}s^{2}t^{3}, \\[5pt] 
\text{III}+\text{IV}^{\text{s}} & : & y(y+\widetilde{a}_{1}stx+\widetilde{a}_{3}st)=x^{3}+\widetilde{a}_{2}s t x^{2}+\widetilde{a}_{4}st^{2}x+\widetilde{a}_{6}s^{2}t^{3}.
\end{aligned}
\end{eqnarray}
We assume that the coefficients $a_1$ and $\widetilde{a}_i$ ($i=1,2,3,4,6$) are 
algebraically independent and $S=V(s)$ and $T=V(t)$ are smooth divisors intersecting transversally. The variables $s$ (resp. $t$) is a section of the normal bundle of the divisor $S$ (resp. $T$). In each case, the Kodaira fiber over the generic point of  $S$ (resp. $T$) has  a dual graph of type $\widetilde{\text{A}}_1$ (resp. $\widetilde{\text{A}}_2$). The difference between these models are the valuations of  the Weierstrass coefficients $a_1$ and $a_2$ with respect to $S$ and $T$, which are listed in Table \ref{table:TateSU2SU3}.
The discriminant locus of each Weierstrass model is given by
$$
\Delta= s^{2+a} t^{3+b} (\cdots),
$$
where $a$ is related to $S=V(s)$ and $b$ to $T=V(t)$, $a=0$ for Kodaira fibers of types $\text{I}_2^{\text{s}}$ and $\text{I}_2^{\text{ns}}$, $a=1$ for Kodaira fibers of type III, $b=0$ for Kodaira fibers of type $\text{I}_3^{\text{s}}$, and $b=1$ for Kodaira fibers of type $\text{IV}^{\text{s}}$. We observe that the reduced discriminant locus is composed of three irreducible components.  The fiber degenerates further at the intersection of these three components and we study them to determine the type of matter.

\begin{rem}\label{Rem:NHC}
As seen on Table \ref{table:TateSU2SU3}, the fiber $\text{IV}^{\text{s}}$ becomes the fiber $\text{IV}^{\text{ns}}$ when  the Weierstrass coefficient $a_6$ is deformed by a term of valuation two. 
Such a deformation does not commute with the resolution and changes the gauge group from   SU($3$) to SU($2$). However, the resulting  SU($2$) has $\mathbb{Q}$-factorial terminal singularities. 
Moreover, these groups SU($2$) and SU($3$) coming from fibers IV$^{\text{ns}}$ and IV$^{\text{s}}$ are strongly coupled and related to Argyres-Douglas theories \cite{Argyres:1995jj,Morrison:1996pp},   both give in the weak coupling limit of type IIB string theory an SO($6$) gauge theory \cite{Esole:2012tf}. 
The non-Higgsable group corresponding to a fiber of type IV is SU($2$) as generically a fiber of type IV is a IV$^{\text{ns}}$. 
A non-Higgsable model of type IV$^{\text{s}}$ will require a very particular setting to avoid the existence of a deformation to the fiber IV$^{\text{ns}}$  as it will break the gauge group from SU($3$) to SU($2$). 
\end{rem}

\section{Geometry} \label{SU2SU3Collision}

In this section, we collect  the geometric data -- crepant resolutions, Euler characteristic, Hodge numbers, triple intersection numbers-- of the \susu-models. 
In Section \ref{sec:crepres}, we present eight sequences of blowups that each give crepant resolutions of all six Weierstrass models from equation \eqref{Eq:Weierstrass}. 
 In total, this results in 48 distinct \susu-models. 
In Section \ref{sec:intersectionBL}, we summarize the main pushforward theorems we use to compute the geometric data of \susu-models. 
Since two smooth $n$-dimensional projective algebraic variety over $\mathbb{C}$ connected  by a crepant birational map have the same Betti numbers \cite[Theorem 4.2]{Batyrev.Betti},  all the crepant resolutions of Weierstrass models of  an \susu-model have the same  Euler characteristic.   In Section \ref{sec:Euler},  we give a generating function for the Euler characteristic of an \susu-model. In the case of Calabi--Yau threefolds, we also compute the Hodge numbers of the \susu-models. In the case of a threefold, we also compute the Hodge numbers and the triple intersection numbers in Section \ref{sec:triple}.
In Section \ref{sec:NK}, 
we discuss the  various non-Kodaira fibers obtain from the resolutions of the  \susu-models, we summarize them in Table \ref{fig:NKlist}.

\subsection{Crepant resolutions}
\label{sec:crepres}
We use the following convention. 
 Let $X$ be a nonsingular variety. 
 Let $Z\subset X$ be a complete intersection defined by the transverse intersection of $r$ hypersurfaces $Z_i=V(g_i)$, where $g_i$ is a section of the line bundle $\mathscr{I}_i$ and $(g_1, \cdots, g_r)$ is a regular sequence. 
 We denote the blowup of a nonsingular variety $X$ along the complete intersection $Z$ by 
 $$\begin{tikzpicture}
	\node(X0) at (0,-.3){$X$};
	\node(X1) at (3,-.3){$\widetilde{X}.$};
	\draw[big arrow] (X1) -- node[above,midway]{$(g_1,\cdots ,g_{r}|e_1)$} (X0);	
	\end{tikzpicture}
	$$
The exceptional divisor is $E_1=V(e_1)$.	
 We abuse notation and use the same symbols for $x$, $y$, $s$, $e_i$ and their successive proper transforms. We also do not write the obvious pullbacks.

Assuming some mild regularity conditions on the coefficients of the Weierstrass equations,   each of the  following eight sequences of blowups gives a different crepant resolution of any of the  \susu-model  given by the Weierstrass models in equation \eqref{Eq:Weierstrass}:

\begin{equation}\label{Eq:8Resolutions}
\begin{aligned}
\text{Resolution I} : & \quad
\begin{tikzcd}[column sep=huge] X_0 \arrow[leftarrow]{r} {\displaystyle (x,y, s|e_1)} &  X_1 \arrow[leftarrow]{r} {\displaystyle (x,y,t| w_1)} &  X_2 \arrow[leftarrow]{r} {\displaystyle (y, w_1| w_2)} &  X_3  \end{tikzcd},\\[3pt]
\text{Resolution II} : & \quad
\begin{tikzcd}[column sep=huge] X_0 \arrow[leftarrow]{r} {\displaystyle (x,y, p_0|p_1)} &  X_1 \arrow[leftarrow]{r} {\displaystyle (y,p_1,t| w_1)} &  X_2 \arrow[leftarrow]{r} {\displaystyle (p_0,t| w_2)} &  X_3  \end{tikzcd}, \\[3pt]
\text{Resolution III} : & \quad
\begin{tikzcd}[column sep=huge] X_0 \arrow[leftarrow]{r} {\displaystyle (x,y,t|w_1)} &  X_1 \arrow[leftarrow]{r} {\displaystyle (x,y,s|e_1)} &  X_2 \arrow[leftarrow]{r} {\displaystyle (y, w_1| w_2)} &  X_3  \end{tikzcd}, \\
\text{Resolution IV} : & \quad
 \begin{tikzcd}[column sep=huge] X_0 \arrow[leftarrow]{r} {\displaystyle (x,y,t|w_1)} &  X_1 \arrow[leftarrow]{r} {\displaystyle (y,w_1| w_2)} &  X_2 \arrow[leftarrow]{r} {\displaystyle (x,y,s|e_1)} &  X_3  \end{tikzcd}, \\[3pt]
\text{Resolution I}^\prime :  & \quad
\begin{tikzcd}[column sep=huge] X_0 \arrow[leftarrow]{r} {\displaystyle (x,q,s|e_1)} &  X_1 \arrow[leftarrow]{r} {\displaystyle (x,q,t| w_1)} &  X_2 \arrow[leftarrow]{r} {\displaystyle (q, w_1| w_2)} &  X_3  \end{tikzcd}, \\[3pt]
\text{Resolution II}^\prime : & \quad
\begin{tikzcd}[column sep=huge] X_0 \arrow[leftarrow]{r} {\displaystyle (x,q, p_0|p_1)} &  X_1 \arrow[leftarrow]{r} {\displaystyle (q,p_1,t| w_1)} &  X_2 \arrow[leftarrow]{r} {\displaystyle (p_0,t| w_2)} &  X_3  \end{tikzcd}, \\[3pt]
\text{Resolution III}^\prime : & \quad
\begin{tikzcd}[column sep=huge] X_0 \arrow[leftarrow]{r} {\displaystyle (x,q,t|w_1)} &  X_1 \arrow[leftarrow]{r} {\displaystyle (x,q,s| e_1)} &  X_2 \arrow[leftarrow]{r} {\displaystyle (q, w_1| w_2)} &  X_3  \end{tikzcd}, \\[3pt]
\text{Resolution IV}^\prime : & \quad
\begin{tikzcd}[column sep=huge] X_0 \arrow[leftarrow]{r} {\displaystyle (x,q,t|w_1)} &  X_1 \arrow[leftarrow]{r} {\displaystyle (q, w_1| w_2)} &  X_2 \arrow[leftarrow]{r} {\displaystyle (x,q,s| e_1)} &  X_3  \end{tikzcd},
\end{aligned}
\end{equation}
where $q=y+a_1x+a_3$ and $p_0= s  t$. 

We observe that the sequences of blowups that define the first four resolutions are exactly the same as the sequence of blowups that define the resolutions of the \sug-model \cite{SU2G2}.

The birational map connecting the resolution I to I$'$ (resp. II to II$'$, III to III$'$, and IV to IV$'$) is induced by the involution  $\sigma: [x:y:z]\to [-q:x:z]$ of the Weierstrass model. 
Fiberwise, the involution $\sigma$ is the  inverse map of the Mordell--Weil group: it  maps a point $P$ to its opposite $-P$ with respect to the Mordell--Weil group law. This is familiar from \cite{EY,ESY1,ESY2}. 
The birational maps induced by $\sigma$ are  {\em pseudo-isomorphisms} of the crepant resolutions over the Weierstrass model, as they are isomorphisms in codimension-one. 

With the exceptions of resolutions II and II', all of the resolutions are defined by sequences of blowups around centers that are smooth, complete intersections.
The resolutions II and II$'$ are also defined by a sequence of blowups; however, one of the blowups does not have a smooth center but still defines a regular sequence. Fortunately, this condition is enough to  use the pushforward theorems and compute the topological invariants as in the other cases.

\subsection{Intersection theory and blowups} \label{sec:intersectionBL}
We compute geometric data of the \susu-models  using intersection theory. Our computations rely on three theorems.
 Each crepant resolution of \susu-models is given as a sequence of three blowups, which is listed in equation \eqref{Eq:8Resolutions}.
The first is a theorem of Aluffi that gives the Chern class of a blowup along a local complete intersection. The second theorem is a pushforward theorem that provides a user-friendly method to compute invariants of the resolved space in terms of the original space. The last theorem  gives a simple method to pushforward analytic expressions in the Chow ring of a projective bundle to  the Chow ring of its base.

\begin{thm}[Aluffi, {\cite[Lemma 1.3]{Aluffi_CBU}}]
\label{Thm:AluffiCBU}
Let $Z\subset X$ be the  complete intersection  of $d$ nonsingular hypersurfaces $Z_1$, \ldots, $Z_d$ meeting transversally in $X$.  Let  $f: \widetilde{X}\longrightarrow X$ be the blowup of $X$ centered at $Z$. We denote the exceptional divisor of $f$  by $E$. The total Chern class of $\widetilde{X}$ is then:
\begin{equation}
c( T{\widetilde{X}})=(1+E) \left(\prod_{i=1}^d  \frac{1+f^* Z_i-E}{1+ f^* Z_i}\right)  f^* c(TX).
\end{equation}
\end{thm}

\begin{thm}[Esole--Jefferson--Kang,  see  {\cite{Euler}}] \label{Thm:Push}
Let the nonsingular variety $Z\subset X$ be a complete intersection of $d$ nonsingular hypersurfaces $Z_1$, \ldots, $Z_d$ meeting transversally in $X$. Let $E$ be the class of the exceptional divisor of the blowup $f:\widetilde{X}\longrightarrow X$ centered at $Z$.
 Let $\widetilde{Q}(t)=\sum_a f^* Q_a t^a$ be a formal power series with $Q_a\in A_*(X)$.
 We define the associated formal power series  ${Q}(t)=\sum_a Q_a t^a$, whose coefficients pullback to the coefficients of $\widetilde{Q}(t)$. Then the pushforward $f_*\widetilde{Q}(E)$ is
\begin{equation*}
f_*  \widetilde{Q}(E) =  \sum_{\ell=1}^d {Q}(Z_\ell) M_\ell, \quad \text{where} \quad  M_\ell=\prod_{\substack{m=1\\
 m\neq \ell}}^d  \frac{Z_m}{ Z_m-Z_\ell }.
\end{equation*}
\end{thm}

\begin{thm}[Esole--Jefferson--Kang, see  \cite{Euler} and  \cite{AE1,AE2,EKY,Fullwood:SVW}]\label{Thm:PushH}
Let $\mathscr{L}$ be a line bundle over a variety $B$ and $\pi: X_0=\mathbb{P}[\mathscr{O}_B\oplus\mathscr{L}^{\otimes 2} \oplus \mathscr{L}^{\otimes 3}]\longrightarrow B$ a projective bundle over $B$.  Let $\widetilde{Q}(t)=\sum_a \pi^* Q_a t^a$ be a formal power series in  $t$ such that $Q_a\in A_*(B)$. Define the auxiliary power series $Q(t)=\sum_a Q_a t^a$. 
Then 
\begin{equation*}
\pi_* \widetilde{Q}(H)=-2\left. \frac{{Q}(H)}{H^2}\right|_{H=-2L}+3\left. \frac{{Q}(H)}{H^2}\right|_{H=-3L}  +\frac{Q(0)}{6 L^2},
\end{equation*}
 where  $L=c_1(\mathscr{L})$ and $H=c_1(\mathscr{O}_{X_0}(1))$ is the first Chern class of the dual of the tautological line bundle of  $ \pi:X_0=\mathbb{P}(\mathscr{O}_B \oplus\mathscr{L}^{\otimes 2} \oplus\mathscr{L}^{\otimes 3})\rightarrow B$.
\end{thm}

\subsection{Euler characteristics and Hodge numbers}\label{sec:Euler}
When an elliptic fibration is defined by the resolution of a singular Weierstrass model by a sequence of blowups with smooth centers defining regular embeddings, there are powerful pushforward theorems to compute its Euler characteristic in few simple  algebraic manipulations \cite{Euler}.

\begin{rem}\label{ref:rem}
All the Weierstrass models of the \susu-models and the \sug-models \cite{SU2G2} share four crepant resolutions that are given by the same  sequences of blowups (Resolutions I, II, III, and IV). Hence, the \susu-model for a choice of ($B$, $S$,$T$,$\mathscr{L}$) and the \sug-model defined with the same choice of  ($B$, $S$,$T$,$\mathscr{L}$), have the same the same Euler characteristics as formal expressions in $S$, $T$, $L$, $c(TB)$. Likewise, since \susu\ and \sug\  have the same rank, their Hodge numbers are also identical in the Calabi--Yau threefold case.
\end{rem}

\begin{thm} \label{Thm:Eulerchar}
 The generating polynomial of the Euler characteristic of an \susu-model  obtained by a crepant resolution of a Weierstrass model given in Section \ref{sec:crepres}:
\begin{align}\nonumber
\chi (Y)=6 \ \frac{S^2-2 L-3 S L+2 (S^2-3 S L+S-2 L) T+(3 S+2) T^2}{(1+S) (1+T) (-1-6 L+2 S+3 T)} \ c(TB) .
\end{align}
\end{thm}
\begin{proof} See \cite[Theorem 2.5]{SU2G2}. 
\end{proof}
By direct expansion and specialization, we have the following three lemmas \cite{SU2G2}:
\begin{lem}
For an elliptic threefold, the Euler characteristic is
\begin{align}\nonumber
\chi (Y_3) =-6 (-2 c_1 L+12 L^2+S^2-5 S L+2 S T-8 L T+2 T^2) .
\end{align}
\end{lem}
\begin{lem}
In the case of a Calabi--Yau threefold, by applying $c_1=L=-K$, we have
\begin{align}\nonumber
\chi (Y_3) =-6 (10 K^2+S^2+5 S K+2 S T+8 K T+2 T^2) .
\end{align}
\end{lem}
\begin{lem}
The Euler characteristic for an elliptic fourfold is given by
\begin{align}\nonumber
\begin{split}
\chi (Y_4)=-6 & \left( -2 c_2 L-72 L^3+12 c_1 L^2+ c_1 S^2-5 c_1 S L+2 c_1 S T-8 c_1 L T+2 c_1 T^2   \right. \\
& \ \ \left. +S^3-15 S^2 L +6 S^2 T+54 S L^2-44 S L T+9 S T^2+84 L^2 T-34 L T^2+4 T^3\right) .
\end{split}
\end{align}
\end{lem}
\begin{lem}
The same Calabi--Yau condition $c_1=L=-K$ is applied to get the Euler characteristic for a Calabi--Yau fourfold:
\begin{align}\nonumber
\begin{split}
\chi (Y_4)=-6 & \left(2 c_2 K+60 K^3+S^3+14 S^2 K+6 S^2 T+49 S K^2+42 S K T+9 S T^2  +76 K^2 T+32 K T^2+4 T^3 \right).
\end{split}
\end{align}
\end{lem}

\begin{thm}\label{Thm:Hodge}
In the Calabi--Yau case, the Hodge numbers of an \susu-model 
 given by the crepant resolution of a Weierstrass model given in Section \ref{sec:crepres} are
$$
h^{1,1}(Y)=14-K^2, \quad h^{2,1}(Y)=29 K^2+15 K S+24 K T+3 S^2+6 S T+6 T^2+14.
$$
\end{thm}
\begin{proof}
 See  \cite[Theorem 2.10]{SU2G2}.
\end{proof}

\subsection{Triple intersection numbers} \label{sec:triple}
Let $Y$ be a crepant resolution of an \susu-model defined by one of the crepant resolutions $f:Y\to Y_0$ given in Section \ref{sec:crepres}.
Assuming that $Y$ is a threefold, the triple intersection polynomial of $Y$ is a polynomial containing the divisors $(D_a\cdot D_b \cdot D_c)\cap [Y]$. 
We express a triple intersection polynomial of the \susu-model as a polynomial in $\psi_0$, $\psi_1$, $\phi_0$, $\phi_1$,  and $\phi_2$ that couples respectively with the fibral divisors $D_0^{\text{s}}$, $D_1^{\text{s}}$, $D_0^t$, $D_1^t$, and $D_2^t$. 
The pushforward is expressed in the base by pushing forward to the Chow ring of $X_0$ and then to the base $B$. 
We recall that $\pi:X_0\to B$ is the projective bundle in which the Weierstrass model is defined. Then,
$$
\mathscr{F}_{trip}=
\int_Y\Big[\Big(
\psi_0 D_0^{\text{s}} +\psi_1 D_1^{\text{s}} +\phi_0 D_0^t +\phi_1 D_1^t +\phi_2 D_2^t  
\Big)^3\Big]=
\int_B\pi_* f_* \Big[\Big(
\psi_0 D_0^{\text{s}} +\psi_1 D_1^{\text{s}} +\phi_0 D_0^t +\phi_1 D_1^t +\phi_2 D_2^t  
\Big)^3\Big].
$$
Once the classes of the fibral divisors are determined, all that is left is to compute the pushforward to the base $B$ using the pushforward theorems of Section \ref{sec:intersectionBL}.
\begin{thm}\label{Thm:TripleInt}
The triple intersection polynomial of an \susu-model defined by the  crepant resolutions in Section \ref{sec:crepres}  is 
\quad
\begin{itemize}
\item Resolution I:
\begin{align}\nonumber
\begin{aligned}
\mathscr{F}^{\text{(I)}}_{trip}=& -2 S(2 L+S) \psi _1^3 -6 S T \psi _1 \left(\phi _1^2-\phi _2 \phi _1+\phi _2^2\right) -4T(T-L)\phi _1^3 \\
& -3T(5 L-S-2 T)\phi _1^2\phi _2 -3T(-4 L+S+T) \phi _1\phi _2^2 -T(5 L-2 S+T)\phi _2^3 \\
& +4 S(L-S) \psi _0^3+6 S (S-2 L) \psi _0^2 \psi _1 +12 L S \psi _0\psi _1^2 \\
& -2 T \phi _0^3 (-2 L+S+2 T)+ \phi _0^2 \left(3 T \left(\phi _1+\phi _2\right) (-2 L+S+T)-6 S T \psi _1\right) \\
& +3 T \phi _0 \left(2 \phi _1 \left(L \phi _2+S \psi _1\right)+\phi _1^2 (L-S)+\phi _2^2 (L-S)-2 S \left(\psi _0-\psi _1\right)^2+2 S \psi _1 \phi _2\right) 
\end{aligned}
\end{align}
\item Resolution  II: 
\begin{align}\nonumber
\begin{aligned}
\mathscr{F}^{\text{(II)}}_{trip}=& -S (4 L+2 S-T)\psi_1^3 +3 T (-5 L+S+2 T)\phi_1^2\phi_2 -3 T(-4 L+S+T) \phi_1\phi_2^2 +4 T(L-T)\phi_1^3 \\
& +T (-5 L+S-T)\phi_2^3 -3 S T \psi_1^2 \phi_2 -3ST \psi_1 \left(2\phi_1^2-2\phi_1 \phi_2+\phi_2^2\right)\\
&-S (-4 L + 4 S + T)\psi_0^3 + 3 S (-4 L + 2 S + T) \psi_0^2 \psi_1 + 3 S (4 L - T) \psi_0 \psi_1^2 -T (-4 L + S + 4 T)\phi _ 0^3 \\
& +3 T \phi_0^2 \left(-S (\psi_0 + \psi_1 ) - (2 L-T)\phi_2 \right)+3 T\phi_0\left (L\phi_2^2 - S\left (\psi_0 - \psi_1 \right)^2 + 2 S\psi_0\phi_2 \right)\\
& +3 T \phi 0^2 \phi 1 (-2 L+S+T)+3 T \phi 0 \phi 1 \left(\phi 1 (L-S)+2 L \phi 2+2 S \psi 1\right)-3 S T \psi 0 \phi 2 \left(\psi 0-2 \psi 1+\phi 2\right).
\end{aligned}
\end{align}
\item Resolution III:
\begin{align}\nonumber
\begin{aligned}
\mathscr{F}^{\text{(III)}}_{trip}=& -2S(2 L+S-T)\psi _1^3 -T(-4 L+S+4 T)\phi _1^3 -3T(5 L-S-2 T)\phi _1^2 \phi _2 \\
& -3T(-4 L+S+T)\phi _1\phi _2^2 -T (5 L-S+T)\phi _2^3 -3 S T \psi _1 \left(\phi _1-\phi _2\right)^2 -3 S T \psi _1^2 \left(\phi _1+\phi _2\right) \\
&-2 S(-2 L+2 S+T) \psi _0^3 +\psi _0^2 \left(6 S \psi _1 (-2 L+S+T)-3 S T \left(\phi _1+\phi _2\right)\right) \\
&+\psi _0 \left(-6 S \psi _1^2 (T-2 L)+6 S T \psi _1 \left(\phi _1+\phi _2\right)-3 S T \left(\phi _1^2+\phi _2^2\right)\right)+4 T \phi _0^3 (L-T) \\
&+\phi _0^2 \left(3 T \left(\phi _1+\phi _2\right) (T-2 L)-6 S T \psi _0\right)+3 T \phi _0 \left(\phi _1+\phi _2\right) \left(L \left(\phi _1+\phi _2\right)+2 S \psi _0\right)
\end{aligned}
\end{align}
\item Resolution IV:
\begin{align}\nonumber
\begin{aligned}
\mathscr{F}^{\text{(IV)}}_{trip}=& \ S (-4 L-2 S+3 T)\psi _1^3  -4T(T-L) \phi _1^3 -3T(5 L-2 T) \phi _1^2 \phi _2 \\
& -3T(T-4 L)\phi _1\phi _2^2-T(5 L+T)\phi _2^3-6 S T \psi _1^2 \phi _2  \\
&+\psi _0^2 \left(3 S \psi _1 (-4 L+2 S+3 T)-6 S T \phi _2\right) +3 S \psi _0 \left(\psi _1^2 (4 L-3 T)+4 T \psi _1 \phi _2\right) \\
& -6 ST \psi _0 \left(\phi _1^2-\phi _2 \phi _1+\phi _2^2\right)+S (4 L-4 S-3 T)\psi _0^3 +4 T \phi _0^3 (L-T) \\
& +\phi _0^2 \left(3 T \left(\phi _1+\phi _2\right) (T-2 L)-6 S T \psi _0\right) +3 T \phi _0 \left(\phi _1+\phi _2\right) \left(L \left(\phi _1+\phi _2\right)+2 S \psi _0\right)
\end{aligned}
\end{align}
\end{itemize}
The triple intersection polynomials for the resolutions $\mathrm{I}'$, $\mathrm{II}'$, $\mathrm{III}'$, and $\mathrm{IV}'$ are respectively derived from those of the resolutions  $\mathrm{I}$, $\mathrm{II}$, $\mathrm{III}$, and $\mathrm{IV}$ by the involution $\phi_1 \leftrightarrow \phi_2$.
\end{thm}
\begin{proof}
We give the proof for the case of Resolution I discussed in detail in Section \ref{Sec:ResI}, the other cases follow the same pattern. 
$$
\begin{aligned}
\mathscr{F}_{trip} &=
\int_Y\Big[{\Big(
\psi_0 D_0^{\text{s}} +\psi_1 D_1^{\text{s}} +\phi_0 D_0^t +\phi_1 D_1^t +\phi_2 D_2^t  
\Big)^3}\Big]\\
&=\int_{X_3}\Big[\Big(
\psi_0 D_0^{\text{s}} +\psi_1 D_1^{\text{s}} +\phi_0 D_0^t +\phi_1 D_1^t +\phi_2 D_2^t  
\Big)^3(3H+6L-2E_1-2W_1-W_2)\Big] \\
&=\int_{X_0}f_{1*}f_{2*} f_{3*}\Big[\Big(
\psi_0 D_0^{\text{s}} +\psi_1 D_1^{\text{s}} +\phi_0 D_0^t +\phi_1 D_1^t +\phi_2 D_2^t  
\Big)^3(3H+6L-2E_1-2W_1-W_2)\Big] \\
&=\int_B\pi_* f_{1*}f_{2*} f_{3*} \Big[\Big(
\psi_0 D_0^{\text{s}} +\psi_1 D_1^{\text{s}} +\phi_0 D_0^t +\phi_1 D_1^t +\phi_2 D_2^t  
\Big)^3   (3H+6L-2E_1-2W_1-W_2)\Big].
\end{aligned}
$$
The classes of the fibral divisors in the Chow ring of $X_3$ are  
$$
[D_0^{\text{s}}]=S-E_1 , \quad [D_1^{\text{s}}]=E_1,\quad  [D_0^t]=T-W_1, \quad [D_1^t]=W_1-W_2, \quad  [D_2^t]=W_2. 
$$
Denoting by $M$ an arbitrary divisor in the  class of the Chow ring of the base $B$, the nonzero intersection numbers of the products of $M$, $H$, $E_1$, $W_1$, and $W_2$ are  
$$
\begin{aligned}
& \int_Y E_1^3=-2 S (2 L+S), \quad \int_Y W_1^3=-2 T (2 L-S+T),  \quad \int_Y M W_1^2=-2 TM,\\
&    \int_Y W_2^3=-T (5 L-2 S+T),\ \  \int_Y W_1^2 E_1= -2 S T, \  \int_Y W_1^2 W_2=T (-2 L+S-T),\ \int_Y M E_1^2=-2 SM,\\
& \quad \int_Y W_2^2 E_1=-2 S T, \quad \int_Y W_2^2 W_1= T (-L+S-2 T),\quad \int_Y E_1 W_1 W_2 =-ST, \\
&  \int_Y H M=3  M,\quad \int_Y H^2 M=-9 L M,\quad \int_Y H^3=27L^2, \quad \int_Y M W_2^2 =-2 TM,
\end{aligned}
$$
 where the right-hand-side of each equality is computed in the Chow ring of the base $B$, the pushforward for $f_{i*}$ ($i=1,2,3$)   are obtained via Theorem 
\ref{Thm:Push}, the pushforward for $\pi_*$ uses Theorem \ref{Thm:PushH}.
The triple intersection numbers of the fibral divisors follow from these by simple linearity. 
\end{proof}

The triple intersection polynomials computed in Theorem \ref{Thm:TripleInt} are very different from each other in chambers I, II, III, and IV. 
In chamber III, we get all possible ten homogeneous monomials in $\psi_1$, $\phi_1$ and $\phi_2$. 
In chamber II, we get nine of them ($\psi_1^2\phi_1$ is missing); in chamber I, we are missing  two ($\psi^2_1 \phi_1$ and  $\psi^2_1\phi_2$); in chamber IV, we are missing four (
$\psi^2_1\phi_1$, $\psi_1 \phi_1^2$, $\psi_1\phi_2^2$ and $\psi_1 \phi_1 \phi_2$).  
These facts become handy when comparing the triple intersection polynomials with the prepotentials in Section \ref{sec:5d}. 


\subsection{Non-Kodaira fibers} \label{sec:NK}

The \susu-models have a particularly rich fiber structure with various types of non-Kodaira fibers.
We have  identified a total of  thirteen non-Kodaira fibers over points in codimension-two or three in the base, which are all summarized in Table \ref{fig:NKlist}. 
The fiber structure of each models studied in this paper is described in Appendix  \ref{sec:FiberEnhancement}. 
 Two of these non-Kodaira fibers appear for the first time in the literature. They are contraction of the fiber I$_2^*$ appearing over codimension codimension-three points in the base in the crepant resolution II or IV  of the collision I$_2$+IV$^s$. They are related to the sequence A$_1$+A$_2\rightarrow$D$_5\rightarrow$ D$_6$. All the non-Kodaira fibers obtained at the collision of $S$ and $T$ can be derived by  removing certain nodes on the Kodaira fibers  I$_0^*$, I$_1^*$, I$_2^*$, IV$^*$, or III$^*$. 
Away from $S\cap T$, there is also a non-Kodaira fiber that appears in the specialization of the fiber of type III. 
The phenomena that non-Kodaira fibers of a flat fibration are contractions of Kodaira fibers has been noticed by Miranda in the particular case of his regularization of elliptic threefolds \cite{Miranda.smooth} and also in generalizations of Miranda's models to $n$-folds as studied by Szydlo \cite{Szydlo.Thesis}.
Cattaneo  argues in \cite{Cattaneo} that this is always the case for flat elliptic threefolds that are crepant resolutions of Weierstrass models. 
The  study of crepant resolutions of singular Weierstrass models and the geography of their flops is an essential endeavor in stringy geometry and has been studied for many models producing in this way most  of the known non-Kodaira fibers \cite{Miranda.smooth,Cattaneo,Szydlo.Thesis,EY,Morrison:2011mb,EFY,Lawrie:2012gg,ESY1,ESY2,Tatar:2012tm,Krause:2011xj,SU2G2,EKY,EKY2,SO4,G2,F4,EP1}.

\begin{table}[htb]
\begin{center}
\def\arraystretch{1.5}
 \scalebox{.7}{
\begin{tabular}{|c|ccc|}
\hline
\raisebox{1 cm}{\scalebox{1.5}{I$_0^*$}}
& \includegraphics[trim=-2cm  -0.5cm 0 0cm,scale=.8]{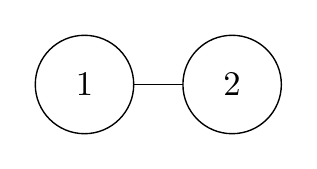} & 
\includegraphics[trim=-2cm  -0.5cm 0 0cm,scale=.8]{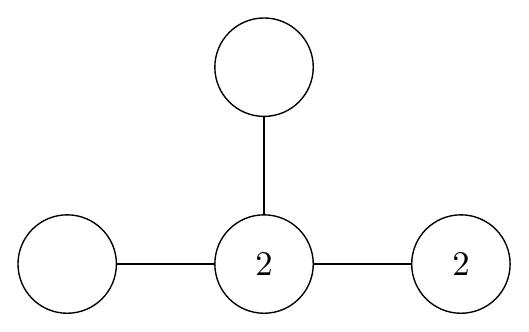} &  \\
\hline 
\raisebox{0 cm}{\scalebox{1.5}{I$_1^*$}}
& \includegraphics[trim=-0.3cm  -0.7cm 0 0cm,scale=.7]{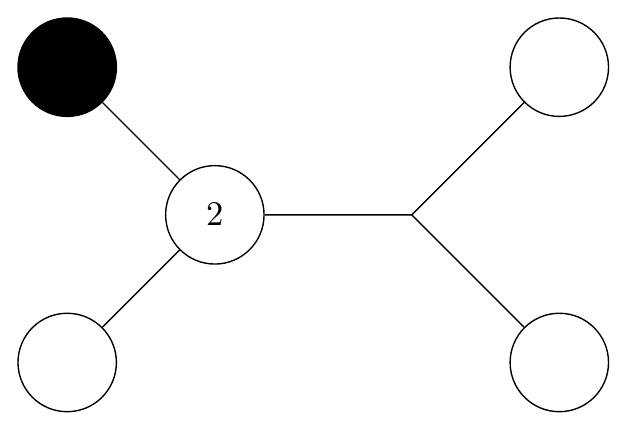} & 
\includegraphics[trim=0cm  -0.5cm 0 0cm,scale=.8]{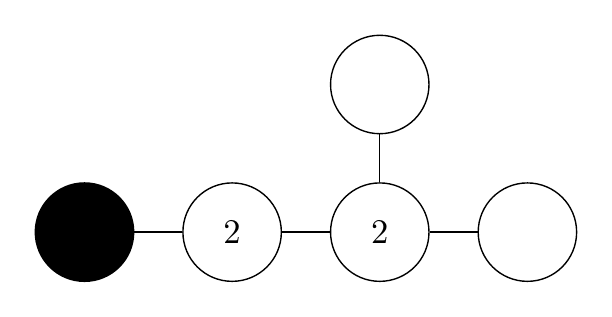} & \\
& 
\includegraphics[trim=0cm  -0.2cm 0 0cm,scale=.7]{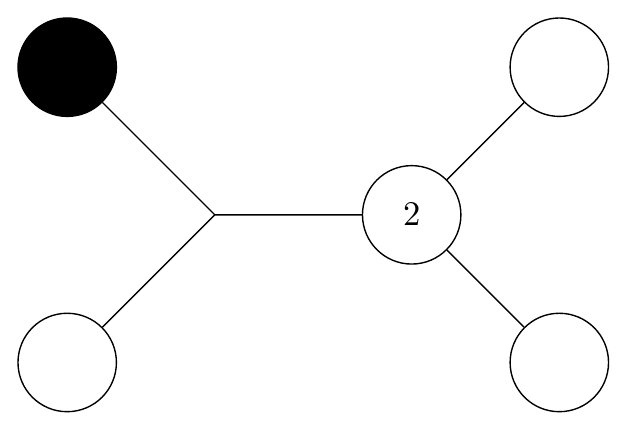} & 
\includegraphics[trim=0cm  -0.2cm 0 0cm,scale=.8]{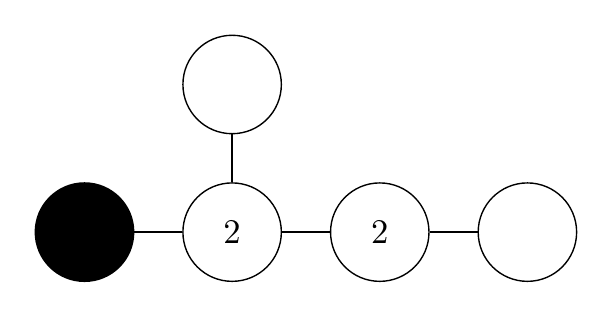} &
\\ 
\hline
\raisebox{1 cm}{\scalebox{1.5}{I$_2^*$}}
&\includegraphics[trim=-0.3cm  -0.5cm 0 0cm,scale=.7]{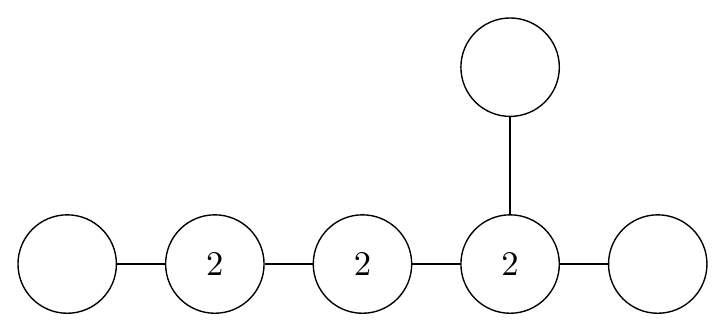} & 
\includegraphics[trim=-0.3cm  -0.5cm 0 0cm,scale=.6]{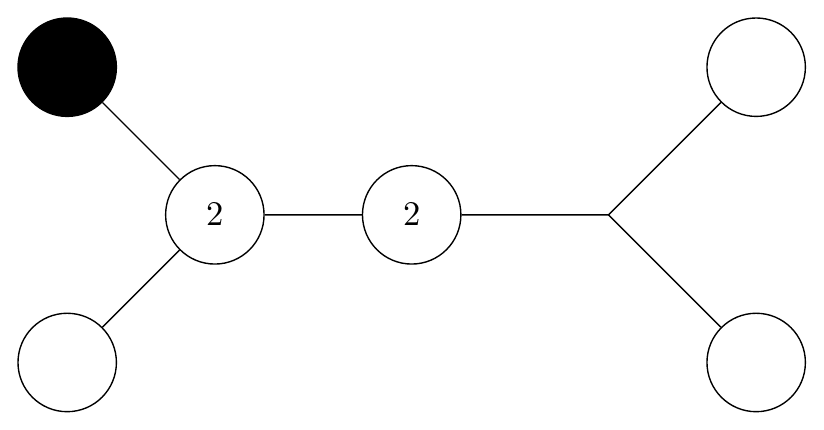}  &
\includegraphics[trim=-0.3cm  -0.5cm 0 0cm,scale=.6]{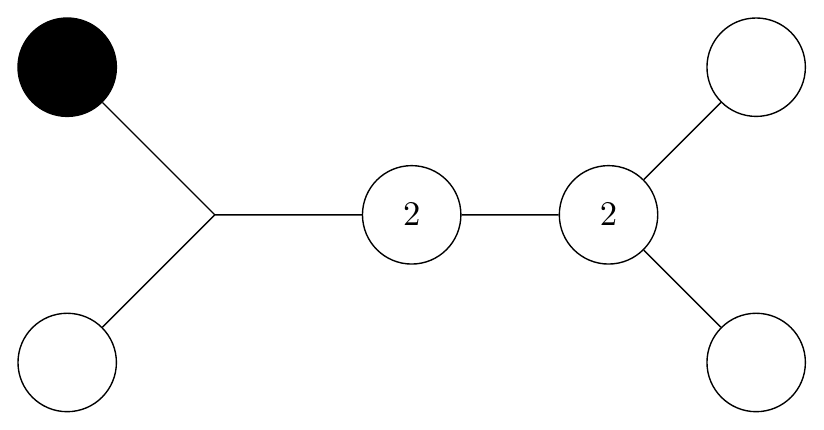} 
\\
\hline
\raisebox{1 cm}{\scalebox{1.5}{IV$^*$}} & 
\includegraphics[trim=0cm  -0.5cm 0 0cm,scale=.8]{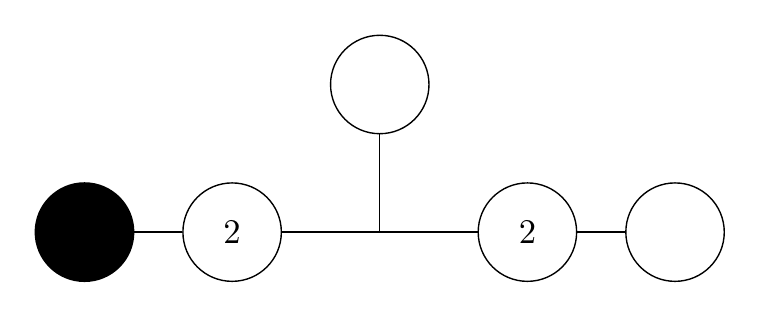} & 
\includegraphics[trim=0cm  -1.2cm 0 0cm,scale=.8]{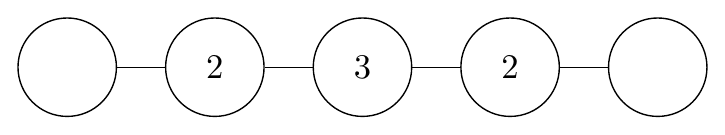} & 
\includegraphics[trim=0cm  -0.5cm 0 0cm,scale=.8]{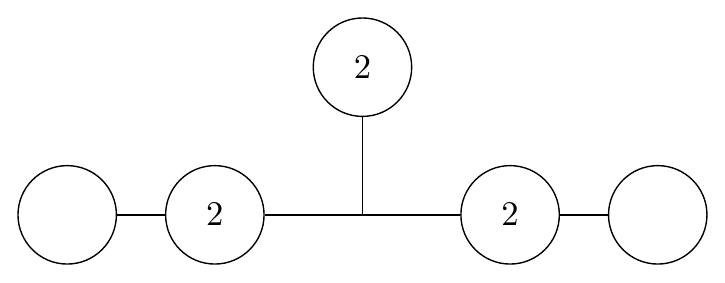} \\
\hline 
\raisebox{1 cm}{\scalebox{1.5}{III$^*$}} & & \includegraphics[trim=0cm  -0.5cm 0 0cm,scale=.8]{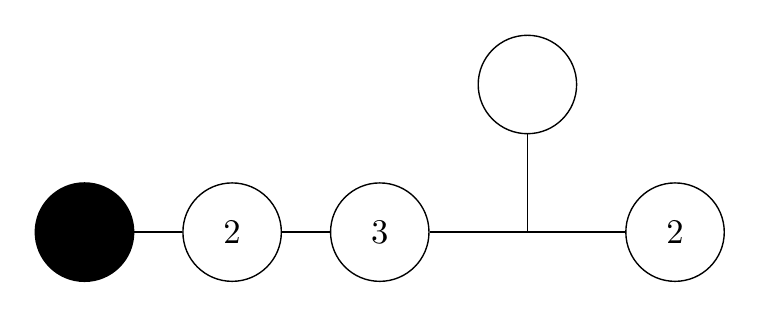} &\\
\hline
\end{tabular}
}
\end{center}
\caption{These are the non-Kodaira fibers observed for the SU$(2)\times$SU$(3)$-models organized by the type of the resulting Kodaira fibers with contracted nodes. When  there is a possible ambiguity, the node that touches the zero section is colored in black. The last two fibers in the row of I$_2^*$ are observed for the first time. }
\label{fig:NKlist}
\end{table}

\section{Hyperplane arrangements and geography of flops}\label{sec:RepFlops}

In Section \ref{sec:RepWeights}, we discuss matter representations of the \susu-model in five and six-dimensional theories with eight supercharges from Katz-Vafa method and confirm its perfect match with the representations computed from the geometry. 

In Section \ref{sec:hyper}, we study the hyperplane arrangement I$({\text{A}}_1\oplus {\text{A}}_2,(\mathbf{1},\mathbf{3})\oplus (\mathbf{2},\mathbf{3}))$. The full representation of the \susu-model is $\bold{R}=(\bold{2},\bold{1})\oplus(\bold{1},\bold{3})\oplus(\bold{1},\bold{\bar{3}})\oplus(\bold{2},\bold{3})\oplus(\bold{2},\bold{\bar{3}})\oplus (\bold{3},\bold{1})\oplus (\bold{1},\bold{8})$. However,  the hyperplane arrangement  I$({\text{A}}_1\oplus {\text{A}}_2,\mathbf{R})$ has the same chamber structure as I$({\text{A}}_1\oplus {\text{A}}_2,(\mathbf{1},\mathbf{3})\oplus (\mathbf{2},\mathbf{3}))$
since the adjoints only define the exterior walls of the dual fundamental Weyl chamber, taking care of the redundancy, and noticing that $((\bold{2},\bold{1}))$ does not contribute interior walls, it is sufficient to consider $(\mathbf{1},\mathbf{3})\oplus (\mathbf{2},\mathbf{3})$ only.

In Section \ref{sec:Corres}, we match the chamber of the hyperplane arrangement I$({\text{A}}_1\oplus {\text{A}}_2,(\mathbf{1},\mathbf{3})\oplus (\mathbf{2},\mathbf{3}))$ with the crepant resolutions as inspired by their interpretation as Coulomb branches of a five-dimensional gauge theory discussed in Section \ref{sec:5d}. 

\subsection{Geometric weights and matter representations}\label{sec:RepWeights}

An  important geometric data is the representation $\mathbf{R}$ under which the matter fields transform. This representation is characterized by its weights, which are computed geometrically by intersection numbers of fibral divisors with vertical curves over codimension-two points. We do not add by hand the chiral conjugates of representations; all representations are seen explicitly by their weights via fibers given by the geometry. 
Starting from a collection of weights, we determine the representation by using the notion of \emph{saturated set of weights} borrowed from Bourbaki. See \cite{F4,G2}  for more information.

The representation $\bold{R}$ that we obtain from purely geometric considerations is consistent with what one would indirectly guess using the Katz-Vafa method \cite{Katz:1996xe}. But the Katz-Vafa method can fail for certain models such as the SU($2$)$\times$ G$_2$ model, while the method of saturations of weight still provides the correct representation $\mathbf{R}$   as discussed in \cite{SU2G2}. 
The sector of $\bold{R}$ that does not contain fundamental representations can be derived from the branching rule of the adjoint representation of  maximal embedding 
$\text{A}_1\oplus\text{A}_2\to \text{A}_4$ while the fundamental  representations follow from the branching rule of the adjoints of the maximal embedding   $\text{A}_1\to \text{A}_2$ and  $\text{A}_2\to \text{A}_3$:
\begin{align}
\begin{cases}
  \mathbf{24}&\longrightarrow (\bold{3},\bold{1})\oplus (\bold{1},\bold{8})\oplus(\bold{2},\bold{3})\oplus(\bold{2},\bold{\bar{3}})\oplus (\bold{1},\bold{\bar{1}}),  \\
 \bf{15}&\longrightarrow \bf{8}\oplus \bf{3} \oplus \bf{\bar{3}}\oplus\bf{1}, \quad \quad    \\
 \bf{8}&\longrightarrow \bf{3}\oplus \bf{2} \oplus \bf{\bar{2}}\oplus {1}.
\end{cases}
\end{align}
A frozen representation is a representation $\mathbf{\rho}$ whose weights are carried by certain curves of the elliptic fibration over codimension-two points in the base, but no hypermultiplet is charged under $\mathbf{\rho}$ \cite{F4,SU2G2}.  
When compactified to a five-dimensional or six-dimensional supergravity theory with eight supercharges, matter in these adjoint representations are frozen when the curves supporting the components of the gauge group are smooth rational curves since the number of adjoint hypermultiplets is given by the arithmetic genus of the curve supporting the gauge group.

\subsection{Hyperplane arrangement} \label{sec:hyper}
We consider the  semi-simple Lie algebra $$\mathfrak{g}=\text{A}_1\oplus \text{A}_2.$$
An irreducible representation of $\text{A}_1\oplus \text{A}_2$ is the tensor product $\bf{r}_1\otimes \bf{r}_2$, where $\bf{r}_1$ and $\bf{r}_2$ are respectively irreducible representations of A$_1$  and A$_2$.  Following a common convention in physics, we denote a  representation of A$_n$ by its dimension in bold character. 
The weights are denoted by $\varpi^I_j$ where the upper index I denotes the representation $\bf{R}_I$ and the lower index $j$ denotes a particular  weight of the representation $\bf{R}_I$. 
A weight of a representation of $\text{A}_1\oplus \text{A}_2$ is denoted by a triple $(a;b,c)$ such that $(a)$ is a weight of A$_1$ and $(b,c)$ is a weight of A$_2$, all in the basis of fundamental weights. 
We use the same notation for coroots. 
 Let $\phi=(\psi_1; \phi_1, \phi_2)$ be a vector of the coroot space of $\text{A}_1 \oplus \text{A}_2$   in the basis of the fundamental coroots.  
 Each weight $\varpi$ defines a linear form $\phi\cdot \varpi$ defined by the natural evaluation on a coroot. We recall that fundamental coroots are dual to fundamental weights. Hence, with our choice of conventions, 
 $\phi\cdot \varpi$ is the usual Euclidian scalar product.

To study the hyperplane arrangement, it is not necessary to consider the full representation  $\bold{R}=(\bold{2},\bold{1})\oplus(\bold{1},\bold{3})\oplus(\bold{1},\bold{\bar{3}})\oplus(\bold{2},\bold{3})\oplus(\bold{2},\bold{\bar{3}})\oplus (\bold{3},\bold{1})\oplus (\bold{1},\bold{8})$  since the adjoints only define the dual fundamental Weyl chamber and $\bold{3}$ and $\bold{\bar{3}}$ defer only by a sign. Thus, we use without loss of generality the representation $\bold{R}$ as:
\begin{equation}
\mathbf{R}=(\bold{2},\bold{1})\oplus(\bold{1},\bold{3})\oplus(\bold{2},\bold{3}),
\end{equation}
 which is the sum of the fundamental  of A$_1$, the fundamental of A$_2$, and the bifundamental representations of A$_1$ and A$_2$. We study the arrangement of hyperplanes  perpendicular to the weights of the representation $\mathbf{R}$ inside the dual  fundamental Weyl chamber of $\text{A}_1 \oplus \text{A}_2$.

The open dual fundamental Weyl chamber is the half cone defined by the  the positivity of the linear form induced by the simple roots: 
\begin{equation}\label{eq:DFW}
\psi_1>0, \quad 2\phi_1-\phi_2>0, \quad -\phi_1+2\phi_2>0. 
\end{equation}
The weight system of the  representation \textbf{2} of A$_1$ and the representation \textbf{3} of A$_2$ are 
\begin{align}
\bf{2}:&\quad \varpi^{\bf{2}}_1=1,\quad \varpi^{\bf{2}}_2=-1\\
\bf{3}:&\quad  \varpi^{\bf{3}}_1=(1,0),\quad \varpi^{\bf{3}}_2=(-1,1),\quad \varpi^{\bf{3}}_3=(0,-1).
\end{align}
The weights of the representation $(\bold{2},\bold{1})$,  $(\bold{1},\bold{3})$ and $(\bold{2},\bold{3})$ are (in the Cartan's basis of fundamental weights).
All the relevant weights are given in Table \ref{Table:Weight}.
\begin{thm}
The hyperplane arrangement $\textnormal{I}(\text{A}_1\oplus\text{A}_2,(\bold{1},\bold{3})\oplus(\bold{2},\bold{3}))$  has eight chambers whose sign vectors   (with respect to the forms   
$(\varpi^{(\bf{1},\bf{3})}_2$,   $\varpi^{(\bf{2},\bf{3})}_5$, $\varpi^{(\bf{2},\bf{3})}_4$,    $\varpi^{(\bf{2},\bf{3})}_3$,   $\varpi^{(\bf{2},\bf{3})}_2)$) are as listed 
in Table \ref{Table:ChambersIneq}. 
The corresponding adjacency graph is  
given in Figure  \ref{2DChambersMatch}. 
\end{thm}
\begin{proof}
There are five  hyperplanes intersecting  the interior of the dual fundamental Weyl chamber: $ 
\varpi^{(\bf{1},\bf{3})}_2, \     \varpi^{(\bf{2},\bf{3})}_2,  \   \varpi^{(\bf{2},\bf{3})}_3,  \   \varpi^{(\bf{2},\bf{3})}_4,  \   \text{ and }  \    \varpi^{(\bf{2},\bf{3})}_5$.
We use them in the order $(\varpi^{(\bf{1},\bf{3})}_2$,   $\varpi^{(\bf{2},\bf{3})}_5$, $\varpi^{(\bf{2},\bf{3})}_4$ ,    $\varpi^{(\bf{2},\bf{3})}_3$,  $ \varpi^{(\bf{2},\bf{3})}_2)$, the sign vector is  $(-\phi_1+\phi_2,-\psi_1 - \phi_1 + \phi_2,-\psi_1 +  \phi_1,\psi_1-\phi_2,\psi_1 - \phi_1 + \phi_2).$
Keeping in mind the conditions in equation \eqref{eq:DFW} defining the open  dual fundamental Weyl  chamber, the results follow from a 
direct check of all possible signs and the chambers are listed on Table \ref{Table:ChambersIneq}. 
\end{proof}

\begin{table}[tbh]
\begin{center}
$
\begin{array}{|c|rrr|}
\hline
\text{Representation}& \multicolumn{3}{c|}{\text{Weights}}\\
\hline
(\bold{2},\bold{1}) &\quad \varpi^{(\bf{2},\bf{1})}_1=(1;0,0)\  \   &\varpi^{(\bf{2},\bf{1})}_2=(-1;0,0)    \  \   & \\
\hline
(\bold{1},\bold{3})&\quad \varpi^{(\bf{1},\bf{3})}_1=(0;1,0)\   \  & \varpi^{(\bf{1},\bf{3})}_2=(0;-1,1)\       \  &  \varpi^{(\bf{2},\bf{1})}_3=(0;0,-1) \\
\hline
(\bold{2},\bold{3}) &\quad  \varpi^{(\bf{2},\bf{3})}_1=(1;1,0)\   \  & \varpi^{(\bf{2},\bf{3})}_2=(1;-1,1)\      \   & \varpi^{(\bf{2},\bf{3})}_3=(1;0,-1)\\
                             &\quad \varpi^{(\bf{2},\bf{3})}_4=(-1;1,0)\  \   & \varpi^{(\bf{2},\bf{3})}_5=(-1;-1,1)\   \   & \varpi^{(\bf{2},\bf{3})}_6=(-1;0,-1)\\
                             \hline
\end{array}
$
\end{center}
\caption{Weights of the representations \label{Table:Weight}}
\end{table}

\subsection{Correspondence between the geometry and the representation theory}\label{sec:Corres}

In this section, we match the crepant resolutions of Section \ref{sec:crepres} and the chambers of the hyperplane arrangement $\text{I}(\mathfrak{g},\mathbf{R})$ of Table \ref{Table:ChambersIneq}. 
The graph of  flops between the crepant resolutions is isomorphic to the adjacency graph of the chambers of the hyperplane arrangement, but the isomorphism is not canonical since the graph has a $\mathbb{Z}_2$ automorphism. A simple way to fix the identification is to compare the triple intersection numbers in each resolution, which are given by Theorem \ref{Thm:TripleInt}, and the prepotentials computed in each chamber, which are given by Theorem \ref{Thm:Prepotentials}.

In Section \ref{SU2SU3Collision}, we described eight different possible resolutions. There  is a $\mathbb{Z}_2$ symmetry in the the structures of the resolutions mapping  resolutions I, II, III, IV and the resolutions I$^\prime$, II$^\prime$, III$^\prime$, IV$^\prime$ and induced by the inverse map of the Mordell--Weil group. Similarly, we see the $\mathbb{Z}_2$ symmetry in the adjacent graph of the chambers between the chambers $1$, $2$, $3$, $4$ and the chambers $1'$, $2'$, $3'$, $4'$. This can be observed easily in Figure \ref{2DChambersMatch}.

We see the explicit correspondence between the chambers and the resolutions as
\begin{equation}
\begin{array}{cccc}
\mathrm{I} \leftrightarrow 1, & \mathrm{II} \leftrightarrow 2, & \mathrm{III} \leftrightarrow 3, & \mathrm{IV} \leftrightarrow 4, \\
\mathrm{I}^\prime \leftrightarrow 1^\prime, & \mathrm{II}^\prime \leftrightarrow 2^\prime, & \mathrm{III}^\prime \leftrightarrow 3^\prime, & \mathrm{IV}^\prime \leftrightarrow 4^\prime .
\end{array}
\end{equation}
We can also compute the weight of the flopping curve between  pairs of resolutions connected by a flop, and compare it with the weight of the wall between the adjacent chambers. As an illustration, we treat the case of I$_2^{\text{s}}+$I$_3^{\text{s}}$-models and show that that the weights of the flopping curves, which are derived in Section \ref{sec:Flops}, match the weights of the walls in Figure \ref{2DChambersMatch}. This solidifies the duality between the chambers and the resolutions, which is represented by the complete structure of the resolutions and the adjacent graph of the chambers juxtaposed in Figure \ref{2DChambersMatch}.

The dual fundamental Weyl chamber is identified with the relative movable cone of a crepant resolution $Y\to Y_0$ over the Weierstrass model $Y_0$. 
This cone is an invariant of minimal models in the same birational class \cite[\S  12-2]{Matsuki.book}. 
The nef cone of any crepant resolution is then identified with a chamber of the hyperplane arrangement I($\mathfrak{g},\mathbf{R})$. 
In particular, two nef cones whose interior coincide represent the same crepant resolution. An interior walls of I($\mathfrak{g},\mathbf{R})$ corresponds to a geometric weight observed up to a sign between two distinct crepant resolutions separated by a flop. 
Two crepant resolutions have nef cones separated by an interior wall they they are connected by an extremal flop  \cite[Propostion  12-2-2]{Matsuki.book}.

\clearpage 
\begin{table}[htb]
\begin{center}
\bgroup
\def\arraystretch{1.5}
$
\begin{array}{|c|ccccc||c|}
\hline
\text{ Subchambers} &   \varpi^{(\bf{1},\bf{3})}_2 &   \varpi^{(\bf{2},\bf{3})}_5   &  \varpi^{(\bf{2},\bf{3})}_4 &     \varpi^{(\bf{2},\bf{3})}_3 &   \varpi^{(\bf{2},\bf{3})}_2 &  \text{Explicit description}  \\
 \hline  
\textcircled{1} & + & - & - & + & + &0< \phi_2-\phi_1<\phi_1< \phi_2< \psi _1  \\
 \hline  
 \textcircled{2} & + & - & - & - & + & 0<\phi_2-\phi_1<\phi_1<\psi_1<\phi_2 \\
 \hline  
 \textcircled{3} & + & - & + & - & +  & 0<\phi_2-\phi_1<\psi_1<\phi_1<\phi_2\\
 \hline  
 \textcircled{4} & + & + & + & - & + &     0<\psi _1<\phi _2-\phi _1<\phi_1<\phi_2\\
 \hline  
 \hline
\textcircled{1'} & - & - & - & + & +  &  0< \phi_1-\phi_2<\phi_2< \phi_1< \psi _1\\
 \hline  
\textcircled{2'} & - & - & + & + & + & 0<\phi_1-\phi_2<\phi_2<\psi_1<\phi_1 \\
 \hline  
\textcircled{3'} & - & - & + & - & + &  0< \phi _1-\phi _2<\psi _1<\phi _2<\phi_1 \\
 \hline  
\textcircled{4'} & - & - & + & - & -  &0<\psi _1<\phi _1-\phi _2<\phi_2<\phi_1 \\
 \hline  
\end{array}
$
\egroup
\end{center}
\caption{ Chambers of the hyperplane arrangement I($\text{A}_1\oplus\text{A}_2, \mathbf{R})$ with $\mathbf{R}=(\bf{1},\bf{3})\oplus(\bf{2},\bf{3})$.  
We will get exactly the same structure if we take the representation $\bold{R}=(\bold{2},\bold{1})\oplus (\bold{1},\bold{3})\oplus (\bold{2},\bold{3})$ since the representation $(\bold{2},\bold{1})$ does not contribute any hyperplane intersecting the interior of the  dual fundamental Weyl chamber. \label{Table:ChambersIneq}
}
\end{table}

\begin{figure}[H]
\begin{tikzpicture}[scale=.8]
\coordinate (A) at (90:5);
\coordinate (C) at (-30:5);
\coordinate (B) at (210:5);
\coordinate (B1) at (barycentric cs:A=2,B=1); 
\coordinate (B2) at (barycentric cs:A=1,B=2); 
\coordinate (C1) at (barycentric cs:A=2,C=1); 
\coordinate (C2) at (barycentric cs:A=1,C=2); 
\coordinate (A2) at (barycentric cs:B=1,C=1); 
\coordinate (A1) at (barycentric cs:A=3,B=1,C=1);
\coordinate (Ib) at  (barycentric cs:A=2,A1=1,B2=1);
\coordinate (IIb) at  (barycentric cs:B1=1,A1=1,B2=1);
\coordinate (IIIb) at  (barycentric cs:A2=1,A1=1,B2=1);
\coordinate (IVb) at  (barycentric cs:A2=1,B=1,B2=1);
\coordinate (I) at  (barycentric cs:A=2,A1=1,C2=1);
\coordinate (II) at  (barycentric cs:C1=1,A1=1,C2=1);
\coordinate (III) at  (barycentric cs:A2=1,A1=1,C2=1);
\coordinate (IV) at  (barycentric cs:A2=1,C=1,C2=1);
\coordinate (D) at (38:3.5); 
\coordinate (D1) at (90:1.67);
\coordinate (D2) at (-43:4.7);
\coordinate (D3) at (70:5.5);
\node[scale=1] at (Ib) {I$'$}; 
\node[scale=1] at (IIb) {II$'$}; 
\node[scale=1] at (IIIb) {III$'$}; 
\node[scale=1] at (IVb) {IV$'$}; 
\node[scale=1] at (I) {I}; 
\node[scale=1] at (II) {II}; 
\node[scale=1] at (III) {III}; 
\node[scale=1] at (IV) {IV}; 
\node[scale=1,right] at (D3) {\begin{tikzcd}[column sep=huge] X_0 \arrow[leftarrow]{r} {\displaystyle (x,y, s|e_1)} &  X_1  \end{tikzcd}};
\node[scale=1,right] at (D2) {\begin{tikzcd}[column sep=huge] X_0 \arrow[leftarrow]{r} {\displaystyle (x,y, t|w_1)} &  X_1  \end{tikzcd}};
\node[scale=1,right] at (D) {\begin{tikzcd}[column sep=huge] X_0 \arrow[leftarrow]{r} {\displaystyle (x,y, p_0|p_1)} &  X_1  \end{tikzcd}};
\node[scale=1] (A2') at (-90:3.5) {$\varpi_2^{\bf{1,3}}$};
\node[scale=1] (B1') at (129:4) {$\varpi_4^{\bf{2,3}}$};
\node[scale=1] (C1') at (51:4) {$\varpi_3^{\bf{2,3}}$};
\node[scale=1] (B2') at (167:4) {$\varpi_2^{\bf{2,3}}$};
\node[scale=1] (C2') at (13:4) {$\varpi_5^{\bf{2,3}}$};
\draw (A)--(B)--(C)--(A);
\draw (B1')--(C2);
\draw (B2)--(C1');
\draw (A)--(A2');
\draw (B2')--(A2);
\draw (C2')--(A2);
\draw[->] (D1)--(D);
\draw[->] (A2)--(D2);
\draw[->] (A)--(D3);
\end{tikzpicture}
\quad
\scalebox{.9}{\begin{tikzpicture}[every node/.style={circle,draw, minimum size= 10 mm},  scale=.45]
\node (A1) at (60*2-60:6cm) {$1$};
 \node (A2)at (60-60:6cm) {$2$};
 \node (A3) at ($(0-60:6cm)$) {$3$};
 \node (A4) at (-60-60:6cm) {$3'$};
 \node (A5) at (-120-60:6cm) {$2'$};
 \node (A6) at (180-60:6cm) {$1'$};
\draw (A1)-- node[draw=none,fill=none,right ]{$\varpi^{(\bf{2},\bf{3})}_3$}(A2)--node[draw=none,fill=none,above, right ]{$\varpi^{(\bf{2},\bf{3})}_4$}(A3)--node[draw=none,fill=none,below]{$\varpi^{(\bf{1},\bf{3})}_2$}(A4)--node[draw=none,fill=none,left]{$\varpi^{(\bf{2},\bf{3})}_3$}(A5)--node[draw=none,fill=none,left]{$\varpi^{(\bf{2},\bf{3})}_4$}(A6)--node[draw=none,fill=none,above]{$\varpi^{(\bf{1},\bf{3})}_2$}(A1);
 \node (A7) at (-60:12 cm)  {$4$};
 \node (A8) at (-60-60:12cm) {$4'$};
\draw  (A3)-- node[draw=none,fill=none,right]{$\varpi^{(\bf{2},\bf{3})}_5$}(A7);
\draw (A4)-- node[draw=none,fill=none,left]{$\varpi^{(\bf{2},\bf{3})}_2$} (A8);
\end{tikzpicture}}
\caption{ \textbf{Left:} the complete structure of the resolutions of \susu. This is a two-dimensional patch of the entire three-dimensional cones. Hence, every point (resp. line) on this picture represents a line (resp. surface).  Accordingly, these eight triangles are the three-dimensional triangular cones. The point in the top of the triangle is the resolution that resolves the SU(2), and the bottom-middle point of the triangle is the first resolution that resolves SU(3) only. The point in the middle of the triangle is the point that describes the blow-up that mixes both SU($2$) and SU($3$). All these three points are connected as expected and it works as the plane of the mirrors. \textbf{Right:} the adjacent graph of the eight chambers of $\text{I}(\mathfrak{g},\bold{R})$  with $\mathfrak{g}=\text{A}_1\oplus \text{A}_2$ and $\bold{R}=(\bold{1},\bold{3})\oplus (\bold{2},\bold{3})$.}
\label{2DChambersMatch}
\end{figure}
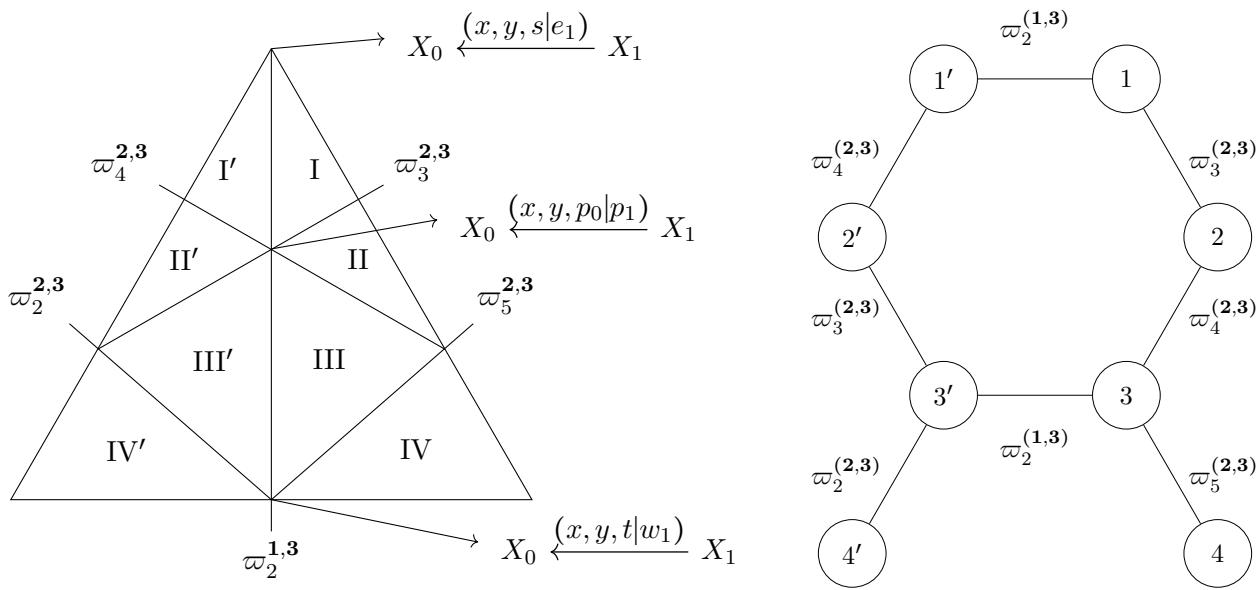

\clearpage

\section{The  I$_2^{\text{s}} +$I$_3^{\text{s}}$ Model \label{sec:I2sI3s}} 
In this section, we study the fiber structure of the crepant resolutions of the I$_2^{\text{s}} +$I$_3^{\text{s}}$ model defined by the following Weierstrass equation:
\begin{equation}
Y_0: \ y^{2}+a_{1}xy+\widetilde{a}_{3}sty=x^{3}+\widetilde{a}_{2}stx^{2}+\widetilde{a}_{4}st^{2}x+\widetilde{a}_{6}s^{2}t^{3}. \label{eq:I2sI3s}
\end{equation}
\subsection{Resolution I}\label{Sec:ResI}
Resolution I is defined by the following sequence of blowups:
\begin{equation} \begin{tikzcd}[column sep=huge] X_0=\mathbb{P}(\mathscr{O}_B\oplus\mathscr{L}^{\otimes 2}\oplus \mathscr{L}^{\otimes 3}) \arrow[leftarrow]{r} {\displaystyle (x,y, s|e_1)} &  X_1 \arrow[leftarrow]{r} {\displaystyle (x,y,t| w_1)} &  X_2 \arrow[leftarrow]{r} {\displaystyle (y, w_1| w_2)} &  X_3  \end{tikzcd} , \end{equation}
where $X_0$ is the projective bundle in which the Weierstrass model is defined; each successive blowup produces a projective bundle over the center of the blowup. The projective coordinates of the fibers of the successive projective bundles are
\begin{equation}
[e_{1}w_{1}w_{2}x\,;\, e_{1}w_{1}w_{2}^{2}y\,;\, z=1][w_{1}w_{2}x\,;\, w_{1}w_{2}^{2}y\,;\, s][x\,;\, w_{2}y\,;\, t][y\,;\, w_{1}].
\end{equation}
The proper transform of $Y_0$ is denoted $Y$ and is a smooth elliptic fibration:
\begin{equation}
Y: \ y(w_{2}y+a_{1}x+\widetilde{a}_{3}st)=w_{1}(e_{1}x^{3}+\widetilde{a}_{2}se_{1}tx^{2}+\widetilde{a}_{4}st^{2}x+\widetilde{a}_{6}s^{2}t^{3}).
\label{proptransf.I2sI3s.res134}
\end{equation}
We denote by $\text{D}^{\text{s}}_a$ and $D^{\text{t}}_a$ the irreducible fibral divisors that project to $S$ and $T$:
\begin{align}
&\text{I}_2^{\text{s}} : \
\begin{cases}
& D^{\text{s}}_0:  s=y(w_{2}y+a_{1}x)-w_{1}e_{1}x^{3}=0 \\
& D^{\text{s}}_1:  e_{1}=y(w_{2}y+a_{1}x+\widetilde{a}_{3}st)-w_{1}(\widetilde{a}_{4}st^{2}x+\widetilde{a}_{6}s^{2}t^{3})=0 \\
\end{cases} \\
&\text{I}_3^{\text{s}} : \
\begin{cases}
& D^{\text{t}}_0: t=y(w_{2}y+a_{1}x)-w_{1}e_{1}x^{3}=0 \\
& D^{\text{t}}_1: w_{1}=w_{2}y+a_{1}x+\widetilde{a}_{3}st=0 \\
& D^{\text{t}}_2: w_{2}=y(a_{1}x+\widetilde{a}_{3}st)-w_{1}(e_{1}x^{3}+\widetilde{a}_{2}se_{1}tx^{2}+\widetilde{a}_{4}st^{2}x+\widetilde{a}_{6}s^{2}t^{3})=0
\end{cases}
\label{5div.I2sI3s.res134}
\end{align}
The generic fiber of $D_a^s$ (resp. $D_a^t$) over $S$ (resp. $T$) is denoted as $C_a^s$ (resp. $C_a^t$). The fiber structure away from the intersection $S\cap T$ is well understood from the study of the individual SU($2$) and SU($3$)-models \cite{ESY1}.
 The generic fiber over the  intersection of $S$ and $T$ is of type I$_5^{\text{s}}$  as in Figure \ref{I5.EWW}, which is produced by the following splittings of $C_a^s$ and $C_a^t$.
\begin{equation}
\text{On} \ S\cap T: \
\begin{cases}
 & \begin{tikzcd}  C^{\text{s}}_0 \arrow[rightarrow]{r}  & \eta_0^0  \end{tikzcd} \\
& \begin{tikzcd}  C^{\text{s}}_1 \arrow[rightarrow]{r}  & \eta_1^{0A}+\eta_1^{0B}+\eta_1^1+\eta_1^2  \end{tikzcd}\\
& \begin{tikzcd}  C^{\text{t}}_0 \arrow[rightarrow]{r}  &\eta_1^{0A}+\eta_1^{0B} + \eta_0^0 \end{tikzcd}\\
& \begin{tikzcd}  C^{\text{t}}_1 \arrow[rightarrow]{r}  & \eta_1^1  \end{tikzcd}\\
& \begin{tikzcd}  C^{\text{t}}_2 \arrow[rightarrow]{r}  & \eta_1^2  \end{tikzcd}
 \end{cases}
 \end{equation}

\begin{equation}
\text{On} \ S\cap T: \
\begin{cases}
\begin{array}{cl}
C^{\text{s}}_{0}\cap C^{\text{t}}_{0}\rightarrow& \eta_0^0: s=t=y(w_{2}y+a_{1}x)-w_{1}e_{1}x^{3}=0\\
C^{\text{s}}_{1}\cap C^{\text{t}}_{0}\rightarrow &\eta_{1}^{0A}: e_{1}=t=y=0, \quad \eta_{1}^{0B} : e_{1}=t=w_{2}y+a_{1}x=0 \\
C^{\text{s}}_{1}\cap C^{\text{t}}_{1}\rightarrow &\eta_1^1: e_{1}=w_{1}=w_{2}y+a_{1}x+\widetilde{a}_{3}st=0\\
C^{\text{s}}_{1}\cap C^{\text{t}}_{2}\rightarrow & \eta_1^2: e_{1}=w_{2}=y(a_{1}x+\widetilde{a}_{3}st)-st^{2}w_{1}(\widetilde{a}_{4}x+\widetilde{a}_{6}st)=0
\end{array}
\end{cases}
\end{equation}

\begin{figure}[H]
\centering
\includegraphics{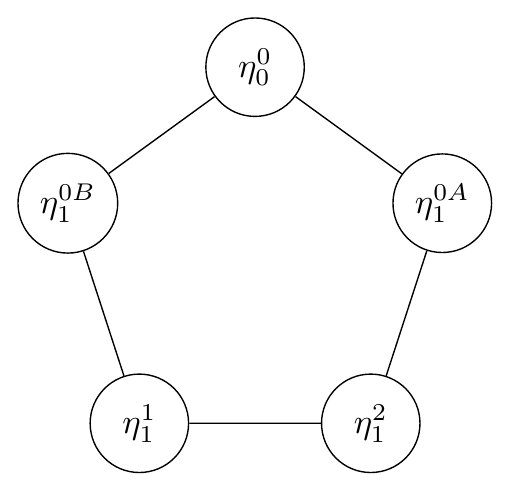}
\caption{Fiber over the generic point of the locus $S\cap T$ in Resolution I of the  I$_2^{\text{s}}$+I$_3^{\text{s}}$-model.\label{I5.EWW}} 
\end{figure}
 The curves $\eta_0^0$, $\eta_1^1$, and $\eta_1^2$ have the same weights as $C_0^{\text{s}}$, $C_1^{\text{t}}$, and $C_2^{\text{t}}$, respectively.
The curves $\eta_1^{0A}$ (resp. $\eta_1^{0B}$) has zero intersection with $D_1^{\text{t}}$ (resp. $D_2^{\text{t}}$). 
The intersection of the curves composing the fiber I$_5^{\text{s}}$ with the fibral divisors are listed on Table \ref{Table:WeightResI}.
\begin{table}[htb]
\begin{center}
\begin{tabular}{|c|c|c|c|c|c|c|c|}
\hline 
 & $D^{\text{s}}_{0}$ &$D^{\text{s}}_{1}$ & $D^{\text{t}}_{0}$ & $D^{\text{t}}_{1}$ & $D^{\text{t}}_{2}$& Weight& Representation\\
\hline 
\hline 
$\eta_0^0$ & -2 & 2 & 0 & 0 & 0& [-2;0,0] & $\bf{(3,1)}$\\
\hline 
$\eta_1^2$ & 0 & 0 & 1 & 1 & -2 & [0;-1,2] &   $\bf{(1,8)}$\\
\hline 
$\eta_1^1$ & 0 & 0 & 1 & -2 & 1& [0;2-1] &   $\bf{(1,8)}$\\
\hline 
$\eta_{1}^{0A}$ & 1 & -1 & -1 & 0 & 1 & [1;0,-1]  & $\bf{(2,3)}$\\
\hline 
$\eta_{1}^{0B}$ & 1 & -1 & -1 & 1 & 0 & [1;-1,0]  & $\bf{(2,\bar{3})}$\\
\hline 
\end{tabular}
\end{center}
\caption{Weights of vertical curves and representations  in the resolution 
 I of the I$_2^{\text{s}}+$I$_3^{\text{s}}$-model.\label{Table:WeightResI}}
\end{table}

The weights of the curves  $\eta_0^0$, $\eta_1^2$, and $\eta_1^1$ are 
are among the weights of the adjoint representation while the weights of the curves  $\eta_1^{0A}$ and $\eta_1^{0B}$ are respectively in the bifundamental representation $\bf{(2,3)}$ and $\bf{(2,\bar{3})}$. 

The fiber I$_5^{\text{s}}$ can degenerate in two different ways by following the degenerations of $\eta_1^{0B}$ and $\eta_1^2$. 
The curve $\eta_1^{0B}$ degenerates at $V(a_1)$, and $\eta_1^2$ is a conic that degenerates at the zero locus of its discriminant. 
The generic fiber over $S\cap T\cap V(a_1)$ is a non-Kodaira fiber corresponding to a contracted fiber of type IV$^*$ described in Figure \ref{NK1.EWW}. 
The generic fiber over  $S\cap T\cap V(a_{1}\widetilde{a}_{6}-\widetilde{a}_{3}\widetilde{a}_{4})$ is an I$_6^2$ fiber obtained by the degeneration of the conic $\eta_1^2$ into two lines intersecting transversally (see  Figure \ref{I6.EWW}).

\begin{equation}
\text{On} \ S\cap T\cap V(a_1): \
\begin{cases}
\eta_0^0  \quad\longrightarrow & \eta_0^0 : s=t=w_{2}y^{2}-w_{1}e_{1}x^{3}=0 \\
\eta_{1}^{0A} \ \longrightarrow & \eta_{1}^{0A} : e_{1}=t=y=0 \\
\eta_{1}^{0B} \ \longrightarrow & \eta_{1}^{0A} : e_{1}=t=y=0 , \ \eta_{1}^{02} : e_{1}=t=w_{2}=0 \\
\eta_1^1 \quad\longrightarrow & \eta_1^1 : e_{1}=w_{1}=w_{2}y+\widetilde{a}_{3}st=0 \\
\eta_1^2 \quad\longrightarrow & \eta_{1}^{02} : e_{1}=w_{2}=t=0 , \ \eta_1^2:  e_{1}=w_{2}=\widetilde{a}_{3}y-tw_{1}(\widetilde{a}_{4}x+\widetilde{a}_{6}st)=0
\end{cases}
\end{equation}
\begin{equation}\label{eq:nonflat}
\text{On} \ S\cap T\cap V(a_{1}\widetilde{a}_{6}-\widetilde{a}_{3}\widetilde{a}_{4}): 
\eta_1^2 \longrightarrow 
\begin{cases}  \eta_{1}^{2A}: a_{1}x+\widetilde{a}_{3}st=0, \\
 \eta_{1}^{2B}: \widetilde{a}_{3}y-\widetilde{a}_{6}st^{2}w_{1}=0
 \end{cases}
\end{equation}
\begin{figure}[H]
\centering
\includegraphics{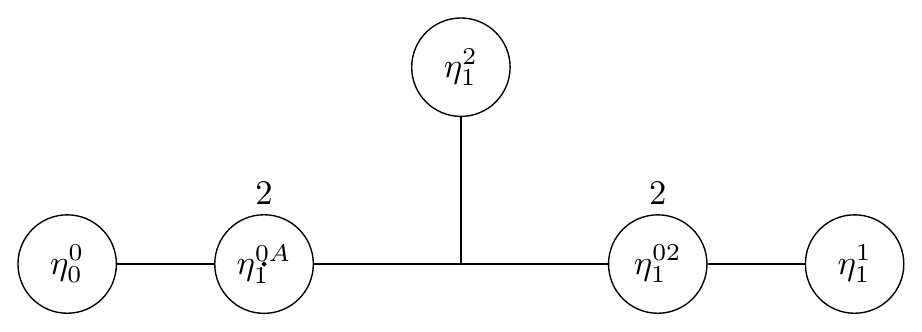}
\caption{Fiber over the generic point of the locus $S\cap T\cap V(a_1)$ in Resolution I of the I$_2^{\text{s}}$+I$_3^{\text{s}}$-model. \label{NK1.EWW}} 
\end{figure}

\begin{figure}[H]
\centering
\includegraphics{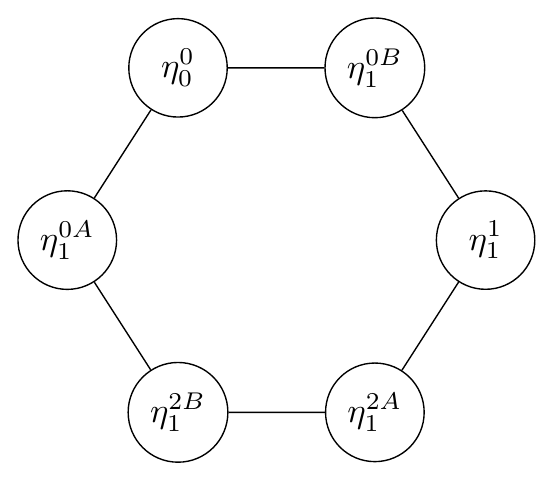}
\caption{Fiber over the generic point of the locus $S\cap T\cap V(a_{1}\widetilde{a}_{6}-\widetilde{a}_{3}\widetilde{a}_{4})$ in Resolution I of the I$_2^{\text{s}}$+I$_3^{\text{s}}$-model. \label{I6.EWW}} 
\end{figure}

As it is clear from  equation \eqref{eq:nonflat}, the curve $\eta_1^{2 A}$ will degenerate to a surface over $S\cap T\cap V(a_1,\widetilde{a}_3)$. For that reason, we assume that the base is at most a threefold  to ensure that the fibration is flat.

\subsection{Resolution II}
 In this section, we study the resolution II.
 In contrast to the other resolutions, some of the centers of the blowups that define resolutions II and II' are singular.  In particular, the first blowup in the sequence of blowups that defines resolution II has a singular center.
 In order to describe the first blowup,  it is useful to rewrite the equation \eqref{eq:I2sI3s} as
\begin{align}
Y_0: \
\begin{cases}
y(y+a_{1}x+\widetilde{a}_{3}p_0)=x^{3}+\widetilde{a}_{2}p_{0}x^{2}+\widetilde{a}_{4}p_0tx+\widetilde{a}_{6}p_0^{2}t\\
p_0=st
\end{cases}.
\end{align}
The resolution II is then given by the following sequence of blowups
\begin{equation} \begin{tikzcd}[column sep=huge] X_0 \arrow[leftarrow]{r} {\displaystyle (x,y, p_0|p_1)} &  X_1 \arrow[leftarrow]{r} {\displaystyle (y,t,p_1| w_1)} &  X_2 \arrow[leftarrow]{r} {\displaystyle (t, p_0| w_2)} &  X_3  \end{tikzcd}, \end{equation}
where $X_0=\mathbb{P}[\mathscr{O}_B\oplus\mathscr{L}^{\otimes 2}\oplus \mathscr{L}^{\otimes 3}]$. The projective coordinates on $X_3$ are then
\begin{equation}
[p_{1}w_1x\,:\, p_{1}w_1^{2}y\,:\, z=1][x\,:\, w_1y\,:\, p_0 w_2][y\,:\, tw_2\,:\, p_1][t\,:\, p_0],
\end{equation}
and the proper transform is 
\begin{align}
Y: \
\begin{cases}
y(w_1y+a_{1}x+\widetilde{a}_{3}p_0w_2)=p_{1}x^{3}+\widetilde{a}_{2}p_{0}p_{1}w_2x^{2}+\widetilde{a}_{4}p_{0}t w_2^2x+\widetilde{a}_{6}p_0^{2}tw_2^{3}\\
p_0p_1=st
\end{cases}.
\end{align}
The variety $X_1=Bl_{(x,y,p_0)} X_0$ has double point singularities at $p_0=p_1= s=t=0$.
The fibral divisors of $Y$  are
\begin{align}
&\text{I}_2^{\text{s}}: \
\begin{cases}
 & D^{\text{s}}_0:  s=p_0=y(w_1y+a_{1}x)-p_{1}x^{3}=0 \\
 & D^{\text{s}}_1:  s=p_{1}=y(w_1y+a_{1}x+\widetilde{a}_{3}p_0w_2)-tw_2^2(\widetilde{a}_{4}p_0x+\widetilde{a}_{6}p_0^{2}w_2)=0 \\
\end{cases} \\
&\text{I}_3^{\text{s}}: \
\begin{cases}
 & D^{\text{t}}_0:  w_2=p_0p_1-st=y(w_1y+a_{1}x)-p_1x^3=0 \\
 & D^{\text{t}}_1:  t=p_1=w_1y+a_{1}x+\widetilde{a}_3p_0w_2=0 \\
 & D^{\text{t}}_2:  w_1=p_0p_1-st=y(a_{1}x+\widetilde{a}_{3}p_0w_2)-(p_{1}x^{3}+\widetilde{a}_{2}p_0p_{1}w_2x^{2}+\widetilde{a}_{4}p_0tw_2^{2}x+\widetilde{a}_{6}p_0^{2}tw_2^{3})=0
\end{cases}
\end{align}

At the intersection of $S$ and $T$ the fiber enhances to an I$_5^{\text{s}}$. This is realized by the following splittings of the curves. 
\begin{equation}
\text{On} \ S\cap T: \
\begin{cases}
& \begin{tikzcd}  C^{\text{s}}_0 \arrow[rightarrow]{r}  & \eta_0^0+\eta_0^2  \end{tikzcd} \\
& \begin{tikzcd}  C^{\text{s}}_1 \arrow[rightarrow]{r}  & \eta_1^0+\eta_1^1+\eta_1^2  \end{tikzcd}\\
& \begin{tikzcd}  C^{\text{t}}_0 \arrow[rightarrow]{r}  & \eta_0^0+\eta_1^0  \end{tikzcd}\\
& \begin{tikzcd}  C^{\text{t}}_1 \arrow[rightarrow]{r}  & \eta_1^1  \end{tikzcd}\\
& \begin{tikzcd}  C^{\text{t}}_2 \arrow[rightarrow]{r}  & \eta_0^2+\eta_1^2  \end{tikzcd}\\
 \end{cases}
 \end{equation}
The curves at the intersection are given by
\begin{equation}
\text{On} \ S\cap T: \
\begin{cases}
C^{\text{s}}_{0}\cap C^{\text{t}}_{0} \quad\rightarrow & \eta_0^0: s=p_{0}=w_2=y(w_1y+a_{1}x)-p_{1}x^{3}=0, \\
C^{\text{s}}_{0}\cap C^{\text{t}}_{2} \quad\rightarrow & \eta_{0}^{2}: s=p_{0}=w_1=a_1y-p_1x^2=0, \\
C^{\text{s}}_{1}\cap C^{\text{t}}_{0} \quad\rightarrow & \eta_1^0: s=p_{1}=w_2=w_1y+a_{1}x=0, \\
C^{\text{s}}_{1}\cap C^{\text{t}}_{1} \quad\rightarrow & \eta_{1}^{1}: s=p_{1}=t=w_1y+a_{1}x+\widetilde{a}_{3}p_{0}w_2=0, \\
C^{\text{s}}_{1}\cap C^{\text{t}}_{2} \quad\rightarrow & \eta_1^2: s=p_{1}=w_1=y(a_{1}x+\widetilde{a}_{3}p_{0}w_2)-p_{0}tw_2^2(\widetilde{a}_{4}x+\widetilde{a}_{6}p_{0}w_2)=0 .
\end{cases}
\end{equation}
This corresponds to $\text{I}^{\text{s}}_{5}$ as in Figure \ref{I5.res2}. The curve $\eta_1^2$ is quadratic in $x$, $y$, and $p_0$ with the discriminant $a_{1}(a_{1}\widetilde{a}_{6}-\widetilde{a}_{3}\widetilde{a}_{4})$.

\begin{table}[htb]
\begin{center}
\begin{tabular}{|c|c|c|c|c|c|c|c|}
\hline 
  & $D^{\text{s}}_{0}$ & $D^{\text{s}}_{1}$ & $D^{\text{t}}_{0}$ & $D^{\text{t}}_{1}$ & $D^{\text{t}}_{2}$& Weight& Representation\\
\hline 
\hline 
$\eta_0^0$ & -1 & 1 & -1 & 0 & 1 & [-1;0,-1] & $\bf{(2,3)}$\\
\hline 
$\eta_0^2$ & -1 & 1 & 1 & 0 & -1 & [-1;0,1] & $\bf{(2,\bar{3})}$\\
\hline 
$\eta_1^1$ & 0  & 0 & 1 & -2 & 1 & [0;2,-1] &   $\bf{(1,8)}$\\
\hline 
$\eta_1^0$ & 1 & -1 & -1 & 1 & 0 & [1;-1,0] &   $\bf{(2,\bar{3})}$\\
\hline 
$\eta_1^2$ & 1 & -1 & 0 & 1 & -1 & [1;-1,1] & $\bf{(2,3)}$\\
\hline 
\end{tabular}
\end{center}
\caption{Weights of vertical curves and representations  in the resolution II of the I$_2^{\text{s}}+$I$_3^{\text{s}}$-model.}
\end{table}

\begin{figure}[H]
\centering
\includegraphics{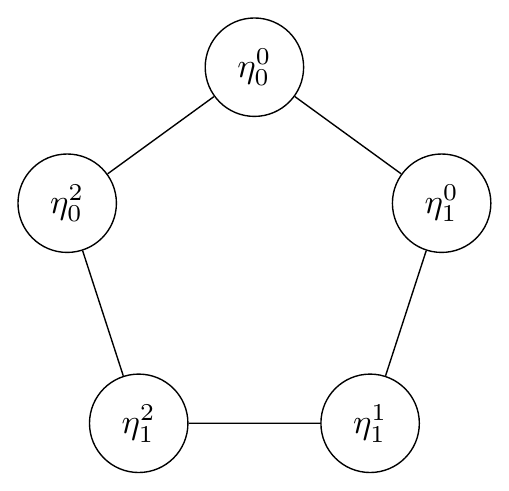}
\caption{Fiber over the generic point of the locus $S\cap T$ in Resolution II of the I$_2^{\text{s}}$+I$_3^{\text{s}}$-model. \label{I5.res2}} 
\end{figure}
There are two enhancements when the discriminant of the curve $\eta_1^2$ vanishes. First enhancement is when $a_{1}=0$:
\begin{equation}
\text{On} \ S\cap T\cap V(a_1): \ 
\begin{cases}
& \begin{tikzcd}  \eta_0^2 \arrow[rightarrow]{r}  & \eta_{01}^2  \end{tikzcd} \\
& \begin{tikzcd}  \eta_1^0 \arrow[rightarrow]{r}  & \eta_1^{02}  \end{tikzcd}\\
& \begin{tikzcd}  \eta_1^2 \arrow[rightarrow]{r}  & \eta_{01}^2+\eta_1^{02}+\eta_1^2  \end{tikzcd}\\
 \end{cases},
 \end{equation}
where the new curves are given by
\begin{align}
\begin{cases}
\eta_{01}^2: & \ s=p_{0}=p_1=w_1=0\\
\eta_{1}^{02}: & \ s=p_{1}=w_2=w_1=0 \\
\eta_1^2: & \ s=p_{1}=w_1=\widetilde{a}_3y-\widetilde{a}_4tw_2x-\widetilde{a}_6p_0tw_2^2=0
\end{cases}.
\end{align}
For this codimension-three enhancement, we get a non-Kodaira fiber that is an incomplete fiber of type IV$^*$, as illustrated in Figure \ref{NK1.res2}.

\begin{figure}[H]
\centering
\includegraphics{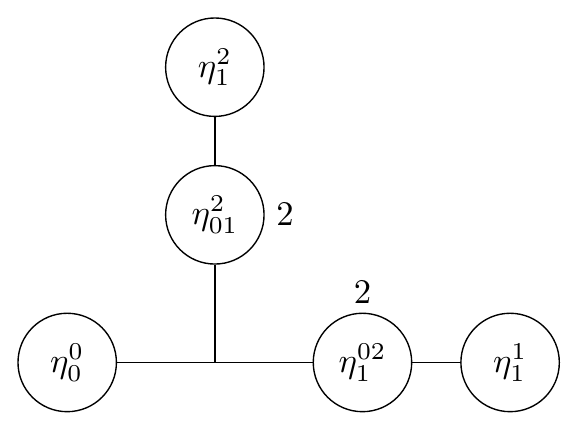}
\caption{Fiber over the generic point of the locus $S\cap T\cap V(a_1)$ in Resolution II of the I$_2^{\text{s}}$+I$_3^{\text{s}}$-model. \label{NK1.res2}} 
\end{figure}

The other specialization is when $(a_{1}\widetilde{a}_{6}-\widetilde{a}_{3}\widetilde{a}_{4})$,  where $\eta_1^2$ splits into two curves that intersect each other such that
\begin{equation}
\text{on} \ S\cap T\cap V(a_{1}\widetilde{a}_{6}-\widetilde{a}_{3}\widetilde{a}_{4}): \
\eta_1^2 \longrightarrow \eta_1^{2A}+\eta_1^{2B},
\end{equation}
while all the other fibers are the same. except $\eta_1^2$. The resulting fiber is a Kodaira fiber of type $\mathrm{I_{6}}$ as in Figure \ref{I6.res2}.

\begin{figure}[H]
\centering
\includegraphics{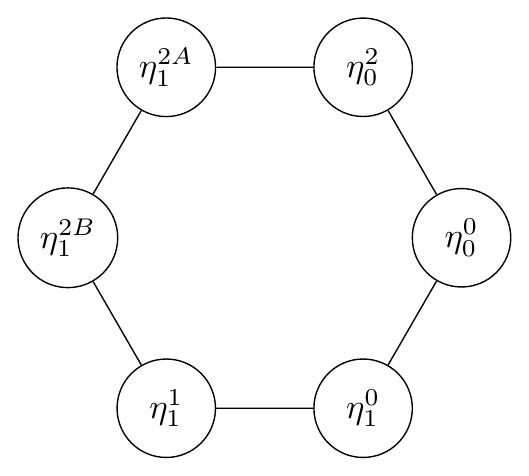}
\caption{Fiber over the generic point of the locus $S\cap T\cap V(a_{1}\widetilde{a}_{6}-\widetilde{a}_{3}\widetilde{a}_{4})$ in Resolution II of the I$_2^{\text{s}}$+I$_3^{\text{s}}$-model.}
\label{I6.res2}
\end{figure}

\subsection{Resolution III}
The  resolution III is defined by the  following sequence of blowups:
\begin{equation} \begin{tikzcd}[column sep=huge] X_0 \arrow[leftarrow]{r} {\displaystyle (x,y, t|w_1)} &  X_1 \arrow[leftarrow]{r} {\displaystyle (x,y, s| e_1)} &  X_2 \arrow[leftarrow]{r} {\displaystyle (y, w_1| w_2)} &  X_3  \end{tikzcd}  \end{equation}
The projective coordinates are then given by
\begin{equation}
[e_{1}w_{1}w_{2}x\,;\, e_{1}w_{1}w_{2}^{2}y\,;\, z=1][e_{1}x\,;\, e_{1}w_{2}y\,;\, t][x\,;\, w_{2}y\,;\, s][y\,;\, w_{1}].
\end{equation}
The proper transform of the elliptic fibration for the resolution III is the same as equation \eqref{proptransf.I2sI3s.res134} and the five fibral divisors are thus identical to those listed on the equation \eqref{5div.I2sI3s.res134}.

As before, we have a fiber of type I$_2^{\text{s}}$ over the generic point of $S$ and fiber of type I$_3^{\text{s}}$ over the generic point of $T$. At the collision of $S$ and $T$, the different curves are 
\begin{equation}
\text{On} \ S\cap T: \ 
\begin{cases}
\begin{array}{clcc}
C_{0}^{\text{s}}\cap C_{0}^t \rightarrow &\eta_0^0: s=t=y(w_{2}y+a_{1}x)-w_{1}e_{1}x^{3}=0 \\
C_{0}^{\text{s}}\cap C_{1}^t \rightarrow &\eta_0^1: s=w_{1}=w_{2}y+a_{1}x=0 & \\
C_{0}^{\text{s}}\cap C_{2}^t\rightarrow & \eta_0^2:  s=w_{2}=a_{1}y-e_{1}w_{1}x^{2}=0 \\
C_{1}^{\text{s}}\cap C_{1}^t\rightarrow &\eta_1^1: e_{1}=w_{1}=w_{2}y+a_{1}x+\widetilde{a}_{3}st=0 \\
C_{1}^{\text{s}}\cap C_{2}^t:\rightarrow & \eta_1^2: e_{1}=w_{2}=y(a_{1}x+\widetilde{a}_{3}st)-st^{2}w_{1}(\widetilde{a}_{4}x+\widetilde{a}_{6}st)=0 
\end{array}
\end{cases}
\end{equation}
The splittings of curves are
\begin{equation}
\text{On} \ S\cap T: \ 
\begin{cases}
& \begin{tikzcd}  C_0^{\text{s}} \arrow[rightarrow]{r}  & \eta_0^0+\eta_0^1+\eta_0^2  \end{tikzcd} \\
& \begin{tikzcd}  C_1^{\text{s}} \arrow[rightarrow]{r}  & \eta_1^1+\eta_1^2  \end{tikzcd}\\
& \begin{tikzcd}  C_0^t \arrow[rightarrow]{r}  & \eta_0^0  \end{tikzcd}\\
& \begin{tikzcd}  C_1^t \arrow[rightarrow]{r}  & \eta_0^1+\eta_1^1  \end{tikzcd}\\
& \begin{tikzcd}  C_2^t \arrow[rightarrow]{r}  & \eta_0^2+\eta_1^2  \end{tikzcd}\\
\end{cases}
\end{equation}
This corresponds to $\mathrm{I}_{5}^{\text{s}}$, which is represented in Figure \ref{I5.WEW}.
\begin{figure}[H]
\begin{center}
\includegraphics{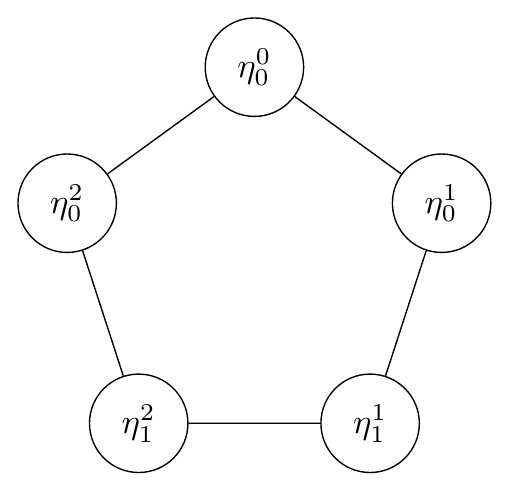}
\caption{Fiber over the generic point of the locus $S\cap T$ in Resolution III of the I$_2^{\text{s}}$+I$_3^{\text{s}}$-model. \label{I5.WEW}} 
\end{center}
\end{figure}
From the splittings of the curves, we compute the intersection numbers to get the weight vectors and further deduce the representations.
\begin{table}[htb]
\begin{center}
\begin{tabular}{|c|c|c|c|c|c|c|c|}
\hline 
  & $D^{\text{s}}_{0}$ & $D^{\text{s}}_{1}$ & $D^{\text{t}}_{0}$ & $D^{\text{t}}_{1}$ & $D^{\text{t}}_{2}$& Weight& Representation\\
\hline 
\hline 
$\eta_0^0$ & 0 & 0 & -2 & 1 & 1& [0;-1,-1] & $\bf{(1,8)}$\\
\hline 
$\eta_0^1$ & -1 & 1 & 1 & -1 & 0& [-1;1,0] & $\bf{(2,3)}$\\
\hline 
$\eta_1^1$ & 1 & -1 & 0 & -1 & 1& [1;1,-1] & $\bf{(2,\bar{3})}$\\
\hline 
$\eta_1^2$ & 1 & -1 & 0 & 1 & -1& [1;-1,1] & $\bf{(2,3)}$\\
\hline 
$\eta_0^2$ & -1 & 1 & 1 & 0 & -1& [-1;0,1] & $\bf{(2,\bar{3})}$\\
\hline 
\end{tabular}
\end{center}
\caption{Weights of vertical curves and representations  in the resolution
 III of the I$_2^{\text{s}}+$I$_3^{\text{s}}$-model.}
\end{table}

This  has two further specializations when the discriminant of $\eta_1^2$ vanishes. The first enhancement is when $a_{1}=0$. We  observe the following splittings for the elliptical fibrations: 
\begin{align}
\text{On} \ S\cap T\cap V(a_1): \ 
\begin{cases}
\eta_0^0 \quad\longrightarrow & \eta_0^0: s=t=w_{2}y^{2}-w_{1}e_{1}x^{3}=0 \\
\eta_0^1 \quad\longrightarrow & \eta_0^{12}: s=w_{1}=w_{2}=0 \\
\eta_0^2 \quad\longrightarrow & \eta_0^{12}: s=w_{2}=w_{1}=0, \  \eta_{01}^2: s=w_{2}=e_{1}=0 \\
\eta_1^1 \quad\longrightarrow & \eta_1^1: e_{1}=w_{1}=w_{2}y+\widetilde{a}_{3}st=0 \\
\eta_1^2 \quad\longrightarrow  & \eta_{01}^2: e_{1}=w_{2}=s=0, \eta_1^2 : e_{1}=w_{2}=\widetilde{a}_{3}y-tw_{1}(\widetilde{a}_{4}x+\widetilde{a}_{6}st)=0
\end{cases}
\end{align}
The splittings from the five divisors to the codimension-three enhancement when $a_1=0$ is 
\begin{equation}
\text{On} \ S\cap T\cap V(a_1): \ 
\begin{cases}
& \begin{tikzcd}  C_0^{\text{s}} \arrow[rightarrow]{r}  & \eta_0^0+\eta_0^{12}+\eta_{01}^2  \end{tikzcd} \\
& \begin{tikzcd}  C_1^{\text{s}} \arrow[rightarrow]{r}  & \eta_{01}^2+\eta_1^1+\eta_1^2  \end{tikzcd}\\
& \begin{tikzcd}  C_0^t \arrow[rightarrow]{r}  & \eta_0^0  \end{tikzcd}\\
& \begin{tikzcd}  C_1^t \arrow[rightarrow]{r}  & \eta_0^{12}+\eta_1^1  \end{tikzcd}\\
& \begin{tikzcd}  C_2^t \arrow[rightarrow]{r}  & \eta_0^{12}+\eta_{01}^2+\eta_1^2 . \end{tikzcd}\\
\end{cases}
\end{equation}
The generic fiber over $S\cap T\cap V(a_1)$ is a non-Kodaira fiber illustrated in Figure \ref{NK1.WEW}  and corresponding to an incomplete fiber of type IV$^*$.
\begin{figure}[H]
\begin{center}
\includegraphics{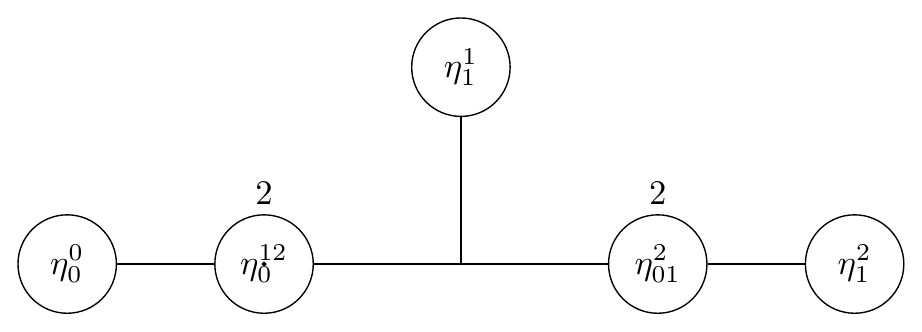}
\caption{Fiber over the generic point of the locus $S\cap T\cap V(a_1)$ in Resolution III of the I$_2^{\text{s}}$+I$_3^{\text{s}}$-model. \label{NK1.WEW}} 
\end{center}
\end{figure}

Now consider the other condition, $a_{1}\widetilde{a}_{6}=\widetilde{a}_{3}\widetilde{a}_{4}$, to get the other specialization. The curve $\eta_1^2$ splits into two fibers intersecting each other:
\begin{equation}
\text{on} \ S\cap T\cap V(a_{1}\widetilde{a}_{6}-\widetilde{a}_{3}\widetilde{a}_{4}): \ 
\begin{tikzcd}  \eta_1^2 \arrow[rightarrow]{r}  & \eta_1^{2A}+\eta_1^{2B} . \end{tikzcd}
\end{equation}
Thus, we get a fiber enhancement of type $\mathrm{I}_{6}^{\text{s}}$, which is represented in Figure \ref{I6.WEW}.
\begin{figure}[H]
\begin{center}
\includegraphics{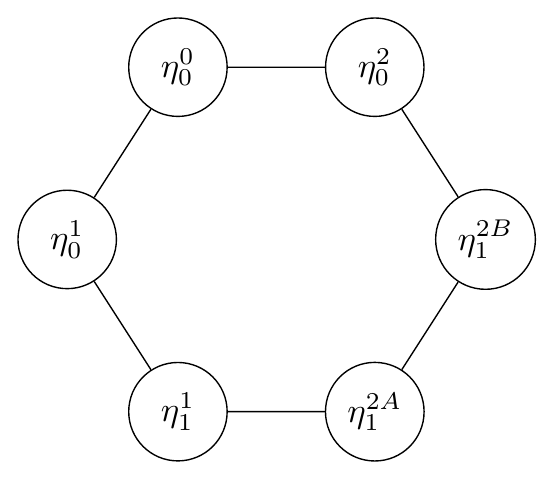}
\caption{Fiber over the generic point of the locus $S\cap T\cap V(a_{1}\widetilde{a}_{6}-\widetilde{a}_{3}\widetilde{a}_{4})$ in Resolution III of the I$_2^{\text{s}}$+I$_3^{\text{s}}$-model. \label{I6.WEW}} 
\end{center}
\end{figure}

\subsection{Resolution IV}

The resolution IV is given by the following sequence of  blowups:
\begin{equation} \begin{tikzcd}[column sep=huge] X_0 \arrow[leftarrow]{r} {\displaystyle (x,y, t|w_1)} &  X_1 \arrow[leftarrow]{r} {\displaystyle (y, w_1| w_2)} &  X_2 \arrow[leftarrow]{r} {\displaystyle (x,y, s| e_1)} &  X_3  \end{tikzcd} . \end{equation}
Its projective coordinates are then given by
\begin{equation}
[e_{1}w_{1}w_{2}x\,;\, e_{1}w_{1}w_{2}^{2}y\,;\, z=1][e_{1}x\,;\, e_{1}w_{2}y\,;\, t][e_{1}y\,;\, w_{1}][x\,;\, y\,;\, s].
\end{equation}
The proper transform of the elliptic fibration for the resolution III is the same as equation \eqref{proptransf.I2sI3s.res134} and the five fibral divisors are thus identical to the equation \eqref{5div.I2sI3s.res134}.
At the intersection of both divisors $S$ and $T$, we get the following curves:
\begin{align}
\text{On} \ S\cap T\cap V(a_1): \ 
\begin{cases}
C_{0}^{\text{s}}\cap C_{0}^t \longrightarrow & \eta_0^0: s=t=y(w_{2}y+a_{1}x)-w_{1}e_{1}x^{3}=0, \\
C_{0}^{\text{s}}\cap C_{1}^t \longrightarrow & \eta_0^1: s=w_{1}=w_{2}y+a_{1}x=0, \\
C_{0}^{\text{s}}\cap C_{2}^t \longrightarrow & \eta_{0}^{2A}: s=w_{2}=x=0, \ \eta_{0}^{2B}: s=w_{2}=a_{1}y-w_{1}e_{1}x^{2}=0, \\
C_{1}^{\text{s}}\cap C_{2}^t \longrightarrow & \eta_1^2: e_{1}=w_{2}=y(a_{1}x+\widetilde{a}_{3}st)-st^{2}w_{1}(\widetilde{a}_{4}x+\widetilde{a}_{6}st)=0.
\end{cases}
\end{align}
From the five fibral divisors, we summarize the splittings of the curves to be the following.
\begin{equation}
\text{On} \ S\cap T\cap V(a_1): \ 
\begin{cases}
& \begin{tikzcd}  C_0^{\text{s}} \arrow[rightarrow]{r}  & \eta_0^0+\eta_0^1+\eta_0^{2A}+\eta_0^{2B}  \end{tikzcd} \\
& \begin{tikzcd}  C_1^{\text{s}} \arrow[rightarrow]{r}  & \eta_1^2  \end{tikzcd}\\
& \begin{tikzcd}  C_0^t \arrow[rightarrow]{r}  & \eta_0^0  \end{tikzcd}\\
& \begin{tikzcd}  C_1^t \arrow[rightarrow]{r}  & \eta_0^1  \end{tikzcd}\\
& \begin{tikzcd}  C_2^t \arrow[rightarrow]{r}  & \eta_0^{2A}+\eta_0^{2B}+\eta_1^2 . \end{tikzcd}\\
\end{cases}
\label{eqn:I2sI3s.WWE.cd2}
\end{equation}
This corresponds to $\mathrm{I}_{5}^{\text{s}}$ as it is represented in Figure \ref{I5.WWE}.
\begin{figure}[H]
\begin{center}
\includegraphics{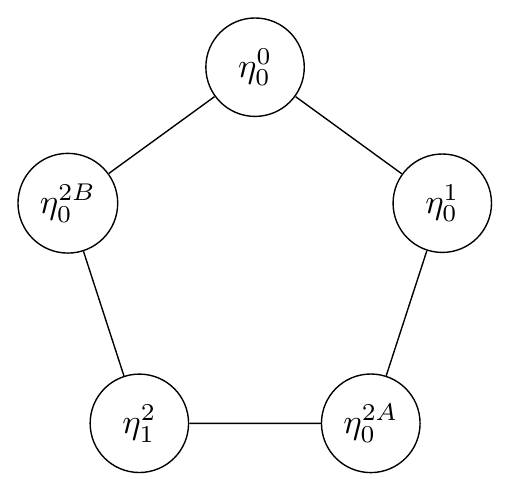}
\caption{Fiber over the generic point of the locus $S\cap T$ in Resolution IV of the I$_2^{\text{s}}$+I$_3^{\text{s}}$-model. \label{I5.WWE}} 
\end{center}
\end{figure}
The intersection numbers are computed using the equation \eqref{eqn:I2sI3s.WWE.cd2} to get the weights and the representations of the curves.
\begin{table}[htb]
\begin{center}
\begin{tabular}{|c|c|c|c|c|c|c|c|}
\hline 
  & $D^{\text{s}}_{0}$ & $D^{\text{s}}_{1}$ & $D^{\text{t}}_{0}$ & $D^{\text{t}}_{1}$ & $D^{\text{t}}_{2}$& Weight& Representation\\
\hline 
\hline 
$\eta_0^0$ & 0 & 0 & -2 & 1 & 1 & [0;-1,-1] & $\bf{(1,8)}$\\
\hline 
$\eta_0^1$ & 0 & 0 & 1 & -2 & 1 & [0;2,-1] & $\bf{(1,8)}$\\
\hline 
$\eta_{0}^{2A}$ & -1 & 1 & 0 & 1 & -1 & [-1;-1,1] & $\bf{(2,3)}$\\
\hline 
$\eta_1^2$ & 2 & -2 & 0 & 0 & 0 & [2;0,0] & $\bf{(3,1)}$\\
\hline 
$\eta_{0}^{2B}$ & -1 & 1 & 1 & 0 & -1 & [-1;0,1] & $\bf{(2,\bar{3})}$\\
\hline 
\end{tabular}
\end{center}
\caption{Weights of vertical curves and representations  in the resolution  
 IV of the I$_2^{\text{s}}+$I$_3^{\text{s}}$-model.}
\end{table}

This has two further specializations in codimension-three. The first specialization is when $a_{1}=0$. We observe the following splittings for the elliptical fibrations
\begin{align}
\text{on} \ S\cap T\cap V(a_1): \ 
\begin{cases}
\eta_0^0 \quad \longrightarrow & \eta_0^0: s=t=w_{2}y^{2}-w_{1}e_{1}x^{3}=0 , \\
\eta_0^1 \quad \longrightarrow & \eta_0^{12}: s=w_{1}=w_{2}=0 , \\
\eta_{0}^{2A} \ \longrightarrow & \eta_{0}^{2A}: s=w_{2}=x=0 , \\
\eta_{0}^{2B} \ \longrightarrow & \eta_0^{12}: s=w_{2}=w_{1}=0 , \ \eta_{01}^{2}: s=w_{2}=e_{1}=0 , \ \eta_{0}^{2A}: s=w_{2}=x=0 , \\
\eta_1^2 \quad \longrightarrow & \eta_{01}^{2}: e_{1}=w_{2}=s=0 , \ \eta_1^2: e_{1}=w_{2}=\widetilde{a}_{3}y-tw_{1}(\widetilde{a}_{4}x+\widetilde{a}_{6}st)=0 .
\end{cases}
\end{align}
For this codimension-three fiber enhancement, we get a specialization of $\mathrm{E_{6}}$, as it is represented in Figure \ref{NK1.WWE}.
\begin{figure}[H]
\begin{center}
\includegraphics{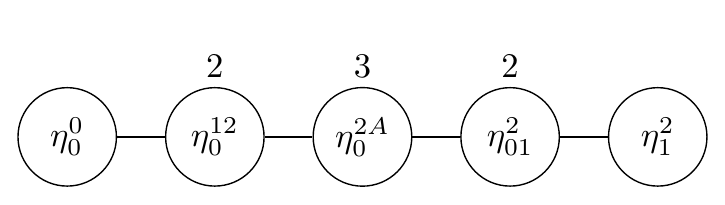}
\caption{Fiber over the generic point of the locus $S\cap T\cap V(a_1)$ in Resolution IV of the I$_2^{\text{s}}$+I$_3^{\text{s}}$-model. \label{NK1.WWE}} 
\end{center}
\end{figure}

The other specialization is when $a_{1}\widetilde{a}_{6}=\widetilde{a}_{3}\widetilde{a}_{4}$. Then all the other fibers are the same except $\eta_1^2$, which splits into two curves intersecting each other:
\begin{equation}
\text{on} \ S\cap T\cap V(a_{1}\widetilde{a}_{6}-\widetilde{a}_{3}\widetilde{a}_{4}): \ 
\begin{tikzcd}  \eta_1^2 \arrow[rightarrow]{r}  & \eta_1^{2A}+\eta_1^{2B} . \end{tikzcd}
\end{equation}
For this codimension-three enhancement, we get a fiber of type $\mathrm{I_{6}}^{\text{s}}$ as in Figure \ref{I6.WWE}.
\begin{figure}[H]
\begin{center}
\includegraphics{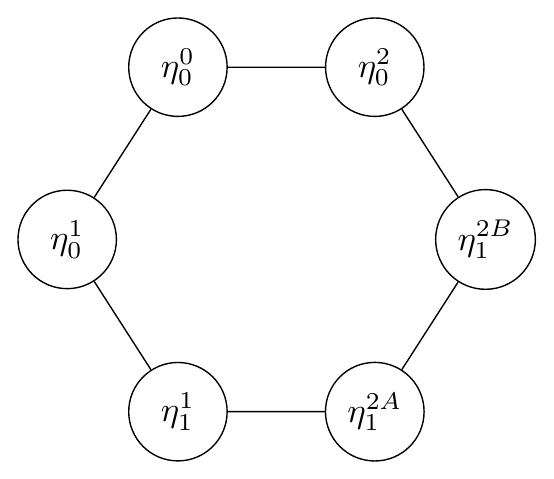}
\caption{Fiber over the generic point of the locus $S\cap T\cap V(a_{1}\widetilde{a}_{6}-\widetilde{a}_{3}\widetilde{a}_{4})$ in Resolution IV of the I$_2^{\text{s}}$+I$_3^{\text{s}}$-model. \label{I6.WWE}} 
\end{center}
\end{figure}

\subsection{Flops}\label{sec:Flops}
In this section, we discuss the flops between the resolutions I, II, III, and IV. We recall that the resolutions I$'$, II$'$, III$'$, IV$'$ are their mirrors under the birational map induced by the involution of the Mordell--Weil group. We consider the case of the I$_2^{\text{s}}+$I$_3^{\text{s}}$ model analyzed in Section \ref{sec:I2sI3s}. The other case follows the same scheme. When there is a simple flop between two resolutions, if the flopping curve has weights $\omega$ in one of the resolutions,  is replaced in the other by a curve of weights $-\omega$. Each resolution corresponds to a minimal model over the Weierstrass model. Hence, in the hyperplane arrangement, each resolution corresponds to a specific chamber. When two resolutions are connected by a flop of a curve of weight $\omega$, the hyperplane separating the corresponding chamber is exactly the hyperplane $\omega^\bot$ perpendicular to $\omega$. It follows that a chamber is uniquely defined by its possible flopping curves. 
We determine in Table \ref{flopcurves} all the flopping curves and show that their weights coincide with the hyperplanes separating two chambers of the hyperplane arrangement. These curves are identified with their corresponding weights from Table \ref{Table:Weight}. This result matches with the analysis in Figure \ref{2DChambersMatch}, which completes the correspondence between the geometry and the representation theory.

\begin{table}[H]
\begin{center}
\bgroup
\def\arraystretch{1.5}
\scalebox{.9}{
\begin{tabular}{|c c c c c c c|c|}
\hline 
\multicolumn{7}{|c|}{Flopping curves} & Weight \\
\hline
\hline
Resolution I: & $\eta_1^{0A}$ & $[1;0,-1]$ \, $\bf{(2,3)}$ & $\leftrightarrow$ & Resolution II: & $\eta_0^2$ & $[-1;0,1]$ \, $\bf{(2,\bar{3})}$ & $\omega_3^{(\bf{2,3})}$ \\
\hline 
Resolution II: & $\eta_1^1$ & $[1;-1,0]$ \, $\bf{(2,\bar{3})}$ & $\leftrightarrow$ & Resolution III: & $\eta_0^1$ & $[-1;1,0]$ \, $\bf{(2,3)}$ & $\omega_4^{(\bf{2,3})}$ \\
\hline 
Resolution III: & $\eta_1^1$ & $[1;1,-1]$ \, $\bf{(2,\bar{3})}$ & $\leftrightarrow$ & Resolution IV: & $\eta_0^{2A}$ & $[-1;-1,1]$ \, $\bf{(2,3)}$ \,& $\omega_5^{(\bf{2,3})}$ \\
\hline 
\end{tabular}
}
\egroup
\caption{The fibers that is the one that separates between the chambers and thus responsible for flops in the I$_2^{\text{s}}+$I$_3^{\text{s}}$-model. The weight of the contracted curve in a (terminal) flop connecting two resolutions is normal to the facet common to the closures of the corresponding chambers in the hyperplane arrangement.}
\label{flopcurves}
\end{center}
\end{table}

\section{The III$+$IV$^{\text{s}}$ Model \label{sec:IIIIVs}} 
In this section, we study the fiber structure of the crepant resolutions of the $\text{III}+\text{IV}^{\text{s}}$-model defined by the following Weierstrass equation:
\begin{equation}
Y_0: y^{2}+\widetilde{a}_{1}stxy+\widetilde{a}_{3}sty=x^{3}+\widetilde{a}_{2}stx^{2}+\widetilde{a}_{4}s t^{2}x+\widetilde{a}_{6}s^{2}t^{3},
\label{eq:IIIIVs}
\end{equation}
where the fiber III  and IV$^{\text{s}}$ are the fibers over the generic point of $S=V(s)$ and $T=V(t)$ respectively.
It corresponds to the low-right corner of Figure \ref{Pic:AllModels}. In particular, the fiber over $S$ and $T$ cannot specialize further while preserving their dual graphs (and hence, the  gauge group \susu).
  In \cite{Grassi:2014zxa}, this model was explored using the point of view of string junctions. 

 Here, we analyze the geometry of the crepant resolutions of the $\text{III}+\text{IV}^{\text{s}}$-model. The triple intersection numbers are the same as those of the I$_2^{\text{s}} +$I$_3^{\text{s}}$-model. 
The fiber over the generic point of $S\cap T$  is a non-Kodaira fiber corresponding to a fiber of type IV$^*$ with some nodes contracted. Such a fiber enhances further over $S\cap T\cap V(\widetilde{a}_3)$ to a non-Kodaira fiber corresponding to a fiber of type III$^*$ with some nodes contracted. The non-Kodaira fibers observed for the III$+$IV$^{\text{s}}$-model were already seen in the I$_2^{\text{s}}$+I$_3^{\text{s}}$-model but one codimension higher. 
  
\subsection{Resolution I}\label{Sec:IIIIVSResI}
The resolution I is defined by the following sequence of blowups:
\begin{equation} 
\begin{tikzcd}[column sep=huge] 
X_0 \arrow[leftarrow]{r} {\displaystyle (x,y, s|e_1)} &  X_1 \arrow[leftarrow]{r} {\displaystyle (x,y, t| w_1)} &  X_2 \arrow[leftarrow]{r} {\displaystyle (y, w_1| w_2)} &  X_3  
\end{tikzcd}  
\end{equation}
The proper transform of the III$+$IV$^{\text{s}}$-model is
\begin{equation}
Y: y(w_{2}y+\widetilde{a}_1se_{1}tw_{1}w_{2}x+\widetilde{a}_3st)=w_{1}(e_{1}x^{3}+\widetilde{a}_2se_{1}tx^{2}+\widetilde{a}_4st^{2}x+\widetilde{a}_6s^{2}t^{3}).
\label{eq:PTofIIIIVs}
\end{equation}
The projective coordinates are then given by
\begin{equation}
[e_{1}w_{1}w_{2}x\,;\, e_{1}w_{1}w_{2}^{2}y\,;\, z=1][w_{1}w_{2}x\,;\, w_{1}w_{2}^{2}y\,;\, s][x\,;\, w_{2}y\,;\, t][y\,;\, w_{1}].
\end{equation}
The fibral divisors are
\begin{align}
&\text{I}_2^{\text{s}}: \ 
\begin{cases}
D_{0}^{\text{s}} :&  s=w_{2}y^{2}-w_{1}e_{1}x^{3}=0, \\
D_{1}^{\text{s}} :& e_{1}=y(w_{2}y+\widetilde{a}_3st)-st^{2}w_{1}(\widetilde{a}_4x+\widetilde{a}_6st)=0, \\
\end{cases}\\
&\text{I}_3^{\text{s}}: \ 
\begin{cases}
D_{0}^t :& t=w_{2}y^{2}-w_{1}e_{1}x^{3}=0, \\
D_{1}^t :& w_{1}=w_{2}y+\widetilde{a}_3st=0, \\
D_{2}^t :& w_{2}=\widetilde{a}_3sty-w_{1}(e_{1}x^{3}+\widetilde{a}_2se_{1}tx^{2}+\widetilde{a}_4st^{2}x+\widetilde{a}_6s^{2}t^{3})=0.
\end{cases}
\end{align}
Over the generic point of  the intersection of $S$ and $T$, we get the following irreducible curves
\begin{align}
\text{On} \ S\cap T: \ 
\begin{cases}
C_{0}^s\cap C_{0}^{t} \quad \rightarrow & \eta_0^0: s=t=w_{2}y^{2}-w_{1}e_{1}x^{3}=0 ,\\
C_{1}^s\cap C_{0}^{t} \quad \rightarrow & \eta_1^{02}: e_{1}=t=w_{2}=0,\ \eta_1^{0A}: e_{1}=t=y=0,\\
C_{1}^s\cap C_{1}^{t} \quad \rightarrow & \eta_1^1: e_{1}=w_{1}=w_{2}y+\widetilde{a}_3st=0,\\
C_{1}^s\cap C_{2}^{t} \quad \rightarrow & \eta_1^{02}: e_{1}=w_{2}=t=0,\ \eta_1^2: e_{1}=w_{2}=\widetilde{a}_3y-tw_{1}(\widetilde{a}_4x+\widetilde{a}_6st)=0 .
\end{cases}
\end{align}
The fiber over the generic point of  $S\cap T$ has a structure given by Figure \ref{fig:IIIIVs.Res1.cd2}, and corresponds to a fiber of type  IV$^*$ with contracted nodes. 
At the collision $S\cap T$, the components of the fibers III and IV$^{\text{s}}$ degenerate as follows.
\begin{equation}
\text{On} \ S\cap T: \ 
\begin{cases}
& \begin{tikzcd}  C_0^{\text{s}} \arrow[rightarrow]{r}  & \eta_0^0,  \end{tikzcd} \\
& \begin{tikzcd}  C_1^{\text{s}} \arrow[rightarrow]{r}  & 2\eta_1^{02}+2\eta_1^{0A}+\eta_1^1+\eta_1^2,  \end{tikzcd}\\
& \begin{tikzcd}  C_0^t \arrow[rightarrow]{r}  & \eta_0^0+\eta_1^{02}+2\eta_1^{0A},  \end{tikzcd}\\
& \begin{tikzcd}  C_1^t \arrow[rightarrow]{r}  & \eta_1^1,  \end{tikzcd}\\
& \begin{tikzcd}  C_2^t \arrow[rightarrow]{r}  & \eta_1^{02}+\eta_1^2.  \end{tikzcd}\\
\end{cases}
\end{equation}

\begin{figure}[H]
\begin{center}
\includegraphics{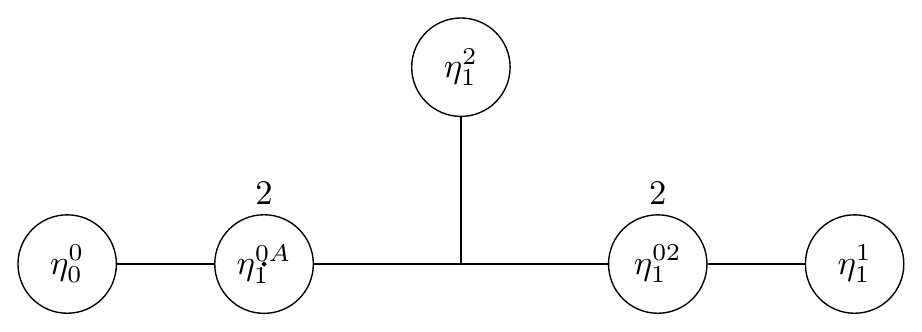}
\caption{Fiber over the generic point of the locus $S\cap T$ in Resolution I of the  III+IV$^{\text{s}}$-model. \label{fig:IIIIVs.Res1.cd2}} 
\end{center}
\end{figure}
We observe that this is identical to the fiber in the resolution I of the $\mathrm{I_{2}^{\text{s}}+I_{3}^{\text{s}}}$-model in codimension-three over $S\cap T\cap V(a_1)$. 

In order to get the weights of the curves, the intersection numbers are computed between the codimension-two curves and the fibral divisors.
\begin{table}[H]
\begin{center}
\begin{tabular}{|c|c|c|c|c|c|c|c|}
\hline 
  & $D^{\text{s}}_{0}$ & $D^{\text{s}}_{1}$ & $D^{\text{t}}_{0}$ & $D^{\text{t}}_{1}$ & $D^{\text{t}}_{2}$& Weight& Representation\\
\hline 
\hline 
$\eta_0^0$ & -2 & 2 & 0 & 0 & 0 & [-2;0,0] & $\bf{(3,1)}$\\
\hline 
$\eta_1^{0A}$ & 1 & -1 & -1 & 0 & 1 & [1;0,-1] & $\bf{(2,3)}$\\
\hline 
$\eta_1^2$ & 0 & 0 & 1 & 0 & -1 & [0;0,1] & $\bf{(1,\bar{3})}$\\
\hline 
$\eta_1^{02}$ & 0 & 0 & 0 & 1 & -1 & [0;-1,1] & $\bf{(1,3)}$\\
\hline 
$\eta_1^1$ & 0 & 0 & 1 & -2 & 1 & [0;2,-1] & $\bf{(1,8)}$\\
\hline 
\end{tabular}
\end{center}
\caption{Weights of vertical curves and representations in the resolution I of the III$+$IV$^{\text{s}}$-model.}
\end{table}

We note that the sum of the two curves $\eta_1^{02}+\eta_1^{0A}$ produce the weight $[1;-1,0]$ of the representation $(\mathbf{2,\bar{3}})$. In  the resolution I of the $\mathrm{I_{2}^{\text{s}}+I_{3}^{\text{s}}}$-model, the weight 
$[1;-1,0]$ corresponds to $\eta_1^{0B}$ in codimension-two, which splits into the two curves in codimension-three with the same weights as $\eta_1^{02}$ and $\eta_1^{0A}$.

The fiber over the generic point of $S\cap T$  shown on  Figure \ref{fig:IIIIVs.Res1.cd2}  specializes further  when $\widetilde{a}_{3}=0$:
\begin{equation}
\text{on} \ S\cap T\cap V(\widetilde{a}_{3}): \ 
\begin{cases}
& \begin{tikzcd}  \eta_1^1 \arrow[rightarrow]{r}  & \eta_1^{12},  \end{tikzcd} \\
& \begin{tikzcd}  \eta_1^2 \arrow[rightarrow]{r}  & \eta_1^{02}+\eta_1^{12}+\eta_1^{2},  \end{tikzcd}\\
\end{cases}
\end{equation}
where the new curves are given by
\begin{align}
\begin{cases}
\eta_1^{12}: \ & e_{1}=w_{1}=w_{2}=0, \\
\eta_1^{02}: \ & e_{1}=w_{2}=t=0, \\
\eta_1^2: \ & e_{1}=w_{2}=\widetilde{a}_4x+\widetilde{a}_6st=0.
\end{cases}
\end{align}
The fiber over the generic point $\ S\cap T\cap V(\widetilde{a}_{3})$ is illustrated in Figure \ref{fig:IIIIVs.Res1.cd3}, and  corresponds to a fiber of type  III$^*$ with contracted nodes.
\begin{figure}[H]
\begin{center}
\includegraphics{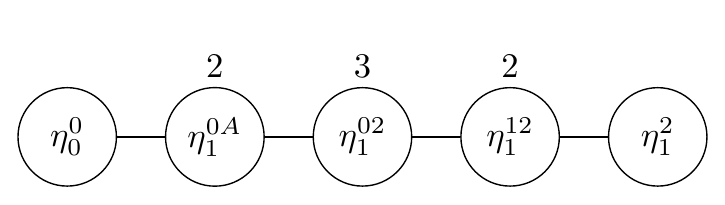}
\caption{Fiber over the generic point of the locus $S\cap T\cap V(\widetilde{a}_3)$ in Resolution I of the  III+IV$^{\text{s}}$-model. \label{fig:IIIIVs.Res1.cd3}} 
\end{center}
\end{figure}

\subsection{Resolution II}\label{Sec:IIIIVSResII}
In this section, we study the resolution II in detail. As the resolution II requires making a first blowup around a singular center, it is useful to rewrite equation \eqref{eq:IIIIVs} as
\begin{align}
Y_0: \ 
\begin{cases}
y(y+\widetilde{a}_1p_{0}x+\widetilde{a}_3p_{0})= x^{3}+\widetilde{a}_2p_{0}x^{2} +\widetilde{a}_4p_{0}t x+\widetilde{a}_6p_{0}^{2}t \\
p_0=st
\end{cases}.
\end{align}
The resolution II is then given by the following sequence of blowups
\begin{equation} 
\begin{tikzcd}[column sep=huge] 
X_0 \arrow[leftarrow]{r} {\displaystyle (x,y, p_0|p_1)} &  X_1 \arrow[leftarrow]{r} {\displaystyle (y,t,p_1| w_1)} &  X_2 \arrow[leftarrow]{r} {\displaystyle (t, p_0| w_2)} &  X_3  
\end{tikzcd} . 
\end{equation}
Where $X_0=\mathbb{P}[\mathscr{O}_B\oplus\mathscr{L}^{\otimes 2}\oplus \mathscr{L}^{\otimes 3}]$. The projective coordinates are then 
\begin{equation}
[p_{1}w_1x\,:\, p_{1}w_1^{2}y\,:\, z=1][x\,:\, w_1y\,:\, p_0 w_2][y\,:\, tw_2\,:\, p_1][t\,:\, p_0],
\end{equation}
and the proper transform is 
\begin{align}
Y: \ 
\begin{cases}
y(w_{1}y+\widetilde{a}_1p_{0}w_{2}x+\widetilde{a}_3p_{0}w_{2})= p_{1}x^{3}+\widetilde{a}_2p_{0}p_{1}w_{2}x^{2} +\widetilde{a}_4p_{0}tw_2^{2}x+\widetilde{a}_6p_{0}^{2}tw_{2}^{3}\\
p_0p_1=st
\end{cases}.
\end{align}
The variety $X_1=Bl_{(x,y,p_0)} X_0$ has double point singularities on $p_0=p_1= s=t=0$. The fibral divisors are:
\begin{align}
&\text{III}: \ 
\begin{cases}
 D^{\text{s}}_0: & s=p_0=w_1y^2-p_{1}x^{3}=0 \\
 D^{\text{s}}_1: & s=p_{1}=y(w_1y+\widetilde{a}_1p_0w_2x+\widetilde{a}_{3}p_0w_2)-(p_1x^3+\widetilde{a}_{4}p_0tw_2^2x+\widetilde{a}_{6}p_0^{2}t w_2^3)=0 \\
\end{cases} \\
&\text{IV}^{\text{s}}: \
\begin{cases}
 D^{\text{t}}_0: & w_2=p_0p_1-st=w_1y^2-p_1x^3=0 \\
 D^{\text{t}}_1: & t=p_1=w_1y+\widetilde{a}_1p_0w_2x+\widetilde{a}_3p_0w_2=0 \\
 D^{\text{t}}_2: & w_1=p_0p_1-st=y(\widetilde{a}_1p_0w_2x+\widetilde{a}_{3}p_0w_2)-(p_{1}x^{3}+\widetilde{a}_{2}p_0p_{1}w_2x^{2}+\widetilde{a}_{4}p_0tw_2^{2}x+\widetilde{a}_{6}p_0^{2}tw_2^{3})=0
\end{cases}
\end{align}

At the intersection of $S$ and $T$, the fiber enhances to a non-Kodaira fiber presented in Figure \ref{fig:IIIIVs.Res2.cd2}. This is realized by the following splitting of curves. 
\begin{equation}
\text{On} \ S\cap T: \ 
\begin{cases}
& \begin{tikzcd}  C^{\text{s}}_0 \arrow[rightarrow]{r}  & \eta_0^0+\eta_{01}^2  \end{tikzcd} \\
& \begin{tikzcd}  C^{\text{s}}_1 \arrow[rightarrow]{r}  & \eta_{01}^2+2\eta_1^{02}+\eta_1^1+\eta_1^2  \end{tikzcd}\\
& \begin{tikzcd}  C^{\text{t}}_0 \arrow[rightarrow]{r}  & \eta_0^0+\eta_1^{02}  \end{tikzcd}\\
& \begin{tikzcd}  C^{\text{t}}_1 \arrow[rightarrow]{r}  & \eta_1^1  \end{tikzcd}\\
& \begin{tikzcd}  C^{\text{t}}_2 \arrow[rightarrow]{r}  & 2\eta_{01}^2+\eta_1^{02}+\eta_1^2  \end{tikzcd}\\
\end{cases}
\end{equation}
The curves at the intersection are given by
\begin{align}
\text{On} \ S\cap T: \ 
\begin{cases}
C^{\text{s}}_{0}\cap C^{\text{t}}_{1} \quad\longrightarrow & \eta_0^0:  s=p_{0}=w_2=w_1y^2-p_{1}x^{3}=0 \\
C^{\text{s}}_{0}\cap C^{\text{t}}_{2} \quad\longrightarrow & \eta_{0}^{2}:  s=p_{0}=w_1=p_1=0 \\
C^{\text{s}}_{1}\cap C^{\text{t}}_{0} \quad\longrightarrow & \eta_1^{02}:  s=p_{1}=w_2=w_1=0 \\
C^{\text{s}}_{1}\cap C^{\text{t}}_{1} \quad\longrightarrow & \eta_{1}^{1}:  s=p_{1}=t=w_1y+\widetilde{a}_{3}p_{0}w_2=0 \\
C^{\text{s}}_{1}\cap C^{\text{t}}_{2} \quad\longrightarrow & \eta_{1}^{02}: s=p_{1}=w_1=w_2=0, \ \eta_{01}^2: s=p_{1}=w_1=p_{0}=0,\\
	& \eta_1^2: s=p_{1}=w_1=\widetilde{a}_{3}y-\widetilde{a}_{4}tw_2x+\widetilde{a}_{6}p_{0}tw_2^2=0 \\
\end{cases}
\end{align}
Note that we had the same fiber earlier in the resolution I of the $\mathrm{I_{2}^{\text{s}}+I_{3}^{\text{s}}}$-model in codimension-three with a condition $a_1=0$.

\begin{table}[htb]
\begin{center}
\begin{tabular}{|c|c|c|c|c|c|c|c|}
\hline 
& $D^{\text{s}}_{0}$ & $D^{\text{s}}_{1}$ & $D^{\text{t}}_{0}$ & $D^{\text{t}}_{1}$ & $D^{\text{t}}_{2}$& Weight& Representation\\
\hline 
\hline 
$\eta_0^0$ & -1 & 1 & -1 & 0 & 1 & [-1;0,-1] & $\bf{(2,3)}$\\
\hline 
$\eta_{01}^2$ & -1 & 1 & 1 & 0 & -1 & [-1;0,1] & $\bf{(2,\bar{3})}$\\
\hline 
$\eta_1^1$ & 0  & 0 & 1 & -2 & 1 & [0;2,-1] &   $\bf{(1,8)}$\\
\hline 
$\eta_1^{02}$ & 1 & -1 & -1 & 1 & 0 & [1;-1,0] &   $\bf{(2,\bar{3})}$\\
\hline 
$\eta_1^2$ & 1 & -1 & 0 & 0 & 0 & [1;0,0] & $\bf{(2,1)}$\\
\hline 
\end{tabular}
\end{center}
\caption{
Weights of vertical curves and representations  in the resolution 
 II of the III$+$IV$^{\text{s}}$-model.}
\end{table}

The chain  $\eta_1^{02}+\eta_{01}^{2}+\eta_1^2$ produces the weight $[1;-1,1]$ of the representation $(\mathbf{2,3})$. In the case of the resolution II of the $\mathrm{I_{2}^{\text{s}}+I_{3}^{\text{s}}}$-model, the sum of the curves corresponds to $\eta_1^2$ in codimension-two, which splits into the three curves in codimension-three.

\begin{figure}[H]
\centering
\includegraphics{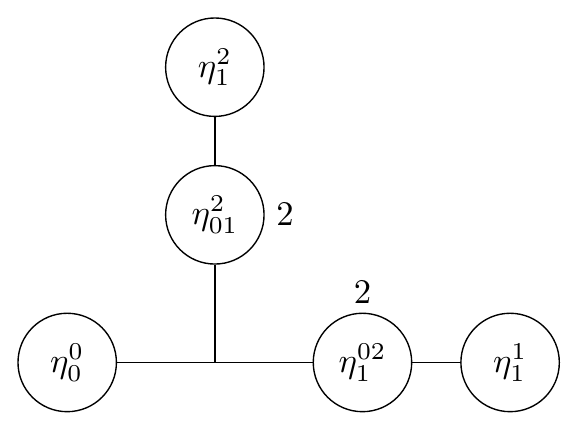}
\caption{Fiber over the generic point of the locus $S\cap T$ in Resolution II of the III+IV$^{\text{s}}$-model.  \label{fig:IIIIVs.Res2.cd2}}
\end{figure}
There is an enhancement when $\widetilde{a}_3$ as the  two curves $\eta_1^1$ and $\eta_1^2$ degenerate as follows:
\begin{equation}
\text{on} \ S\cap T\cap V(\widetilde{a}_{3}): \ 
\begin{cases}
& \begin{tikzcd}  \eta_1^1 \arrow[rightarrow]{r}  & \eta_1^{12} ,  \end{tikzcd} \\
& \begin{tikzcd}  \eta_1^2 \arrow[rightarrow]{r}  & \eta_1^{02}+\eta_1^{12}+\eta_1^{2} , \end{tikzcd}\\
 \end{cases}
 \end{equation}
 where the new curves are given by
 \begin{equation}
\begin{cases}
\eta_1^{12}: & s=p_1=t=w_1=0,\\
\eta_1^{02}: & s=p_1=w_1=w_2=0, \\
\eta_1^{2}: & s=p_1=w_1=\widetilde{a}_4 x-\widetilde{a}_6p_0 w_2=0.
\end{cases}
\end{equation}
For this codimension-three enhancement, we get a non-Kodaira fiber corresponding to a non-Kodaira fiber corresponding to a contracted fiber of type IV$^*$ as in Figure \ref{fig:IIIIVs.Res2.cd3}.
\begin{figure}[H]
\centering
\includegraphics{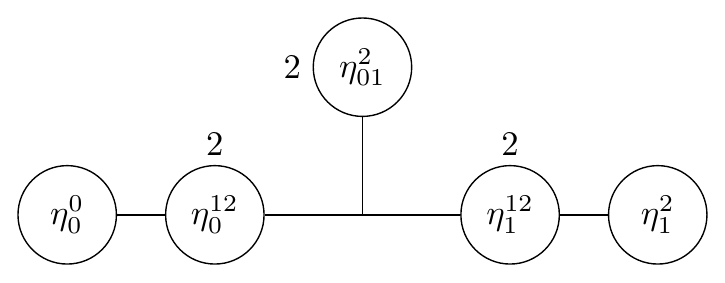}
\caption{Fiber over the generic point of the locus $S\cap T\cap V(\widetilde{a}_3)$ in Resolution II of  the III+IV$^{\text{s}}$-model. \label{fig:IIIIVs.Res2.cd3}} 
\end{figure}

\subsection{Resolution III}\label{Sec:IIIIVSResIII}
 The Resolution III is defined by the following sequence of  blowups:
\begin{equation} 
\begin{tikzcd}[column sep=huge] 
X_0 \arrow[leftarrow]{r} {\displaystyle (x,y, t|w_1)} &  X_1 \arrow[leftarrow]{r} {\displaystyle (x,y, s| e_1)} &  X_2 \arrow[leftarrow]{r} {\displaystyle (y, w_1| w_2)} &  X_3 
\end{tikzcd}
\end{equation}
The projective coordinates are then given by
\begin{equation}
[e_{1}w_{1}w_{2}x\,;\, e_{1}w_{1}w_{2}^{2}y\,;\, z=1][e_{1}x\,;\, e_{1}w_{2}y\,;\, t][x\,;\, w_{2}y\,;\, s][y\,;\, w_{1}].
\end{equation}
The proper transform is identical to equation \eqref{eq:PTofIIIIVs}.

On the intersection of $S$ and $T$, we see the following curves:
\begin{align}
\text{On} \ S\cap T: \ 
\begin{cases}
C_{0}^{\text{s}}\cap C_{0}^{\text{t}}& \longrightarrow  \eta_0^0: s=t=w_{2}y^{2}-w_{1}e_{1}x^{3}=0, \\
C_{0}^{\text{s}}\cap C_{1}^{\text{t}}& \longrightarrow  \eta_0^{12}: s=w_{1}=w_{2}=0 , \\
C_{0}^{\text{s}}\cap C_{2}^{\text{t}}& \longrightarrow  \eta_{01}^2: s=w_{2}=e_{1}=0,\ \eta_0^{12}: s=w_{2}=w_{1}=0, \\
C_{1}^{\text{s}}\cap C_{1}^{\text{t}}& \longrightarrow  \eta_1^1:  e_{1}=w_{1}=w_{2}y+\widetilde{a}_3st=0 , \\
C_{1}^{\text{s}}\cap C_{2}^{\text{t}}& \longrightarrow  \eta_{01}^2: e_{1}=w_{2}=s=0, \eta_1^2: e_{1}=w_{2}=\widetilde{a}_3y-tw_{1}(\widetilde{a}_4x+\widetilde{a}_6st)=0. \\
\end{cases}
\end{align}
Hence, we can deduce that the five fibral divisors split in the following way to produce the fiber in codimension-two, which is presented in Figure \ref{fig:IIIIVs.Res3.cd2}. 
\begin{equation}
\text{On} \ S\cap T: \ 
\begin{cases}
& \begin{tikzcd}  C_0^{\text{s}} \arrow[rightarrow]{r}  & \eta_0^0+2\eta_0^{12}+\eta_{01}^2 , \end{tikzcd} \\
& \begin{tikzcd}  C_1^{\text{s}} \arrow[rightarrow]{r}  & \eta_{01}^2+\eta_1^1+\eta_1^2 , \end{tikzcd}\\
& \begin{tikzcd}  C_0^t \arrow[rightarrow]{r}  & \eta_0^0,  \end{tikzcd}\\
& \begin{tikzcd}  C_1^t \arrow[rightarrow]{r}  & \eta_0^{12}+\eta_1^1,  \end{tikzcd}\\
& \begin{tikzcd}  C_2^t \arrow[rightarrow]{r}  & \eta_0^{12}+2\eta_{01}^2+\eta_1^2 . \end{tikzcd}\\
\end{cases}
\end{equation}

We observe that we had the same fiber earlier in the resolution III of the $\mathrm{I_{2}^{\text{s}}+I_{3}^{\text{s}}}$-model in codimension-three with a condition $a_1=0$.

In order to get the weights of the curves, the intersection numbers are computed between the codimension-two curves and the fibral divisors.
\begin{table}[H]
\begin{center}
\begin{tabular}{|c|c|c|c|c|c|c|c|}
\hline 
& $D^{\text{s}}_{0}$ & $D^{\text{s}}_{1}$ & $D^{\text{t}}_{0}$ & $D^{\text{t}}_{1}$ & $D^{\text{t}}_{2}$& Weight& Representation\\
\hline 
\hline 
$\eta_0^0$ & 0 & 0 & -2 & 1 & 1 & [0;-1,-1] & $\bf{(1,8)}$\\
\hline 
$\eta_0^{12}$ & -1 & 1 & 1 & -1 & 0 & [-1;1,0] & $\bf{(2,3)}$\\
\hline 
$\eta_1^1$ & 1 & -1 & 0 & -1 & 1 & [1;1,-1] & $\bf{(2,\bar{3})}$\\
\hline 
$\eta_{01}^2$ & 0 & 0 & 0 & 1 & -1 & [0;-1,1] & $\bf{(1,3)}$\\
\hline 
$\eta_1^2$ & 1 & -1 & 0 & 0 & 0 & [1;0,0] & $\bf{(2,1)}$\\
\hline 
\end{tabular}
\end{center}
\caption{Weights of vertical curves and representations  in the resolution 
 III of the III$+$IV$^{\text{s}}$-model at the III$+$IV$^{\text{s}}$ collision.}
\end{table}

\begin{figure}[H]
\begin{center}
\includegraphics{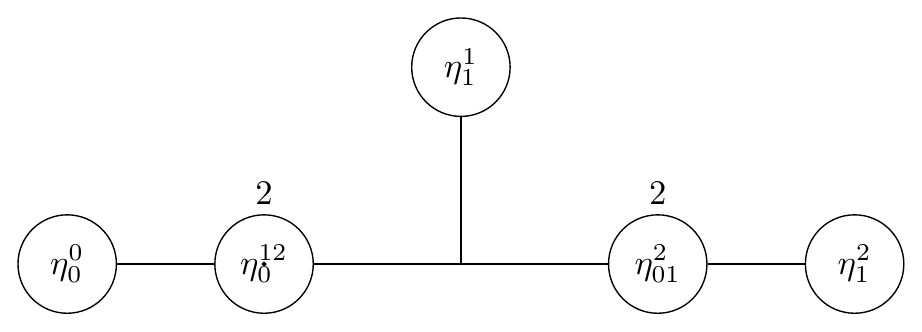}
\caption{Fiber over the generic point of the locus $S\cap T$ in Resolution III of  the III+IV$^{\text{s}}$-model.   
 \label{fig:IIIIVs.Res3.cd2}}
\end{center}
\end{figure}

The sum of the two curves $\eta_0^{12}+\eta_{01}^{2}$ produce the weight $[-1;0,1]$ of the representation $(\mathbf{2,\bar{3}})$. Moreover, the sum of $\eta_{01}^{2}+\eta_1^2$ produce the weight $[1;-1,1]$, which corresponds to a representation $(\mathbf{2,3})$. In the case of the resolution III of the $\mathrm{I_{2}^{\text{s}}+I_{3}^{\text{s}}}$-model, the former sum corresponds to $\eta_0^2$ and the latter sum corresponds to $\eta_1^2$ in codimension-two.

Over $S\cap T\cap V(\widetilde{a}_3)$ the curve $\eta_1^2$ splits as:
\begin{equation}
\text{on} \ S\cap T\cap V(\widetilde{a}_{3}): \ 
\begin{cases}
\begin{aligned}
& \begin{tikzcd}  \eta_1^1 \arrow[rightarrow]{r}  & \eta_1^{12}, \end{tikzcd} \\
& \begin{tikzcd}  \eta_1^2 \arrow[rightarrow]{r}  & \eta_1^{12}+\eta_1^{2}, \end{tikzcd}\\
\end{aligned}
\end{cases}
\end{equation}
where the new curves are given by
\begin{align}
\begin{cases}
\eta_1^{12}: & e_{1}=w_{1}=w_{2}=0, \\
\eta_1^{2}: & e_{1}=w_{2}=\widetilde{a}_4x+\widetilde{a}_6st=0.
\end{cases}
\end{align}
This corresponds to the codimension-three enhancement in Figure \ref{fig:IIIIVs.Res3.cd3}.
\begin{figure}[H]
\begin{center}
\includegraphics{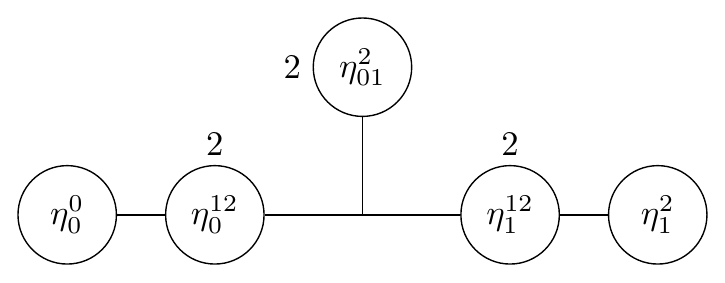}
\caption{Fiber over the generic point of the locus $S\cap T\cap V(\widetilde{a}_3)$ in Resolution III of  the III+IV$^{\text{s}}$-model.   
}
\label{fig:IIIIVs.Res3.cd3}
\end{center}
\end{figure}

\subsection{Resolution IV}\label{Sec:IIIIVSResIV}
The Resolution IV is defined by the following sequence of blowups 
\begin{equation} 
\begin{tikzcd}[column sep=huge] 
X_0 \arrow[leftarrow]{r} {\displaystyle (x,y, t|w_1)} &  X_1 \arrow[leftarrow]{r} {\displaystyle (y, w_1| w_2)} &  X_2 \arrow[leftarrow]{r} {\displaystyle (x,y, s| e_1)} &  X_3  .
\end{tikzcd}  
\end{equation}
Its projective coordinates are then given by
\begin{equation}
[e_{1}w_{1}w_{2}x\,;\, e_{1}w_{1}w_{2}^{2}y\,;\, z=1][e_{1}x\,;\, e_{1}w_{2}y\,;\, t][e_{1}y\,;\, w_{1}][x\,;\, y\,;\, s].
\end{equation}
The proper transform is identical to equation \eqref{eq:PTofIIIIVs}.

Over the generic point of  the intersection of $S$ and $T$, we see the following irreducible vertical curves.
\begin{align}
\text{On}\ S\cap T: \ 
\begin{cases}
C_0^{\text{s}}\cap C_0^t \quad\longrightarrow & \eta_0^0:  s=t=w_{2}y^{2}-w_{1}e_{1}x^{3}=0, \\
C_0^{\text{s}}\cap C_1^t \quad\longrightarrow & \eta_0^{12}:  s=w_{1}=w_{2}=0, \\
C_0^{\text{s}}\cap C_2^t \quad\longrightarrow & \eta_{0}^{2A}:  s=w_{2}=x=0, \ \eta_{01}^{2}: s=w_{2}=e_{1}=0, \ \eta_0^{12}: s=w_{2}=w_{1}=0 \\
C_1^{\text{s}}\cap C_2^t \quad\longrightarrow & \eta_{01}^{2}: e_{1}=w_{2}=s=0, \eta_1^2:  e_{1}=w_{2}=\widetilde{a}_{3}y-tw_{1}(\widetilde{a}_{4}x+\widetilde{a}_{6}st)=0.
\end{cases}
\end{align} 
The collision gives the following splitting of curves from the fiber III and IV$^{\text{s}}$ resulting in the fiber illustrated in Figure \ref{fig:IIIIVs.Res4.cd2}:
\begin{equation}
\text{On}\ S\cap T: \ 
\begin{cases}
& \begin{tikzcd}  C_0^{\text{s}} \arrow[rightarrow]{r}  & \eta_0^0+2\eta_0^{12}+3\eta_0^{2A}+\eta_{01}^2,  \end{tikzcd} \\
& \begin{tikzcd}  C_1^{\text{s}} \arrow[rightarrow]{r}  & \eta_{01}^2+\eta_1^2,  \end{tikzcd}\\
& \begin{tikzcd}  C_0^t \arrow[rightarrow]{r}  & \eta_0^0,  \end{tikzcd}\\
& \begin{tikzcd}  C_1^t \arrow[rightarrow]{r}  & \eta_0^{12},  \end{tikzcd}\\
& \begin{tikzcd}  C_2^t \arrow[rightarrow]{r}  & \eta_0^{12}+3\eta_0^{2A}+2\eta_{01}^2+\eta_1^2.  \end{tikzcd}\\
\end{cases}
\end{equation}
\begin{figure}[H]
\begin{center}
\includegraphics{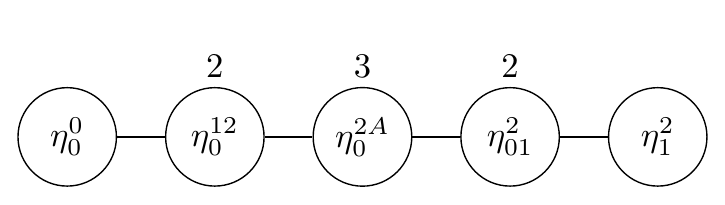}
\caption{Fiber over the generic point of the locus $S\cap T$ in Resolution IV of  the III+IV$^{\text{s}}$-model.     
 \label{fig:IIIIVs.Res4.cd2}}
\end{center}
\end{figure}
Note that we had the same fiber earlier in the resolution III of the $\mathrm{I_{2}^{\text{s}}+I_{3}^{\text{s}}}$-model in codimension-three with a condition $a_1=0$.

The weights of these vertical curves and the corresponding representations are collected in Table \ref{Table:Weights.Res.IV.III+IV}.
In order to get the weights of the curves, the intersection numbers are computed between the codimension-two curves and the fibral divisors.
\begin{table}[htb]
\begin{center}
\begin{tabular}{|c|c|c|c|c|c|c|c|}
\hline 
  & $D^{\text{s}}_{0}$ & $D^{\text{s}}_{1}$ & $D^{\text{t}}_{0}$ & $D^{\text{t}}_{1}$ & $D^{\text{t}}_{2}$& Weight& Representation\\
\hline 
\hline 
$\eta_0^0$ & 0 & 0 & -2 & 1 & 1 & [0;-1,-1] & $\bf{(1,8)}$\\
\hline 
$\eta_0^{12}$ & 0 & 0 & 1 & -2 & 1 & [0;2,-1] & $\bf{(1,8)}$\\
\hline 
$\eta_{0}^{2A}$ & -1 & 1 & 0 & 1 & -1 & [-1;-1,1] & $\bf{(2,3)}$\\
\hline 
$\eta_{01}^{2}$ & 1 & -1 & 0 & 0 & 0 & [1;0,0] & $\bf{(2,1)}$\\
\hline 
$\eta_1^2$ & 1 & -1 & 0 & 0 & 0  & [1;0,0] & $\bf{(2,1)}$\\
\hline 
\end{tabular}
\end{center}
\caption{Weights of vertical curves and representations  in the resolution 
 IV of the III+IV$^{\text{s}}$.\label{Table:Weights.Res.IV.III+IV}}
\end{table}

The sum of the three curves $\eta_1^{02}+\eta_{01}^{2}+\eta_1^2$ produce the weight $[1;-1,1]$, which yields a representation $(\mathbf{2,3})$. In the case of the resolution IV of the $\mathrm{I_{2}^{\text{s}}+I_{3}^{\text{s}}}$-model, the sum of the curves corresponds to $\eta_1^2$ in codimension-two, which splits into the three curves in codimension-three.

The fiber specialize further over $S\cap T\cap V(\widetilde{a}_3)$ as the curve  $\eta_1^2$ becomes
\begin{equation}
\text{on}\ S\cap T\cap V(\widetilde{a}_{3}): \ 
\eta_1^2 \rightarrow \eta_1^{2}: \, e_{1}=w_{2}=\widetilde{a}_{4}x+\widetilde{a}_{6}st=0.
\end{equation}
What is now different from other points of $S\cap T$ is that the curve $\eta_1^2$ now intersect both $\eta_{01}^{2}$ and $\eta_{0}^{2A}$ at the same point  resulting in a different fiber structure  represented in Figure \ref{fig:IIIIVs.Res4.cd3}.

\begin{figure}[H]
\begin{center}
\includegraphics{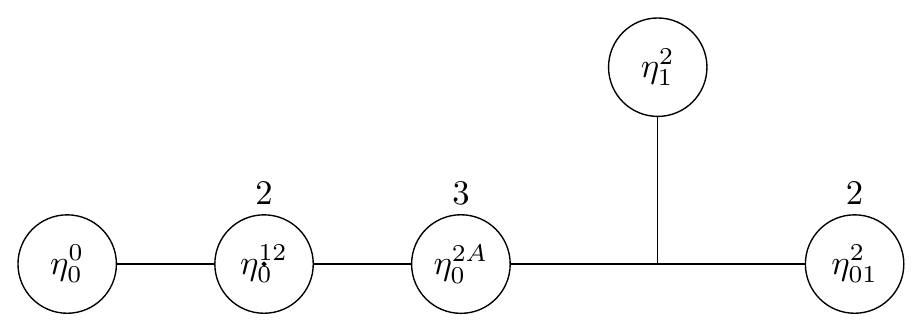}
\caption{Fiber over the generic point of the locus $S\cap T\cap V(\widetilde{a}_3)$ in Resolution IV of  the III+IV$^{\text{s}}$-model.   
 \label{fig:IIIIVs.Res4.cd3}}
\end{center}
\end{figure}

\section{$5d$ and $6d$ supergravity theories with eight supercharges}\label{sec:Physics}
In this section, we discuss the  five and six-dimensional theories with eight supercharges via Calabi--Yau compactification of M-theory and F-theory on an \susu-model. We first compute the one-loop prepotentials of the \susu-model in all eight chambers and match them with the triple intersection numbers of the corresponding crepant resolutions to get the number of charged hypermultiplets in each irreducible representations (fundamentals, adjoints, and bifundamental) in Section \ref{sec:5d}. We then describe in Section \ref{sec:6dAnomaly} how the anomaly cancellation conditions in six-dimensional theory with a gauge group can be derived with geometric data with number of charged hypermultiplets as the unknowns. We return to our case for the gauge group \susu\, in Section \ref{sec:anomaly} and show that we get a unique solution of number of hypers for the six-dimensional theory and find it to match with the result we got for the five-dimensional theory in Section \ref{sec:5d}.

\subsection{$5d$ $\mathcal{N}=1$ supergravity theory with a gauge group \susu }
\label{sec:5d}

M-theory (an eleven-dimensional supergravity theory ) compactified on a Calabi--Yau threefold yields a low energy physics that describes five-dimensional ${\mathcal N}=1$ supergravity. 
The field content and interactions of such a 5d theory are determined by the topology and intersection ring of the Calabi--Yau threefold \cite{Cadavid:1995bk,Ferrara:1996hh}. 
When the Calabi--Yau threefold is a $G$-model, the resulting five-dimensional supergravity has a gauge group $G$. 

In the Coulomb phase of an $\mathcal{N}=1$ supergravity theory in five-dimensional spacetime, 
 the scalar fields of the vector multiplets are restricted to the Cartan sub-algebra of the Lie group as the Lie group is broken to $U(1)^r$ where $r$ is the rank of the group. It follows that the charge of a  hypermultiplet is simply given by a  
 weight of the representation under which it transforms.   When the representation has zero weights, only the hypermultiplets charged by non-zero weights are considered charged hypermultiplets. Let $\phi$ be in the Cartan subalgebra of a Lie algebra $\mathfrak{g}$. The one-loop prepotential obtained by integrating out the charged hypermultiplets is \cite{IMS}
 \begin{align}
6\mathscr{F}_{\text{IMS}} =&\frac{1}{2} \left(
\sum_{\alpha} |\langle \alpha, \phi \rangle|^3-\sum_{i} \sum_{\varpi\in \mathbf{R}_i} n_{\mathbf{R}_i} |\langle \varpi, \phi\rangle|^3 
\right),
\end{align}
where $\alpha$ are the simple roots, $\mathbf{R}_i$ are the irreducible representations of the gauge group, $\varpi$ are the weights of $\mathbf{R}_i$ \label{Eq:IMS}. This one-loop prepotential is the quantum contribution to the prepotential of a five-dimensional gauge theory with the matter fields in the representations  of the gauge group. The prepotential depends on the choice of a Coulomb chamber to get rid of the absolute values. 
We compute it for each of the eight chambers of an \susu-model. The chambers are defined by Table \ref{Table:ChambersIneq}. 
\begin{thm}\label{Thm:Prepotentials}
The prepotential of an \susu-model in the eight phases defined by the 
chambers of Table \ref{Table:ChambersIneq} are 
\quad

\begin{itemize}
\item Chamber 1
\begin{align}\nonumber
\begin{aligned}
6\mathscr{F}^{(1)}_{\text{IMS}}=& -8 (n_{\mathbf{1,8}}-1) \phi _1^3 -\frac{3}{2} (n_{\mathbf{1,3}}+n_{\mathbf{1,\bar{3}}}-2 n_{\mathbf{1,8}}+2)\phi _1^2 \phi _2 \\
&+\frac{3}{2} (n_{\mathbf{1,3}}+n_{\mathbf{1,\bar{3}}}+2 n_{\mathbf{1,8}}-2) \phi _1 \phi _2^2 +(-n_{\mathbf{1,3}}-n_{\mathbf{1,\bar{3}}}-8 n_{\mathbf{1,8}}+8)\phi _2^3 \\
& -(n_{\mathbf{2,1}}+3 n_{\mathbf{2,3}}+3 n_{\mathbf{2,\bar{3}}}+8 n_{\mathbf{3,1}}-8)\psi _1^3 -6(n_{\mathbf{2,3}}+n_{\mathbf{2,\bar{3}}}) \psi _1 \left(\phi _1^2 -\phi _1 \phi _2 +\phi _2^2\right)
\end{aligned}
\end{align}
\item Chamber 2
\begin{align}\nonumber
\begin{aligned}
6\mathscr{F}^{(2)}_{\text{IMS}}=& -8 (n_{\mathbf{1,8}}-1) \phi_1^3 -\frac{3}{2} (n_{\mathbf{1,3}}+n_{\mathbf{1,\bar{3}}}-2 n_{\mathbf{1,8}}+2) \phi_1^2\phi_2 +\frac{3}{2} (n_{\mathbf{1,3}}+n_{\mathbf{1,\bar{3}}}+2 n_{\mathbf{1,8}}-2)\phi_1 \phi_2^2 \\
&-(n_{\mathbf{1,3}}+n_{\mathbf{1,\bar{3}}}+8 n_{\mathbf{1,8}}+n_{\mathbf{2,3}}+n_{\mathbf{2,\bar{3}}}-8)\phi_2^3 -(n_{\bf{2,1}}+2 (n_{\mathbf{2,3}}+n_{\mathbf{2,\bar{3}}}+4 \text{n31}-4))\psi_1^3 \\
&-3(n_{\mathbf{2,3}}+n_{\mathbf{2,\bar{3}}}) \psi_1^2 \phi_2 
-3(n_{\mathbf{2,3}}+n_{\mathbf{2,\bar{3}}}) \psi_1 \left(2\phi_1^2 -2\phi_1 \phi_2+\phi_2^2\right)
\end{aligned}
\end{align}
\item Chamber 3
\begin{align}\nonumber
\begin{aligned}
6\mathscr{F}^{(3)}_{\text{IMS}}=& -(8 n_{\mathbf{1,8}}+n_{\mathbf{2,3}}+n_{\mathbf{2,\bar{3}}}-8) \phi _1^3 -\frac{3}{2} (n_{\mathbf{1,3}}+n_{\mathbf{1,\bar{3}}}-2 n_{\mathbf{1,8}}+2) \phi _1^2 \phi _2 \\
& +\frac{3}{2} (n_{\mathbf{1,3}}+n_{\mathbf{1,\bar{3}}}+2 n_{\mathbf{1,8}}-2) \phi _1 \phi _2^2 + (-n_{\mathbf{1,3}}-n_{\mathbf{1,\bar{3}}}-8 n_{\mathbf{1,8}}-n_{\mathbf{2,3}}-n_{\mathbf{2,\bar{3}}}+8)\phi _2^3 \\
& +(-n_{\mathbf{2,1}}-n_{\mathbf{2,3}}-n_{\mathbf{2,\bar{3}}}-8 n_{\mathbf{3,1}}+8)\psi _1^3 -3 (n_{\mathbf{2,3}}+n_{\mathbf{2,\bar{3}}}) \psi _1^2 \left(\phi _1+ \phi _2\right)\\
& -3 (n_{\mathbf{2,3}}+n_{\mathbf{2,\bar{3}}}) \psi _1 \left(\phi _1^2 -2\phi _1\phi _2 +\phi _2^2\right)
\end{aligned}
\end{align}
\item Chamber 4
\begin{align}\nonumber
\begin{aligned}
6\mathscr{F}^{(4)}_{\text{IMS}}=& -8 (n_{\mathbf{1,8}}-1) \phi _1^3 -\frac{3}{2} (n_{\mathbf{1,3}}+n_{\mathbf{1,\bar{3}}}-2n_{\mathbf{1,8}}+2n_{\mathbf{2,3}}+2n_{\mathbf{2,\bar{3}}}+2) \phi _1^2\phi _2 \\
& +\frac{3}{2} (n_{\mathbf{1,3}}+n_{\mathbf{1,\bar{3}}}+2 n_{\mathbf{1,8}}+2n_{\mathbf{2,3}}+2n_{\mathbf{2,\bar{3}}}-2) \phi _1 \phi _2^2 \\
& -(n_{\mathbf{1,3}}+n_{\mathbf{1,\bar{3}}}+8n_{\mathbf{1,8}}+2n_{\mathbf{2,3}}+2n_{\mathbf{2,\bar{3}}}-8))\phi _2^3 \\
& -(n_{\mathbf{2,1}}+8 n_{\mathbf{3,1}}-8)\psi _1^3 -6 (n_{\mathbf{2,3}}+n_{\mathbf{2,\bar{3}}})\psi _1^2 \phi _2 
\end{aligned}
\end{align}
\end{itemize}
The prepotential $\mathscr{F}^{(i')}_{\text{IMS}}$ (for $i=1,2,3,4$) is obtained from  $\mathscr{F}^{(i)}_{\text{IMS}}$ by  the involution $\phi_1\leftrightarrow\phi_2$. 
\end{thm}
\begin{proof}
Direct computation starting with equation \eqref{Eq:IMS} and using Table \ref{Table:ChambersIneq} to remove the absolute values.
\end{proof}

Following \cite{EY}, the number of hypermultiplets are computed by comparing the prepotential and the triple intersection polynomial given  in  Theorem \ref{Thm:TripleInt}. Comparing the triple intersection numbers obtained in the resolutions I, II, III, IV with the prepotentials computed  respectively in chambers 1, 2, 3, 4, we get
\begin{align}\label{eq:numbers}
\begin{aligned}
&  n_{\mathbf{2,1}}+ 8 n_{\mathbf{3,1}} = 4 L S + 2 S^2 - 3 S T + 8, \quad &&n_{\mathbf{1,8}} = \frac{1}{2} (-L T + T^2 + 2), \\
& n_{\mathbf{2,3}}+ n_{\mathbf{2,\overline{3}}} = S T,  \quad  &&n_{\mathbf{1,3}}+n_{\mathbf{1,\overline{3}}} =  T (9 L - 2 S - 3 T). 
\end{aligned}
\end{align}
We see in particular that the numbers $ n_{\mathbf{2,1}}$ and $n_{\mathbf{3,1}}$ are restricted by a linear relation but are not fixed by this method. 
In the case of Calabi--Yau threefolds, the vanishing of the first Chern class yields  $K=-L$ where $K$ is the canonical class of the base $B$. 
Using Witten's genus formula \cite{Witten},  we get
\begin{equation}
n_{\mathbf{2,1}} =-S (8K+2 S+3 T),\quad n_{\mathbf{3,1}}= \frac{1}{2} (KS+S^2+2) .
\end{equation}
We notice that the number of (bi)fundamental matter is  compatible with what is expected from  the technique of intersecting branes, and when the threefold is Calabi--Yau, $K=-L$ and $n_{\mathbf{1,8}}$ becomes the arithmetic genus of the curve $T$ as expected from Witten's genus formula.
 

\subsection{Anomaly cancellations in general $6d$ $\mathcal{N}=(1,0)$ supergravity theory}
\label{sec:6dAnomaly}
F-theory compactified  on a Calabi--Yau threefold gives a six-dimensional supergravity theory with eight supercharges coupled to $n_V$ vector, $n_T$ tensor, and $n_H^0$ neutral hypermultiplets \cite{Morrison:1996pp}. 
When the Calabi--Yau variety is elliptically-fibered with a gauge group $G$ and a representation $\mathbf{R}$ \cite{Sadov:1996zm,GM1,Park}, 
\begin{itemize}
\item the number of vectors: $n_V=\dim G$, 
\item the number of tensors: $n_T=9 -K^2$,
\item the number of hypers: $n_H=n_H^0+n_H^{ch}$ with neutral hypers $n_H^0=h^{2,1}+1$,
\end{itemize}
where $K$ is the canonical class of the base $Bst g$ of the elliptic fibration, and we have charged hypermultiplets transforming in the representation $\mathbf{R}$ of $G$. We consider the semisimple gauge group with simple components $G_a$ such that $G=\sum_a G_a$.

Six-dimensional ${\mathcal N}=(1,0)$  gauge theories can suffer from anomalies. 
We only care about local anomalies: pure gravitational anomalies, pure gauge anomalies, and mixed gravitational and gauge anomalies \cite{Sadov:1996zm, GM1,Park}. An effective way to address these anomalies are via using Green-Schwarz mechanism in six-dimensions \cite{Green:1984bx,Sagnotti:1992qw,Schwarz:1995zw}. The anomaly polynomial I$_8$ has a pure gravitational contribution from the term proportional to the $\tr R^4$, which is given by $\propto (n_H-n_V^{(6)}+29n_T-273)\tr R^4$, where $R$ is the Riemann tensor thought of as a $6\times 6$ real matrix of two-form values. 
In order to have vanishing gravitational anomalies, the coefficient of $\tr R^4$ has to vanish \cite{Salam}:
\begin{equation}
n_H-n_V^{(6)}+29n_T-273=0.
\end{equation}
The remainder terms of the anomaly polynomial I$_8$ is given by \cite{Sadov:1996zm, GM1,Park}
\begin{equation}
I_8=\frac{K^2}{8} (\tr R^2)^2 +\frac{1}{6}\sum_{a} X^{(2)}_{a} \tr R^2-\frac{2}{3}\sum_{a} X^{(4)}_{a}+4\sum_{a<b}Y_{ab},
\end{equation}
where contributions from each simple gauge component $X^{(n)}_{a}$ for $n=2,4$ and the mixed contribution $Y_{ab}$ are given by  \cite{Sadov:1996zm, GM1,Park}
\begin{equation}
X^{(n)}_{a}=\tr_{\bf{adj}}F^n_a -\sum_{i}n_{\bf{R}_{i,a}}\tr_{\bf{R}_{i,a}}F^n_a, \quad 
Y_{ab}=\sum_{i,j} n_{\bf{R}_{i,a}, \bf{R}_{j,b}} \tr_{\bf{R}_{i,a}}F^2_a \tr_{\bf{R}_{j,b}}F^2_b,
\end{equation}
where $n_{\bf{R}_{i,a}, \bf{R}_{j,b}}$ is the number of hypermultiplets transforming in the representation $(\mathbf{R}_{i,a},\mathbf{R}_{j,b})$ of the gauge group $G_a\times G_b$. Note that the mixed term by computing all possible pairs of the simple components of the gauge groups which are denoted by two indices.

It is important to note that when a representation is charged on both simple components of the group, it affects not only  $Y_{ab}$ but also $X^{(2)}_a$ and $X^{(4)}_a$. Consider a representation $(\bf{R_1},\bf{R_2})$ for of a semisimple group with two simple components $G=G_1\times G_2$, where $\bf{R_a}$ is a representation of $G_a$. Then this representation contributes $\dim{\bf{R_2}}$ times   to $n_{\bf{R_1}}$, and contributes  $\dim{\bf{R_1}}$ times to $n_{\bf{R_2}}$:
\begin{align}
n_{\bf{R_1}}=\cdots+\dim{\bf{R_2}} \ n_{\bf{R_1},\bf{R_2}}, \quad n_{\bf{R_2}}=\cdots +\dim{\bf{R_1}} \ n_{\bf{R_1},\bf{R_2}}.
\label{eq:coupledtermnR}
\end{align}

Siince the hypermultiplets of zero weights are neutral, we have to remove the neutral hypermultiplet contributions to get the charged dimension of the hypermultiplets in each irreducible representation. By denoting the zero weights of a representation $\bf{R}_i$ as $\bf{R}^{(0)}_i$, the charged dimension of the hypermultiplets in representation $\bf{R}_i$ is given by \cite{GM1}
$$\dim{\bf{R}_i}-\dim{\bf{R}_{i}^{(0)}}.$$
For a representation $\bf{R}_i$, $n_{\bf{R}_i}$ denotes the multiplicity of the representation $\bf{R}_i$. Then the number of charged hypermultiplets is given by \cite{GM1}
\begin{equation}
n_H^{ch}=\sum_{i} n_{\bf{R}_i} \left( \dim{\bf{R}_i}-\dim{\bf{R_{i}^{(0)}}} \right).
\end{equation}

The trace identities for a representation $\mathbf{R}_{i,a}$ of a simple group $G_a$ are
\begin{equation}
\tr_{\bf{R}_{i,a}} F^2_a=A_{\bf{R}_{i,a}} \tr_{\bf{F}_a} F^2_a , \quad \tr_{\bf{R}_{i,a}} F^4_a=B_{\bf{R}_{i,a}} \tr_{\bf{F}_a} F^4_a+C_{\bf{R}_{i,a}} (\tr_{\bf{F}_a} F^2_a )^2
\end{equation}
with respect to a reference representation $\bf{F}_a$ for each simple component $G_a$ of the gauge group.\footnote{We denoted this representation as $\bf{F}_a$ as we picked the fundamental representations for convenience. However, any representation can be used as a reference representation.} The coefficients $A_{\bf{R}_{i,a}}$, $B_{\bf{R}_{i,a}}$, and $C_{\bf{R}_{i,a}}$ depends on the gauge groups and are listed in \cite{Erler,Avramis:2005hc,vanRitbergen:1998pn}.

Then we can write the gauge contribution terms with respect to the coefficients from the trace identities:
\begin{align}
X^{(2)}_a&=\left(A_{a,\bf{adj}}-\sum_{i}n_{\bf{R}_{i,a}}A_{\bf{R}_{i,a}}\right)\tr_{\bf{F}_a}F^2_a, \\
X^{(4)}_a&=\left(B_{a,\bf{adj}}-\sum_{i}n_{\bf{R}_{i,a}}B_{\bf{R}_{i,a}}\right)\tr_{\bf{F}_a}F^4_a
+\left(C_{a,\bf{adj}}-\sum_{i}n_{\bf{R}_{i,a}}C_{\bf{R}_{i,a}}\right)(\tr_{\bf{F}_a}F^2_a)^2 , \label{eq:X4}\\
Y_{ab}&=\sum_{i,j} n_{\bf{R}_{i,a},\bf{R}_{j,b}} A_{\bf{R}_{i,a}} A_{\bf{R}_{j,b}} \tr_{\bf{F}_a}F^2_a \tr_{\bf{F}_b}F^2_b, \quad (a\neq b).
\end{align}
For each simple component $G_a$, the anomaly polynomial  I$_8$ has a pure gauge contribution proportional to the quartic term $\tr F^4_a$, which is contained in equation \eqref{eq:X4}. In order to have a vanishing pure gauge anomalies, the coefficients of these terms have to vanish:
$$
B_{a,\bf{adj}}-\sum_{i}n_{\bf{R}_{i,a}}B_{\bf{R}_{i,a}}=0.
$$
When the coefficients of all quartic terms  ($\tr R^4$ and $\tr F^4_a$) vanish,  the remaining part of the anomaly polynomial I$_8$ is
\begin{align}
\begin{cases}
&\mathrm{I}_8=\frac{K^2}{8} (\tr R^2)^2 +\frac{1}{6}\sum_{a} X^{(2)}_{a} \tr R^2-\frac{2}{3}\sum_{a} X^{(4)}_{a}+4\sum_{a<b}Y_{ab}, \\[4pt]
&X^{(2)}_a=\left(A_{a,\bf{adj}}-\sum_{i}n_{\bf{R}_{i,a}}A_{\bf{R}_{i,a}}\right)\tr_{\bf{F}_a}F^2_a, \quad
X^{(4)}_a=\left(C_{a,\bf{adj}}-\sum_{i}n_{\bf{R}_{i,a}}C_{\bf{R}_{i,a}}\right)(\tr_{\bf{F}_a}F^2_a)^2, \\[4pt]
&Y_{ab}=\sum_{i,j} n_{\bf{R}_{i,a},\bf{R}_{j,b}} A_{\bf{R}_{i,a}} A_{\bf{R}_{j,b}} \tr_{\bf{F}_a}F^2_a \tr_{\bf{F}_b}F^2_b, \quad (a\neq b).
\end{cases}
\end{align}
The anomalies are canceled by the Green-Schwarz mechanism when I$_8$ factorizes \cite{Green:1984bx,Sagnotti:1992qw,Schwarz:1995zw}. 
The modification of the field strength $H$ of  the antisymmetric tensor $B$ is 
\begin{equation}
H= dB + \frac{1}{2} K \omega_{3L}+ 2\sum_a\frac{S_a}{\lambda_a}\omega_{a,3Y}, 
\end{equation}
where  $\omega_{3L}$ and $\omega_{a,3Y}$ are respectively the gravitational Yang--Mills and Chern--Simons  terms. 
 If I$_8$ factors as 
 \begin{equation}
 \text{I}_8= X\cdot  X, \quad X= \frac{1}{2} K \tr R^2 +\sum_a \frac{2}{\lambda_a} S_a\tr F^2_a,
 \end{equation}
where the  $\lambda_a$ are normalization factors  chosen such that the  smallest topological charge of an embedded SU($2$) instanton in G$_a$ is one \cite{Kumar:2010ru, Park, Bernard}. This forces $\lambda_a$ to be the Dynkin index of the fundamental representation of  $G_a$ as summarized in Table \ref{tb:normalization} \cite{Park}. \\

\begin{table}[h!]
\begin{center}
\begin{tabular}{|c|c|c|c|c|c|c|c|c|c|}
\hline
 $\mathfrak{g}$ & A$_n$ & B$_n$ & C$_n$ & D$_n$ & E$_8$ & E$_7$ & E$_6$&  F$_4$ & G$_2$ \\
 \hline
 $\lambda$ & $1$ & $2$  & $1$ & $2$ & $60$ & $12$ & $6$ & $6$ & $2$ \\
 \hline  
\end{tabular}
\caption{The normalization factors for each simple gauge algebra. See \cite{Kumar:2010ru}.}
\label{tb:normalization}
\end{center}
\end{table}
When I$_8$ factors, the anomaly is canceled by adding the following Green-Schwarz counter-term 
\begin{equation}
\Delta L_{GS}\propto \frac{1}{2} B \wedge X,
\end{equation}
which implies that $X$ carries string charges.

With all the local anomaly cancellation conditions in six-dimensions investigated, we can coalesce as the set of equations with number of charged hypermultiplets in each irreducible representations as unknwons. If the simple group $G_a$ is supported on a divisor $S_a$, the local anomaly cancellation conditions are the following equations   \cite{Sadov:1996zm, GM1,Park}:
\begin{subequations}\label{eq:AnomalyEqn}
\begin{align}
n_T&=9-K^2 , \\
n_H-n_V^{(6)}+29n_T-273 &=0,\\
\left(B_{a,\bf{adj}}-\sum_{i}n_{\bf{R}_{i,a}}B_{\bf{R}_{i,a}}\right)& = 0, \\
\lambda_a  \left(A_{a,\bf{adj}}-\sum_{i}n_{\bf{R}_{i,a}}A_{\bf{R}_{i,a}}\right) & =6  K\cdot S_a, \\
\lambda_a^2 \left(C_{a,\bf{adj}}-\sum_{i}n_{\bf{R}_{i,a}}C_{\bf{R}_{i,a}}\right) & =-3 S_a ^2, \\
\lambda_a \lambda_b \sum_{i,j} n_{\bf{R}_{i,a},\bf{R}_{j,b}} A_{\bf{R}_{i,a}} A_{\bf{R}_{j,b}}  & =S_a\cdot S_b, \quad (a\neq b).
\end{align}
\end{subequations}

Assuming the first three equations hold, 
cancelling the anomalies is equivalent to factoring the  anomaly polynomial  \cite{Sadov:1996zm, GM1,Park}
\begin{equation}
I_8 =\frac{K^2}{8} (\tr R^2)^2 +\frac{1}{6} (X^{(2)}_{1} +X^{(2)}_{2}) \tr R^2-\frac{2}{3} (X^{(4)}_{1}+X^{(4)}_{2})+4Y_{12} .
\end{equation}

To summarize, for a compactification on an elliptically-fibered Calabi--Yau threefold $Y$, the number of multiplets is \cite{GM1}
\begin{subequations}
\begin{align}
&n_V^{(6)}=\dim{G}, \quad n_T=h^{1,1}(B)-1=9-K^2 , \\
&n_H=n_H^0+n_H^{ch}=h^{2,1}(Y)+1+\sum_{i} n_{\bf{R}_i} \left( \dim{\bf{R}_i}-\dim{\bf{R_{i}^{(0)}}} \right),
\end{align}
\end{subequations}
where $n_{\bf{R}_i}$ is a number of hypermultiplets charged under each irreducible representation $\bf{R}_i$ that satisfies equations \eqref{eq:AnomalyEqn}.

\subsection{Anomaly cancellations in $6d$ $\mathcal{N}=(1,0)$ supergravity theory}
\label{sec:anomaly}
In this section, we check that the gravitational, gauge, and mixed gravitational-gauge anomalies of the six-dimensional supergravity are all canceled when the Lie algebra and the representation are
$$\mathfrak{g}= \text{A}_1\oplus \text{A}_2, \quad \bold{R}=(\bold{2},\bold{1})\oplus(\bold{1},\bold{3})\oplus(\bold{1},\bold{\bar{3}})\oplus(\bold{2},\bold{3})\oplus(\bold{2},\bold{\bar{3}})\oplus (\bold{3},\bold{1})\oplus (\bold{1},\bold{8}).$$
 We follow the approach of Sadov \cite{Sadov:1996zm} (see also \cite{GM1} and \cite{Monnier:2017oqd}) and use the notation of  \cite{  EKY2}. The six-dimensional anomaly cancellation conditions put linear constraints on the number of charged hypermultiplets transforming in each irreducible representations. 

First, we recall that the Euler characteristic of a Calabi--Yau threefold defined as a crepant resolution of the Weierstrass model of an \susu-model is 
\begin{equation}
\chi(Y)=-6 (10 K^2+5 K S+8 K T+S^2+2 S T+2 T^2),
\end{equation}
where $S$ supports  $\text{A}_1$ and $T$ supports  $\text{A}_2$.  We also assume that $S$ and $T$ are smooth divisors intersecting transversally. 

We will use the anomaly cancellation conditions to explicitly compute the number of hypermultiplets transforming in each representation by requiring all anomalies to cancel.
We will see that they are the same as those found in five-dimensional supergravity  by comparing the triple intersection numbers of the fibral divisors and the cubic prepotentials in the Coulomb phase. 

The Lie algebra of type A$_1$ (resp. A$_2$) only has a unique  quartic Casimir invariant so that we don't have to impose the vanishing condition for the coefficients of $ \mathrm{tr}\  F^4_1$  (resp. $ \mathrm{tr}\  F^4_2$). 
We have the following  trace identities
\begin{align}
& \mathrm{tr}_{\bf{3}}F^2_1=4 \mathrm{tr}_{\bf{2}}F^2_2, \quad  \mathrm{tr}_{\bf{3}} F^4_1=8 ( \mathrm{tr}_{\bf{3}} F^2_1)^2, \quad  \mathrm{tr}_{\bf{2}}F^4_1=\frac{1}{2} ( \mathrm{tr}_{\bf{2}}F^2_1)^2,\\
& \mathrm{tr}_{\bf{8}}F^2_2=6 \mathrm{tr}_{\bf{3}}F^2_2 , \quad  \mathrm{tr}_{\bf{8}} F^4_2=9 ( \mathrm{tr}_{\bf{3}} F^2_2)^2 , \quad   \mathrm{tr}_{\bf{3}} F^4_2=\frac{1}{2} ( \mathrm{tr}_{\bf{3}} F^2_2)^2,
\end{align}
which give 
\begin{align}
\begin{split}
X^{(2)}_1 &= (4 -4 n_{\bf{3,1}} - n_{\bf{2,1}} - 3 n_{\bf{2,3}} -3 n_{\bf{2,\bar{3}}}) \ { \mathrm{tr}}_{\bf{2}} F^2_1,\\
X^{(2)}_2 &=( 6 - 6n_{\bf{1,8}} - n_{\bf{1,3}} - n_{\bf{1,\bar{3}}}- 2 n_{\bf{2,3}} -2 n_{\bf{2,\bar{3}}})\ { \mathrm{tr}}_{\bf{3}} F^2_2,\\
X^{(4)}_1 &= (8  -8 n_{\bf{3,1}} - \frac{1}{2} n_{\bf{2,1}} - \frac{3}{2}n_{\bf{2,3}} -\frac{3}{2} n_{\bf{2,\bar{3}}} ) ( \mathrm{tr}_{\bf{2}} F^2_1)^2,\\
X^{(4)}_2 &=   (9 -9 n_{\bf{1,8}}-\frac{1}{2} n_{\bf{1,3}}-\frac{1}{2} n_{\bf{1,\bar{3}}}-n_{\bf{2,3}}-n_{\bf{2,\bar{3}}})(\mathrm{tr}_{\bf{3}} F^2_2)^2,\\
Y_{\bf{23}} \ &= (n_{\bf{2,3}} + n_{\bf{2,\bar{3}}}) \ { \mathrm{tr}}_{\bf{3}} F^2_2 \ { \mathrm{tr}}_{\bf{2}} F^2_1.
\end{split}
\end{align}
Following Sadov, the anomaly cancellation conditions are \cite{Sadov:1996zm}:
\begin{align}
\begin{split}
X^{(2)}_1 &=6 K S \ { \mathrm{tr}}_{\bf{2}}  F^2_1,\quad \quad
X^{(2)}_2 =6 K T\ { \mathrm{tr}}_{\bf{3}} F^2_2, \\
X^{(4)}_1 &= -3  S^2 ( \mathrm{tr}_{\bf{2}} F^2_1)^2,\quad
X^{(4)}_2 = -3 T^2(\mathrm{tr}_{\bf{3}} F^2_2)^2,\quad
Y_{\bf{23}} = S T \ { \mathrm{tr}}_{\bf{3}} F^2_2 \ { \mathrm{tr}}_{\bf{2}} F^2_1.
\end{split}
\end{align}
Comparing the coefficients, we get the following linear equations
\begin{align}
\begin{aligned}
6 (1 - n_{\bf{1,8}}) - (n_{\bf{1,3}} + n_{\bf{1,\bar{3}}}) - 2 (n_{\bf{2,3}} + n_{\bf{2,\bar{3}}}) &=  6 K T,\quad 
 4 (1- n_{\bf{3,1}}) - n_{\bf{2,1}} -3 (n_{\bf{2,3}} + n_{\bf{2,\bar{3}}}) = 6 K S, \\
9(1 - n_{\bf{1,8}})-\frac{1}{2} (n_{\bf{1,3}}+n_{\bf{1,\bar{3}}})-(n_{\bf{2,3}}+n_{\bf{2,\bar{3}}}) &=-3 T^2,\quad
 8 (1 - n_{\bf{3,1}}) - \frac{1}{2} n_{\bf{2,1}} - \frac{3}{2}(n_{\bf{2,3}} + n_{\bf{2,\bar{3}}})  =-3 S^2,\\
 n_{\bf{2,3}} + n_{\bf{2,\bar{3}}} &= ST.
 \end{aligned}
\end{align}
These linear equations  have the following unique solution\footnote{ We recall that Witten's genus formula asserts  that the number of hypermultiplets charged under the adjoint representation is given by the genus of the curve supporting the corresponding gauge group.} 
\begin{align}\label{eq:numbersReal}
\begin{aligned}
 n_{\bf{1,8}} & = \frac{1}{2}(KT+T^2+2), \quad   && \quad \quad\  \    n_{\bf{3,1}} = \frac{1}{2}(KS+S^2+2),\\
  n_{\bf{2,1}} & =-S(8 K +2S+3T), \quad   &&n_{\bf{1,3}}+n_{\bf{1,\bar{3}}} = -T(9K+2S+3 T), \quad 
 && n_{\bf{2,3}} + n_{\bf{2,\bar{3}}}= ST.
 \end{aligned}
 \end{align}
 These numbers have simple geometric interpretations. The numbers 
$ n_{\bf{1,8}}$ and  $n_{\bf{3,1}}$ are respectively the genus of the curves $T$ and $S$. The number  $n_{\bf{2,3}}+n_{\bf{2,\bar{3}}}$ is the degree of $S\cdot T$ (intersection number of $S$ and $T$).  The number $n_{\bf{2,1}}$ is the intersection number of  $S$ and the discriminant of the generic fiber of $D_1^{\text{s}}$. The number $ n_{\bf{1,3}}+n_{\bf{1},\bar{3}}$ is  the intersection number of $T$ and the discriminant of $D_2^t$.

From here we can get the numbers of hypermultiplets $n_{\bf{2}}$ and $n_{\bf{3}}+n_{\bf{\bar{3}}}$ tracing back from equation \eqref{eq:coupledtermnR}:
\begin{align}
n_{\bf{2}}=-2(4K+S), \quad
n_{\bf{3}}+n_{\bf{\bar{3}}}=-3T(3K+T).
\end{align}

We recall that the Hodge numbers of a crepant resolution of an \susu-model are (see Theorem \ref{Thm:Hodge})
\begin{equation}
h^{1,1}(Y)=14-K^2, \quad h^{2,1}(Y)=29 K^2+15 K S+24 K T+3 S^2+6 S T+6 T^2+14.
\end{equation}
The total number of hypermultiplets is the sum of the number of neutral hypermultiplets $(n_H^0= h^{2,1}(Y)+1)$ and the number of charged hypermultiplets\footnote{ To count the charged hypermultiplets, each irreducible representation $\mathbf{R}$ 
contributes $\dim^{ch}{\bf{R}} \times n_{\bf{R}}$ where $\dim^{ch}\ \bf{R}$ is the number of non-zero weights of the representation $\bf{R}$ and $n_{\bf{R}}$ is the number of hypermultiplets transforming in the representation $\bf{R}$ \cite{GM1}.}. 
\begin{align}\label{eq:NH0}
 n_H^0&=h^{2,1}(Y)+1 =29 K^2+15 K S+24 K T+3 S^2+6 S T+6 T^2+15,\\
 n_H^{ch}&=2 n_{\bf{2,1}}+6 (n_{\bf{2,3}}+ n_{\bf{2,\bar{3}}})+3 (n_{\bf{1,3}}+ n_{\bf{1,\bar{3}}})+(8-2) n_{\bf{1,8}}+(3-1) n_{\bf{3,1}}.
\end{align}
Thus, the total number of hypermultiplets is  
\begin{align}
n_H &=n_H^0+n_H^{ch}=
29 K^2+23.
\end{align}
The numbers of vector multiplets and  tensor multiplets are
\begin{align}
& n_T=9-K^2, \quad n_V=\dim G=\dim\  \text{SU}(2)+\dim\ \text{SU}(3)=3+8=11.
\end{align}
We  can now check that the  coefficient of $\mathrm{tr}\ R^4$ vanishes \cite{Salam}:
\begin{equation}
n_H-n_V+29n_T-273=0.
\end{equation}
Finally, we show that the  anomaly polynomial I$_8$ indeed factors as a perfect square:
\begin{align}
\begin{split}
I_8 &= \frac{K^2}{8} ( \mathrm{tr} R^2)^2 + \frac{1}{6} (X_1^{(2)} + X_2^{(2)})  \mathrm{tr} R^2 - \frac{2}{3} (X_1^{(4)} + X_2^{(4)}) + 4 Y_{\bf{23}},\\
&=
\frac{1}{2} \left(\frac{1}{2} K\ {\mathrm{tr}} R^2 +2 S\ {\mathrm{tr}}_{\bf{2}} F^2_1 +2 T\ {\mathrm{tr}}_{\bf{3}} F^2_2 \right)^2.
\end{split}
\end{align}
Hence, we  conclude that all the local anomalies are canceled via Green-Schwarz mechanism.

The global anomalies for the Standard Model in 4d was discussed in Section \ref{sec:sm} by examining the fourth homotopy group of SU($2$) and SU($3$). Similarly, the global anomaly contributions from SU($2$) and SU($3$) in six-dimensions can be discussed by looking into the sixth homotopy group for SU($2$) and SU($3$), which are given by $\mathbb{Z}_{12}$ and $\mathbb{Z}_6$ respectively. Bershadsky and Vafa has shown that this yields the linear constraints on the number of hypermultiplets \cite{Bershadsky:1997sb}. In the case of SU($2$) and SU($3$),  we have:
\begin{align}
\begin{cases}
& \text{SU}(2):\quad 4-n_{\mathbf{2}} =0\quad \mod\ 6, \\
& \text{SU}(3):\quad n_{\mathbf{3}}=0\quad \mod\ 6,
\end{cases}
\end{align}
where $n_{\mathbf{2}}$ and $n_{\mathbf{3}}$ are the number of hypermultiplets transforming in the fundamental representation of SU($2$) and SU($3$) respectively. 
Using equation \eqref{eq:numbersReal}, we immediately compute geometrically the number of charged hypermultiplets as
\begin{align}
\begin{cases}
& n_{\mathbf{2}} =n_{\mathbf{2,1}}+3(n_{\mathbf{2,3}}+n_{\mathbf{2,\overline{3}}}) = -16(g(S)-1)+6S^2, \\
& n_{\mathbf{3}} =n_{\mathbf{1,3}}+n_{\mathbf{1,\overline{3}}}+2(n_{\mathbf{2,3}}+n_{\mathbf{2,\overline{3}}}) = -18(g(T)-1)+6T^2,
\end{cases}
\end{align}
where $g(S)$ and $g(T)$ are the genus of the curves supporting $S$ and $T$ respectively.
By using these conditions we find that the global anomalies of SU($3$) always vanishes whereas SU($2$) global anomaly canceling condition is found to be
\begin{align}
g(S)=0 \quad \mod\ 3.
\end{align}

\section*{Acknowledgements}
The authors are grateful to Prateek Agrawal, Paolo Aluffi, Patrick Jefferson, Kenji Matsuki, Washington Taylor and Shing-Tung Yau for helpful discussions. 
M.J.K. would like to thank Korean Institute for Advanced Studies for their hospitality during part of this work. 
M.E. is supported in part by the National Science Foundation (NSF) grant DMS-1701635 ``Elliptic Fibrations and String Theory''.
M.J.K. is supported by the National Science Foundation (NSF) grant PHY-1352084. 
R.J. is supported by a National Science Foundation (NSF) Graduate Research Fellowship. 

\appendix
\section{Fiber enhancement}
 \label{sec:FiberEnhancement} 
\begin{figure}[htb]
\begin{center}
\scalebox{1}{
\begin{tabular}{ccccc}
  & Type I$_2$ & Type III & Type IV  & Type I$_3$ \\
\raisebox{1cm}{Singular fiber}  & \includegraphics[scale=.7,angle=90]{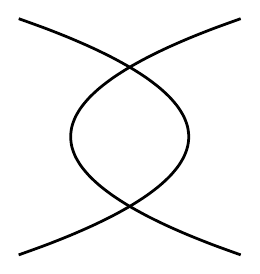}
 &
  \includegraphics[scale=1.2, angle=90]{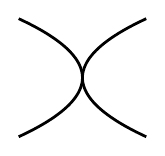}
  &
\raisebox{-.3cm}{  \includegraphics[scale=1,angle=30]{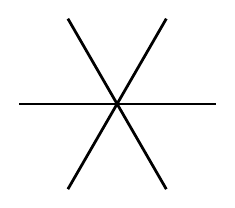}} & \includegraphics[scale=.4]{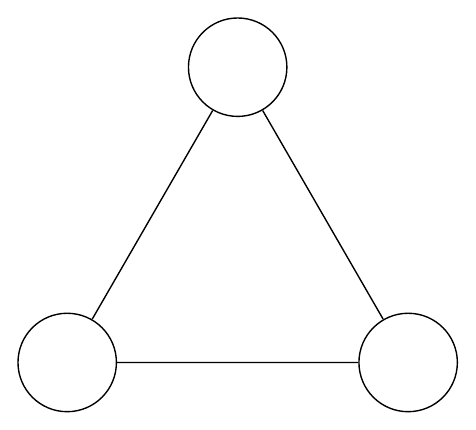}    
  \\
\raisebox{1cm}{Our notation}  &
 \includegraphics[scale=.6]{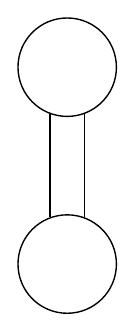}
 &
  \includegraphics[scale=.6]{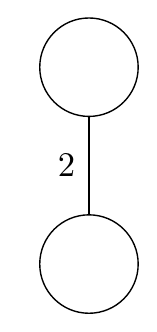}
  &
  \includegraphics[scale=.5]{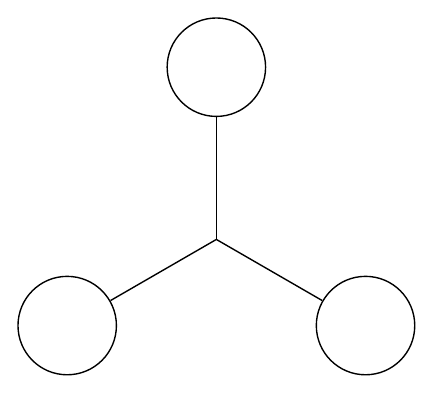} &  \includegraphics[scale=.4]{I3}    \\
\raisebox{1cm}{ Dual graph }& \includegraphics[scale=.6]{I2GraphV} &  \includegraphics[scale=.6]{I2GraphV}&
  \includegraphics[scale=.4]{I3}  & \includegraphics[scale=.4]{I3}   
   \end{tabular}}
  \end{center}
\caption{Convention for Kodaira fibers of type I$_2$, III, IV, and I$_3$.\label{fig:Convention}}
\end{figure}

\clearpage

\begin{table}
\begin{tikzcd}[column sep=normal]
& \begin{tabular}{c} \includegraphics[scale=.6]{I3} \end{tabular} \arrow[rightarrow]{r}{\displaystyle{a_1=0}}   \arrow[rightarrow]{rd}{\displaystyle{T}}& \begin{tabular}{c} \includegraphics[scale=1.2]{IV} \end{tabular} &  \\
\begin{tabular}{c} \includegraphics[trim=0cm  1cm 0 0cm,scale=1]{I2} \end{tabular}\arrow[rightarrow]{ru}{\displaystyle{P_1=0}}  \arrow[rightarrow]{r} {\displaystyle{a_1=0}}  \arrow[rightarrow]{rd}{\displaystyle{T}} &  \raisebox{-.5cm}{ \scalebox{1.4}{ \includegraphics{III}} } \arrow[rightarrow]{ru}{\displaystyle{P_1=0}}  \arrow[rightarrow]{rd} {\displaystyle{T}}     &   \raisebox{-.5cm}{ \scalebox{.6}{ \includegraphics{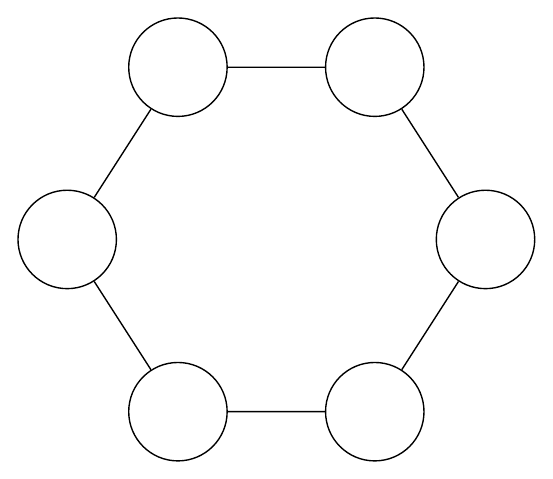}} }  &  \\
 & \raisebox{-.5cm}{ \scalebox{.6}{ \includegraphics{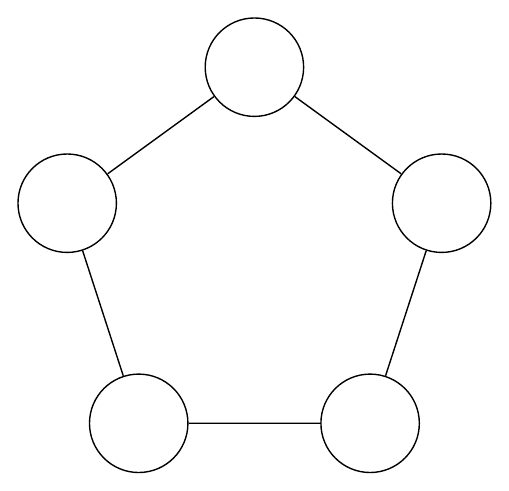}} } \arrow[rightarrow]{r}{\displaystyle{a_1=0}}  \arrow[rightarrow,near end]{ru}{\displaystyle{a_1 \widetilde{a}_6-\widetilde{a}_3 \widetilde{a}_4=0}} 
 & 
  \begin{tabular}{c} \includegraphics[trim=-1.2cm  0cm 0 0cm,scale=.6]{Tstar12-12-1} 
  \end{tabular} &   \\
\begin{tabular}{c} \includegraphics[scale=.6]{I3}\end{tabular} \arrow[rightarrow]{ru}{\displaystyle{S}}  \arrow[rightarrow]{r}{\displaystyle{a_1=0}}  \arrow[rightarrow]{rd}{\displaystyle{P_2=0}} &  \begin{tabular}{c}\includegraphics[scale=1.2]{IV} \end{tabular}\arrow[rightarrow]{rd}{\displaystyle{P_2=0}}  \arrow[rightarrow]{ru} {\displaystyle{S}}  & &  \\
& \raisebox{-.5cm}{ \scalebox{.6}{ \includegraphics{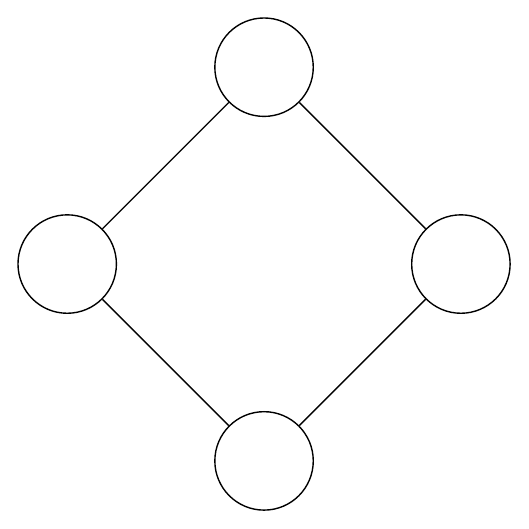}} }\arrow[rightarrow]{r} {\displaystyle{a_1=0}}  \arrow[rightarrow]{ruuu} {\displaystyle{S}} &     \begin{tabular}{c}\includegraphics[scale=.7]{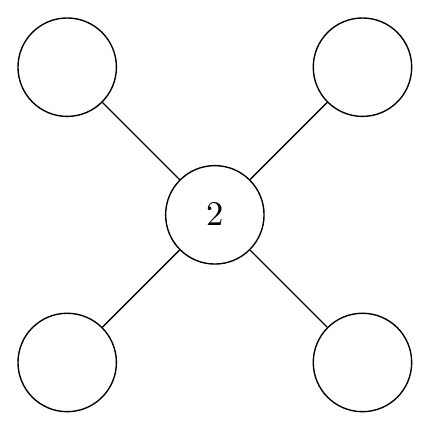}\end{tabular}  & 
\end{tikzcd}
\caption{I$_2^{\text{s}}$ + I$_3^{\text{s}}$, Resolution I, II, and III.   { $P_1=\widetilde{a}_4^2 t -a_1^2 \widetilde{a}_6$ and $P_2=\widetilde{a}_3^3 s-a_1\widetilde{a}_2\widetilde{a}_3^2s+a_1^2\widetilde{a}_3\widetilde{a}_4-a_1^3\widetilde{a}_6$. 
The non-Kodaira fiber in codimension-three is a contraction of a IV$^*$. 
}}
\end{table}
\clearpage 

\begin{table}
\begin{tikzcd}[column sep=normal]
& \begin{tabular}{c} \includegraphics[scale=.6]{I3} \end{tabular} \arrow[rightarrow]{r}{\displaystyle{a_1=0}}   \arrow[rightarrow]{rd}{\displaystyle{T}}& \begin{tabular}{c} \includegraphics[scale=1.2]{IV} \end{tabular} &  \\
\begin{tabular}{c} \includegraphics[trim=0cm  1cm 0 0cm,scale=1]{I2} \end{tabular}\arrow[rightarrow]{ru}{\displaystyle{P_1=0}}  \arrow[rightarrow]{r} {\displaystyle{a_1=0}}  \arrow[rightarrow]{rd}{\displaystyle{T}} &  \raisebox{-.5cm}{ \scalebox{1.4}{ \includegraphics{III}} } \arrow[rightarrow]{ru}{\displaystyle{P_1=0}}  \arrow[rightarrow]{rd} {\displaystyle{T}}     &   \raisebox{-.5cm}{ \scalebox{.6}{ \includegraphics{I6}} }  &  \\
 & \raisebox{-.5cm}{ \scalebox{.6}{ \includegraphics{N_I5}} } \arrow[rightarrow]{r}{\displaystyle{a_1=0}}  \arrow[rightarrow,near end]{ru}{\displaystyle{a_1 \widetilde{a}_6-\widetilde{a}_3 \widetilde{a}_4=0}} 
 & 
  \begin{tabular}{c} \includegraphics[trim=-2.5cm  0cm 0 0cm,scale=.6]{T12321} 
  \end{tabular} &   \\
\begin{tabular}{c} \includegraphics[scale=.6]{I3}\end{tabular} \arrow[rightarrow]{ru}{S}  \arrow[rightarrow]{r}{\displaystyle{a_1=0}}  \arrow[rightarrow]{rd}{\displaystyle{P_2=0}} &  \begin{tabular}{c}\includegraphics[scale=1.2]{IV} \end{tabular}\arrow[rightarrow]{rd}{\displaystyle{P_2=0}}  \arrow[rightarrow]{ru} {\displaystyle{S}}  & &  \\
& \raisebox{-.5cm}{ \scalebox{.6}{ \includegraphics{I4}} }\arrow[rightarrow]{r} {\displaystyle{a_1=0}}  \arrow[rightarrow]{ruuu} {\displaystyle{S}} &     \begin{tabular}{c}\includegraphics[scale=.6]{Istar0}\end{tabular}  & 
\end{tikzcd}
\caption{I$_2^{\text{s}}$ + I$_3^{\text{s}}$, Resolution IV.    {$P_1=\widetilde{a}_4^2 t -a_1^2 \widetilde{a}_6$ and $P_2=\widetilde{a}_3^3 s-a_1\widetilde{a}_2\widetilde{a}_3^2s+a_1^2\widetilde{a}_3\widetilde{a}_4-a_1^3\widetilde{a}_6$. The non-Kodaira fiber is a contracted fiber of type IV$^*$.}}
\end{table}

\clearpage 
\newpage

\newpage 
\begin{table}
\begin{tikzcd}[column sep=normal]
& \begin{tabular}{c} \includegraphics[scale=.6]{I3} \end{tabular} \arrow[rightarrow]{r}{\displaystyle{a_1^2+4\widetilde{a}_2 t=0}}    \arrow[rightarrow, near end]{rd}{\displaystyle{T}}  & \begin{tabular}{c} \includegraphics[scale=1.2]{IV} \end{tabular} &  \\
\begin{tabular}{c} \includegraphics[trim=0cm  1cm 0 0cm,scale=1]{I2} \end{tabular}\arrow[rightarrow]{ru}{\displaystyle{P_1=0}}  \arrow[rightarrow]{r} {\displaystyle{a_1^2+4\widetilde{a}_2 t=0}}  \arrow[rightarrow]{rd} {\displaystyle{T}} &  \raisebox{-.5cm}{ \scalebox{1.4}{ \includegraphics{III}} } \arrow[rightarrow]{ru} {\displaystyle{P_1=0}}  \arrow[rightarrow]{rd}   {\displaystyle{T}}   &   \raisebox{-.5cm}{ \scalebox{.6}{ \includegraphics{I6}} }  &  \\
 & \raisebox{-.5cm}{ \scalebox{.6}{ \includegraphics{N_I5}} } \arrow[rightarrow]{r} {\displaystyle{a_1=0}} \arrow[rightarrow, near end]{ru} {\displaystyle{P_3=0}}
 & 
  \begin{tabular}{c} \includegraphics[trim=-2.5cm  1cm 0 0cm,scale=.6]{T122-1-1} 
  \end{tabular} &   \\
\begin{tabular}{c} \includegraphics[scale=.6]{I3}\end{tabular} \arrow[rightarrow]{ru}  {\displaystyle{S}} \arrow[rightarrow]{r} {\displaystyle{a_1=0}} \arrow[rightarrow]{rd} {\displaystyle{P_2=0}} &  \begin{tabular}{c}\includegraphics[scale=1.2]{IV} \end{tabular}\arrow[rightarrow]{rd} {\displaystyle{P_2=0}}  \arrow[rightarrow]{ru} {\displaystyle{S}}  & &  \\
& \raisebox{-.5cm}{ \scalebox{.6}{ \includegraphics{I4}} }\arrow[rightarrow]{r} {\displaystyle{a_1=0}}   \arrow[rightarrow, near end]{ruuu} {\displaystyle{S}}  &     \begin{tabular}{c}\includegraphics[scale=.6]{Istar0}\end{tabular}  & 
\end{tikzcd} 
\caption{I$_2^{\text{ns}}$ + I$_3^{\text{s}}$ Resolution I and IV.  $P_1=-2 \widetilde{a}_4^2 t+4 \widetilde{a}_2 \widetilde{a}_6 t+a_1^2 \widetilde{a}_6 -2 \widetilde{a}_2 \widetilde{a}_3^2$, $P_2=\widetilde{a}_3^3 s-a_1\widetilde{a}_2\widetilde{a}_3^2+a_1^2\widetilde{a}_3\widetilde{a}_4-a_1^3\widetilde{a}_6$, and $P_3=\widetilde{a}_2\widetilde{a}_3^2 s-a_1(a_1\widetilde{a}_6-\widetilde{a}_3\widetilde{a}_4)$. The non-Kodaira fiber is a contraction of an I$_1^*$.}
\end{table}
\clearpage 

\begin{table}
\begin{tikzcd}[column sep=normal]
& \begin{tabular}{c} \includegraphics[scale=.6]{I3} \end{tabular} \arrow[rightarrow]{r}{\displaystyle{a_1^2+4\widetilde{a}_2 t=0}}    \arrow[rightarrow, near end]{rd}{\displaystyle{T}}  & \begin{tabular}{c} \includegraphics[scale=1.2]{IV} \end{tabular} &  \\
\begin{tabular}{c} \includegraphics[trim=0cm  1cm 0 0cm,scale=1]{I2} \end{tabular}\arrow[rightarrow]{ru}{\displaystyle{P_1=0}}  \arrow[rightarrow]{r} {\displaystyle{a_1^2+4\widetilde{a}_2 t=0}}  \arrow[rightarrow]{rd} {\displaystyle{T}} &  \raisebox{-.5cm}{ \scalebox{1.4}{ \includegraphics{III}} } \arrow[rightarrow]{ru} {\displaystyle{P_1=0}}  \arrow[rightarrow]{rd}   {\displaystyle{T}}   &\includegraphics[trim=-0.5cm  -0.5cm 0 0cm,scale=0.6]{I6}  &  \\
 & \raisebox{-.5cm}{ \scalebox{.6}{ \includegraphics{N_I5}} } \arrow[rightarrow]{r} {\displaystyle{a_1=0}} \arrow[rightarrow, near end]{ru} {\displaystyle{P_3=0}}
 & 
  \begin{tabular}{c} \includegraphics[trim=-2.5cm  0.5cm 0 0cm,scale=.6]{Istar2b} 
  \end{tabular} &   \\
\begin{tabular}{c} \includegraphics[scale=.6]{I3}\end{tabular} \arrow[rightarrow]{ru}  {\displaystyle{S}} \arrow[rightarrow]{r} {\displaystyle{a_1=0}} \arrow[rightarrow]{rd} {\displaystyle{P_2=0}} &  \begin{tabular}{c}\includegraphics[scale=1.2]{IV} \end{tabular}\arrow[rightarrow]{rd} {\displaystyle{P_2=0}}  \arrow[rightarrow]{ru} {\displaystyle{S}}  & &  \\
& \raisebox{-.5cm}{ \scalebox{.6}{ \includegraphics{I4}} }\arrow[rightarrow]{r} {\displaystyle{a_1=0}}   \arrow[rightarrow, near end]{ruuu} {\displaystyle{S}}  &     \begin{tabular}{c}\includegraphics[scale=.6]{Istar0}\end{tabular}  & 
\end{tikzcd} 
\caption{I$_2^{\text{ns}}$ + I$_3^{\text{s}}$ Resolution II.  $P_1=-2 \widetilde{a}_4^2 t+4 \widetilde{a}_2 \widetilde{a}_6 t+a_1^2 \widetilde{a}_6 -2 \widetilde{a}_2 \widetilde{a}_3^2$, $P_2=\widetilde{a}_3^3 s-a_1\widetilde{a}_2\widetilde{a}_3^2+a_1^2\widetilde{a}_3\widetilde{a}_4-a_1^3\widetilde{a}_6$, and $P_3=\widetilde{a}_2\widetilde{a}_3^2 s -a_1(a_1\widetilde{a}_6-\widetilde{a}_3\widetilde{a}_4)$.  The non-Kodaira fiber is a contraction of an I$_1^*$.}
\end{table}
\clearpage 

\begin{table}
\begin{tikzcd}[column sep=normal]
& \begin{tabular}{c} \includegraphics[scale=.6]{I3} \end{tabular} \arrow[rightarrow]{r}{\displaystyle{a_1^2+4\widetilde{a}_2 t=0}}    \arrow[rightarrow, near end]{rd}{\displaystyle{T}}  & \begin{tabular}{c} \includegraphics[scale=1.2]{IV} \end{tabular} &  \\
\begin{tabular}{c} \includegraphics[trim=0cm  1cm 0 0cm,scale=1]{I2} \end{tabular}\arrow[rightarrow]{ru}{\displaystyle{P_1=0}}  \arrow[rightarrow]{r} {\displaystyle{a_1^2+4\widetilde{a}_2 t=0}}  \arrow[rightarrow]{rd} {\displaystyle{T}} &  \raisebox{-.5cm}{ \scalebox{1.4}{ \includegraphics{III}} } \arrow[rightarrow]{ru} {\displaystyle{P_1=0}}  \arrow[rightarrow]{rd}   {\displaystyle{T}}   &\includegraphics[trim=-0.5cm  -0.5cm 0 0cm,scale=0.6]{I6}  &  \\
 & \raisebox{-.5cm}{ \scalebox{.6}{ \includegraphics{N_I5}} } \arrow[rightarrow]{r} {\displaystyle{a_1=0}} \arrow[rightarrow, near end]{ru} {\displaystyle{P_3=0}}
 & 
  \begin{tabular}{c} \includegraphics[trim=-2.5cm  0.5cm 0 0cm,scale=.6]{Istar2bflipped} 
  \end{tabular} &   \\
\begin{tabular}{c} \includegraphics[scale=.6]{I3}\end{tabular} \arrow[rightarrow]{ru}  {\displaystyle{S}} \arrow[rightarrow]{r} {\displaystyle{a_1=0}} \arrow[rightarrow]{rd} {\displaystyle{P_2=0}} &  \begin{tabular}{c}\includegraphics[scale=1.2]{IV} \end{tabular}\arrow[rightarrow]{rd} {\displaystyle{P_2=0}}  \arrow[rightarrow]{ru} {\displaystyle{S}}  & &  \\
& \raisebox{-.5cm}{ \scalebox{.6}{ \includegraphics{I4}} }\arrow[rightarrow]{r} {\displaystyle{a_1=0}}   \arrow[rightarrow, near end]{ruuu} {\displaystyle{S}}  &     \begin{tabular}{c}\includegraphics[scale=.6]{Istar0}\end{tabular}  & 
\end{tikzcd} 
\caption{I$_2^{\text{ns}}$ + I$_3^{\text{s}}$ Resolution III.  $P_1=-2 \widetilde{a}_4^2 t+4 \widetilde{a}_2 \widetilde{a}_6 t+a_1^2 \widetilde{a}_6 -2 \widetilde{a}_2 \widetilde{a}_3^2$, $P_2=\widetilde{a}_3^3 s-a_1\widetilde{a}_2\widetilde{a}_3^2+a_1^2\widetilde{a}_3\widetilde{a}_4-a_1^3\widetilde{a}_6$, and $P_3=\widetilde{a}_2\widetilde{a}_3^2 s -a_1(a_1\widetilde{a}_6-\widetilde{a}_3\widetilde{a}_4)$.  The non-Kodaira fiber is a contraction of an I$_1^*$.}
\end{table}
\clearpage

\begin{table}
\begin{tikzcd}[column sep=normal]
& \raisebox{-1cm}{\includegraphics[scale=1.4]{IV}} \arrow[rightarrow]{r}{\displaystyle{P_1=0}} \arrow[rightarrow]{dd} {\displaystyle{T}} & \begin{tabular}{c} \includegraphics[scale=1.0]{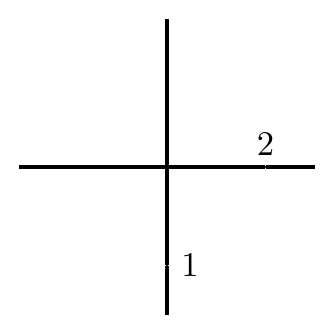} \end{tabular}  &  \\
\begin{tabular}{c} \includegraphics[trim=0cm  1cm 0 0cm,scale=1.4]{III} \end{tabular}  \arrow[rightarrow]{ru}{\displaystyle{\widetilde{a}_4=0}} \arrow[rightarrow]{rd} {\displaystyle{T}}&   &    &  \\
 & \raisebox{-.5cm}{ \scalebox{.6}{ \includegraphics{Tstar12-12-1}} } \arrow[rightarrow]{r} {\displaystyle{\widetilde{a}_3=0}}
 & 
  \begin{tabular}{c} \includegraphics[trim=0cm  0cm 0 0cm,scale=.6]{T12321} 
  \end{tabular} &   \\
\begin{tabular}{c} \includegraphics[scale=.6]{I3}\end{tabular} \arrow[rightarrow]{ru} {\displaystyle{S}} \arrow[rightarrow]{r} {\displaystyle{a_1=0}} \arrow[rightarrow]{rd}  {\displaystyle{P_2=0}} &  \begin{tabular}{c}\includegraphics[scale=1.2]{IV}  \end{tabular}
 \arrow[rightarrow]{u}{\displaystyle{S}}  
\arrow[rightarrow]{rd} {\displaystyle{P_2=0}} 
   & &  \\
& \raisebox{-.5cm}{ \scalebox{.6}{ \includegraphics{I4}} }\arrow[rightarrow]{r} {\displaystyle{a_1=0}} &     \begin{tabular}{c}\includegraphics[scale=.6]{Istar0}\end{tabular} %
& 
\end{tikzcd} 
\caption{III + I$_3^{\text{s}}$ Resolution I. $P_1=\widetilde{a}_3^2+4\widetilde{a}_6 t$ and $P_2=\widetilde{a}_3^3-\widetilde{a}_1\widetilde{a}_2\widetilde{a}_3^2 s+\widetilde{a}_1^2\widetilde{a}_3\widetilde{a}_4 s-\widetilde{a}_1^3\widetilde{a}_6 s^2$}
\end{table}
\clearpage 

\begin{table}
\begin{tikzcd}[column sep=normal]
& \raisebox{-1cm}{\includegraphics[scale=1.4]{IV}} \arrow[rightarrow]{r}{\displaystyle{P_1=0}} \arrow[rightarrow]{dd} {\displaystyle{T}} & \begin{tabular}{c} \includegraphics[scale=1.0]{T12} \end{tabular}  &  \\
\begin{tabular}{c} \includegraphics[trim=0cm  1cm 0 0cm,scale=1.4]{III} \end{tabular}  \arrow[rightarrow]{ru}{\displaystyle{\widetilde{a}_4=0}} \arrow[rightarrow]{rd} {\displaystyle{T}}&   &    &  \\
 & \raisebox{-.5cm}{ \scalebox{.6}{ \includegraphics{Tstar12-12-1}} } \arrow[rightarrow]{r} {\displaystyle{\widetilde{a}_3=0}}
 & 
  \begin{tabular}{c} \includegraphics[trim=0cm  0cm 0 0cm,scale=.6]{Tstar123-2-1} 
  \end{tabular} &   \\
\begin{tabular}{c} \includegraphics[scale=.6]{I3}\end{tabular} \arrow[rightarrow]{ru} {\displaystyle{S}} \arrow[rightarrow]{r} {\displaystyle{a_1=0}} \arrow[rightarrow]{rd}  {\displaystyle{P_2=0}} &  \begin{tabular}{c}\includegraphics[scale=1.2]{IV}  \end{tabular}
 \arrow[rightarrow]{u}{\displaystyle{S}}  
\arrow[rightarrow]{rd} {\displaystyle{P_2=0}} 
   & &  \\
& \raisebox{-.5cm}{ \scalebox{.6}{ \includegraphics{I4}} }\arrow[rightarrow]{r} {\displaystyle{a_1=0}} &     \begin{tabular}{c}\includegraphics[scale=.6]{Istar0}\end{tabular} & 
\end{tikzcd} 
\caption{III + I$_3^{\text{s}}$ Resolution II. $P_1=\widetilde{a}_3^2+4\widetilde{a}_6 t$ and $P_2=\widetilde{a}_3^3-\widetilde{a}_1\widetilde{a}_2\widetilde{a}_3^2 s+\widetilde{a}_1^2\widetilde{a}_3\widetilde{a}_4 s-\widetilde{a}_1^3\widetilde{a}_6 s^2$}
\end{table}
\clearpage

\begin{table}
\begin{tikzcd}[column sep=normal]
& \raisebox{-1cm}{\includegraphics[scale=1.4]{IV}} \arrow[rightarrow]{r}{\displaystyle{P_1=0}} \arrow[rightarrow]{dd} {\displaystyle{T}} & \begin{tabular}{c} \includegraphics[scale=1.0]{T12} \end{tabular}  &  \\
\begin{tabular}{c} \includegraphics[trim=0cm  1cm 0 0cm,scale=1.4]{III} \end{tabular}  \arrow[rightarrow]{ru}{\displaystyle{\widetilde{a}_4=0}} \arrow[rightarrow]{rd} {\displaystyle{T}}&   &    &  \\
 & \raisebox{-.5cm}{ \scalebox{.6}{ \includegraphics{Tstar12-12-1}} } \arrow[rightarrow]{r} {\displaystyle{\widetilde{a}_3=0}}
 & 
  \begin{tabular}{c} \includegraphics[trim=0cm  0cm 0 0cm,scale=.6]{Tstar12-12-2} 
  \end{tabular} &   \\
\begin{tabular}{c} \includegraphics[scale=.6]{I3}\end{tabular} \arrow[rightarrow]{ru} {\displaystyle{S}} \arrow[rightarrow]{r} {\displaystyle{a_1=0}} \arrow[rightarrow]{rd}  {\displaystyle{P_2=0}} &  \begin{tabular}{c}\includegraphics[scale=1.2]{IV}  \end{tabular}
 \arrow[rightarrow]{u}{\displaystyle{S}}  
\arrow[rightarrow]{rd} {\displaystyle{P_2=0}}
   & &  \\
& \raisebox{-.5cm}{ \scalebox{.6}{ \includegraphics{I4}} }\arrow[rightarrow]{r} {\displaystyle{a_1=0}} &     \begin{tabular}{c}\includegraphics[scale=.6]{Istar0}\end{tabular}   & 
\end{tikzcd} 
\caption{III + I$_3^{\text{s}}$ Resolution III. $P_1=\widetilde{a}_3^2+4\widetilde{a}_6 t$ and $P_2=\widetilde{a}_3^3-\widetilde{a}_1\widetilde{a}_2\widetilde{a}_3^2 s+\widetilde{a}_1^2\widetilde{a}_3\widetilde{a}_4 s-\widetilde{a}_1^3\widetilde{a}_6 s^2$}
\end{table}
\clearpage

\begin{table}
\begin{tikzcd}[column sep=normal]
& \raisebox{-1cm}{\includegraphics[scale=1.4]{IV}} \arrow[rightarrow]{r}{\displaystyle{P_1=0}} \arrow[rightarrow]{dd} {\displaystyle{T}} & \begin{tabular}{c} \includegraphics[scale=1.0]{T12} \end{tabular}  &  \\
\begin{tabular}{c} \includegraphics[trim=0cm  1cm 0 0cm,scale=1.4]{III} \end{tabular}  \arrow[rightarrow]{ru}{\displaystyle{\widetilde{a}_4=0}} \arrow[rightarrow]{rd} {\displaystyle{T}}&   &    &  \\
 & \raisebox{-.5cm}{ \scalebox{.6}{ \includegraphics{T12321}} } \arrow[rightarrow]{r} {\displaystyle{\widetilde{a}_3=0}}
 & 
  \begin{tabular}{c} \includegraphics[trim=0cm  0cm 0 0cm,scale=.6]{Tstar123-2-1} 
  \end{tabular} &   \\
\begin{tabular}{c} \includegraphics[scale=.6]{I3}\end{tabular} \arrow[rightarrow]{ru} {\displaystyle{S}} \arrow[rightarrow]{r} {\displaystyle{a_1=0}} \arrow[rightarrow]{rd}  {\displaystyle{P_2=0}} &  \begin{tabular}{c}\includegraphics[scale=1.2]{IV}  \end{tabular}
 \arrow[rightarrow]{u}{\displaystyle{S}}  
\arrow[rightarrow]{rd} {\displaystyle{P_2=0}} 
   & &  \\
& \raisebox{-.5cm}{ \scalebox{.6}{ \includegraphics{I4}} }\arrow[rightarrow]{r} {\displaystyle{a_1=0}} &     \begin{tabular}{c}\includegraphics[scale=.6]{Istar0}\end{tabular} & 
\end{tikzcd} 
\caption{III + I$_3^{\text{s}}$ Resolution IV. $P_1=\widetilde{a}_3^2+4\widetilde{a}_6 t$ and $P_2=\widetilde{a}_3^3-\widetilde{a}_1\widetilde{a}_2\widetilde{a}_3^2 s+\widetilde{a}_1^2\widetilde{a}_3\widetilde{a}_4 s-\widetilde{a}_1^3\widetilde{a}_6 s^2$}
\end{table}
\clearpage

 \newpage 
\begin{table}
\begin{center}
\begin{tikzcd}[column sep=normal]
& \begin{tabular}{c} \includegraphics[trim=0cm  1cm 0 0cm,scale=.6]{I3} \end{tabular}   \arrow[rightarrow]{r}{\displaystyle{a_1=0}}& \begin{tabular}{c} \includegraphics[trim=0cm  1cm 0 0cm,scale=1.2]{IV} \end{tabular} \\
\begin{tabular}{c} \includegraphics[trim=0cm  1cm 0 0cm,scale=.8]{I2} \end{tabular}   \arrow[rightarrow]{ru}{\displaystyle{P_1=0}} \arrow[rightarrow]{rd}{\displaystyle{T}}   \arrow[rightarrow]{r}{\displaystyle{a_1=0}} & 
\begin{tabular}{c} \includegraphics[trim=0cm  1cm 0 0cm,scale=1.2]{III} \end{tabular} 
  \arrow[rightarrow]{ru}{\displaystyle{P_1=0}} & \begin{tabular}{c} \includegraphics[trim=0cm  1cm 0 0cm,scale=.6]{Tstar12-12-1} \end{tabular}\\
& 
\raisebox{-.5cm}{
\begin{tabular}{c} \includegraphics[trim=0cm  1cm 0 0cm,scale=.6]{T122-1-1} \end{tabular} }
  \arrow[rightarrow]{rd}  {\displaystyle{\widetilde{a}_3=0}}
\arrow[rightarrow]{ru} {\displaystyle{\widetilde{a}_2=0}} & \\
\begin{tabular}{c} \includegraphics[trim=0cm  1cm 0 0cm,scale=1.4]{IV} \end{tabular}  \arrow[rightarrow]{ru}  {\displaystyle{S}} \arrow[rightarrow]{rd} {\displaystyle{\widetilde{a}_3=0}} &  & 
\begin{tabular}{c} \includegraphics[trim=0cm  1cm 0 0cm,scale=.6]{N_NK7} \end{tabular}
 \\
& 
\raisebox{-.5cm}{ \scalebox{.6}{ \includegraphics{Istar0}} }
  \arrow[rightarrow]{r}{\displaystyle{P_2=0}} &   \begin{tabular}{c}\includegraphics[scale=.7]{T12-1-2}\end{tabular} 
\end{tikzcd} 
\end{center}
\caption{I$_2^{\text{ns}}+$IV$^{\text{s}}$ or I$_2^{\text{s}}+$IV$^{\text{s}}$, Resolution I. $P_1=\widetilde{a}_4^2-\widetilde{a}_1^2 \widetilde{a}_6 t$ and $P_2=\widetilde{a}_2^2 \widetilde{a}_4^2-4\widetilde{a}_4^3 s-4\widetilde{a}_2^3\widetilde{a}_6+18\widetilde{a}_2\widetilde{a}_4\widetilde{a}_6 s-27\widetilde{a}_6^2 s^2$ for I$_2^{\text{ns}}+$IV$^{\text{s}}$, $P_2=\widetilde{a}_2^2 \widetilde{a}_4^2 s -4\widetilde{a}_4^3-4\widetilde{a}_2^3\widetilde{a}_6 s^2 +18\widetilde{a}_2\widetilde{a}_4\widetilde{a}_6 s^2 -27\widetilde{a}_6^2 s $ for I$_2^{\text{s}}+$IV$^{\text{s}}$.}
\end{table}
\clearpage 

\newpage 
\begin{table}
\begin{center}
\begin{tikzcd}[column sep=normal]
& \begin{tabular}{c} \includegraphics[trim=0cm  1cm 0 0cm,scale=.6]{I3} \end{tabular}   \arrow[rightarrow]{r}{\displaystyle{a_1=0}} & \begin{tabular}{c} \includegraphics[trim=0cm  1cm 0 0cm,scale=1.2]{IV} \end{tabular} \\
\begin{tabular}{c} \includegraphics[trim=0cm  1cm 0 0cm,scale=.8]{I2} \end{tabular}   \arrow[rightarrow]{ru}{\displaystyle{P_1=0}} \arrow[rightarrow]{rd}{\displaystyle{T}}   \arrow[rightarrow]{r}{\displaystyle{a_1=0}} & 
\begin{tabular}{c} \includegraphics[trim=0cm  1cm 0 0cm,scale=1.2]{III} \end{tabular} 
  \arrow[rightarrow]{ru}{\displaystyle{P_1=0}} & \begin{tabular}{c} \includegraphics[trim=0cm  0cm 0 0cm,scale=.6]{Tstar12-12-1} \end{tabular}\\
& 
\begin{tabular}{c} \raisebox{.5cm}{\includegraphics[trim=0cm  1cm 0 0cm,scale=.6]{Istar2b}} \end{tabular} 
  \arrow[rightarrow]{rd}  {\displaystyle{\widetilde{a}_3=0}}
\arrow[rightarrow]{ru} {\displaystyle{\widetilde{a}_2=0}} & \\
\begin{tabular}{c} \includegraphics[trim=0cm  1cm 0 0cm,scale=1.4]{IV} \end{tabular}  \arrow[rightarrow]{ru}  {\displaystyle{S}} \arrow[rightarrow]{rd} {\displaystyle{\widetilde{a}_3=0}} &  & 
\begin{tabular}{c} \includegraphics[trim=0cm  1cm 0 0cm,scale=.6]{N_NK12} \end{tabular}
 \\
& 
\raisebox{-.5cm}{ \scalebox{.6}{ \includegraphics{Istar0}} }
  \arrow[rightarrow]{r}{\displaystyle{P_2=0}} &   \begin{tabular}{c}\includegraphics[scale=.7]{T12-1-2}\end{tabular} 
\end{tikzcd} 
\end{center}
\caption{I$_2^{\text{s}}+$IV$^{\text{s}}$ or I$_2^{\text{ns}}+$IV$^{\text{s}}$, Resolution II. {$P_1=\widetilde{a}_4^2 t -a_1^2 \widetilde{a}_6$ and $P_2=\widetilde{a}_2^2 \widetilde{a}_4^2-4\widetilde{a}_4^3 s-4\widetilde{a}_2^3\widetilde{a}_6+18\widetilde{a}_2\widetilde{a}_4\widetilde{a}_6 s-27\widetilde{a}_6^2 s^2$ for I$_2^{\text{ns}}+$IV$^{\text{s}}$, $P_2=\widetilde{a}_2^2 \widetilde{a}_4^2 s -4\widetilde{a}_4^3-4\widetilde{a}_2^3\widetilde{a}_6 s^2 +18\widetilde{a}_2\widetilde{a}_4\widetilde{a}_6 s^2 -27\widetilde{a}_6^2 s $ for I$_2^{\text{s}}+$IV$^{\text{s}}$.}}
\end{table}

\newpage 
\begin{table}
\begin{center}
\begin{tikzcd}[column sep=normal]
& \begin{tabular}{c} \includegraphics[trim=0cm  1cm 0 0cm,scale=.6]{I3} \end{tabular}   \arrow[rightarrow]{r}{\displaystyle{a_1=0}} & \begin{tabular}{c} \includegraphics[trim=0cm  1cm 0 0cm,scale=1.2]{IV} \end{tabular} \\
\begin{tabular}{c} \includegraphics[trim=0cm  1cm 0 0cm,scale=.8]{I2} \end{tabular}   \arrow[rightarrow]{ru}{\displaystyle{P_1=0}} \arrow[rightarrow]{rd}{\displaystyle{T}}   \arrow[rightarrow]{r}{\displaystyle{a_1=0}} & 
\begin{tabular}{c} \includegraphics[trim=0cm  1cm 0 0cm,scale=1.2]{III} \end{tabular} 
  \arrow[rightarrow]{ru}{\displaystyle{P_1=0}} & \begin{tabular}{c} \includegraphics[trim=-0.2cm  0.8cm 0 0cm,scale=.5]{Tstar12-12-1} \end{tabular}\\
& 
\raisebox{-.5cm}{
\begin{tabular}{c} \includegraphics[trim=0cm  1cm 0 0cm,scale=.6]{Istar2bflipped} \end{tabular} }
  \arrow[rightarrow]{rd}  {\displaystyle{\widetilde{a}_3=0}}
\arrow[rightarrow]{ru} {\displaystyle{\widetilde{a}_2=0}} & \\
\begin{tabular}{c} \includegraphics[trim=0cm  1cm 0 0cm,scale=1.4]{IV} \end{tabular}  \arrow[rightarrow]{ru}  {\displaystyle{S}} \arrow[rightarrow]{rd} {\displaystyle{\widetilde{a}_3=0}} &  & 
\begin{tabular}{c} \includegraphics[trim=0cm  1cm 0 0cm,scale=.6]{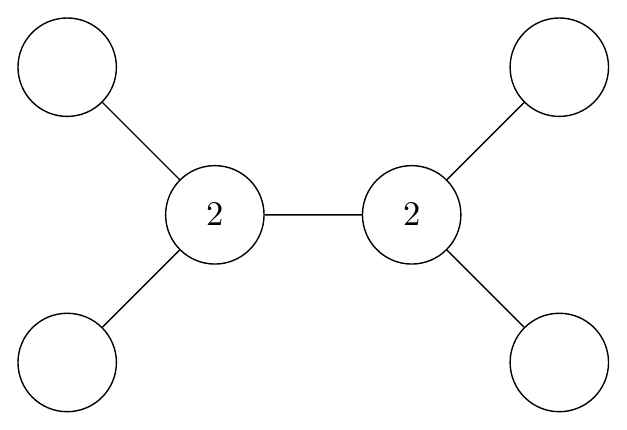} \end{tabular}
 \\
& 
\raisebox{-.5cm}{ \scalebox{.6}{ \includegraphics{Istar0}} }
  \arrow[rightarrow]{r}{\displaystyle{P_2=0}} &   \begin{tabular}{c}\includegraphics[scale=.7]{T12-1-2}\end{tabular} 
\end{tikzcd} 
\end{center}
\caption{I$_2^{\text{ns}}+$IV$^{\text{s}}$ or I$_2^{\text{s}}+$IV$^{\text{s}}$, Resolution III. {$P_1=\widetilde{a}_4^2 t -a_1^2 \widetilde{a}_6$ and $P_2=\widetilde{a}_2^2 \widetilde{a}_4^2-4\widetilde{a}_4^3 s-4\widetilde{a}_2^3\widetilde{a}_6+18\widetilde{a}_2\widetilde{a}_4\widetilde{a}_6 s-27\widetilde{a}_6^2 s^2$ for I$_2^{\text{ns}}+$IV$^{\text{s}}$, $P_2=\widetilde{a}_2^2 \widetilde{a}_4^2 s -4\widetilde{a}_4^3-4\widetilde{a}_2^3\widetilde{a}_6 s^2 +18\widetilde{a}_2\widetilde{a}_4\widetilde{a}_6 s^2 -27\widetilde{a}_6^2 s $ for I$_2^{\text{s}}+$IV$^{\text{s}}$.}}
\end{table}
\clearpage

\newpage 
\begin{table}
\begin{center}
\begin{tikzcd}[column sep=normal]
& \begin{tabular}{c} \includegraphics[trim=0cm  1cm 0 0cm,scale=.6]{I3} \end{tabular}   \arrow[rightarrow]{r}{\displaystyle{a_1=0}} & \begin{tabular}{c} \includegraphics[trim=0cm  1cm 0 0cm,scale=1.2]{IV} \end{tabular} \\
\begin{tabular}{c} \includegraphics[trim=0cm  1cm 0 0cm,scale=.8]{I2} \end{tabular}   \arrow[rightarrow]{ru}{\displaystyle{P_1=0}} \arrow[rightarrow]{rd}{\displaystyle{T}}   \arrow[rightarrow]{r}{\displaystyle{a_1=0}} & 
\begin{tabular}{c} \includegraphics[trim=0cm  1cm 0 0cm,scale=1.2]{III} \end{tabular} 
  \arrow[rightarrow]{ru}{\displaystyle{P_1=0}} & \begin{tabular}{c} \includegraphics[trim=0cm  0cm 0 0cm,scale=.6]{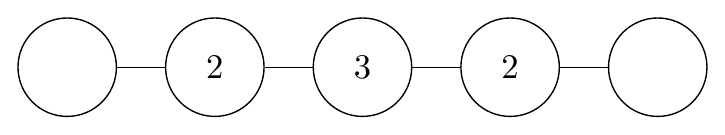} \end{tabular}\\
& 
\begin{tabular}{c} \raisebox{.5cm}{\includegraphics[trim=0cm  1cm 0 0cm,scale=.6]{N_NK10}} \end{tabular} 
  \arrow[rightarrow]{rd}  {\displaystyle{\widetilde{a}_3=0}}
\arrow[rightarrow]{ru} {\displaystyle{\widetilde{a}_2=0}} & \\
\begin{tabular}{c} \includegraphics[trim=0cm  1cm 0 0cm,scale=1.4]{IV} \end{tabular}  \arrow[rightarrow]{ru}  {\displaystyle{S}} \arrow[rightarrow]{rd} {\displaystyle{\widetilde{a}_3=0}} &  & 
\begin{tabular}{c} \includegraphics[trim=0cm  1cm 0 0cm,scale=.6]{N_NK8} \end{tabular}
 \\
& 
\raisebox{-.5cm}{ \scalebox{.6}{ \includegraphics{Istar0}} }
  \arrow[rightarrow]{r}{\displaystyle{P_2=0}} &   \begin{tabular}{c}\includegraphics[scale=.7]{T12-1-2}\end{tabular} 
\end{tikzcd} 
\end{center}
\caption{I$_2^{\text{s}}+$IV$^{\text{s}}$ or I$_2^{\text{ns}}+$IV$^{\text{s}}$, Resolution IV.{ $P_1=\widetilde{a}_4^2 t -a_1^2 \widetilde{a}_6$ and $P_2=\widetilde{a}_2^2 \widetilde{a}_4^2-4\widetilde{a}_4^3 s-4\widetilde{a}_2^3\widetilde{a}_6+18\widetilde{a}_2\widetilde{a}_4\widetilde{a}_6 s-27\widetilde{a}_6^2 s^2$ for I$_2^{\text{ns}}+$IV$^{\text{s}}$, $P_2=\widetilde{a}_2^2 \widetilde{a}_4^2 s -4\widetilde{a}_4^3-4\widetilde{a}_2^3\widetilde{a}_6 s^2 +18\widetilde{a}_2\widetilde{a}_4\widetilde{a}_6 s^2 -27\widetilde{a}_6^2 s $ for I$_2^{\text{s}}+$IV$^{\text{s}}$.}}
\end{table}
\clearpage 

 \newpage 
\begin{table}
\begin{tikzcd}[column sep=normal]
& \raisebox{-1cm}{\includegraphics[scale=1.4]{IV}} \arrow[rightarrow]{r}{\displaystyle{P_1=0}} \arrow[rightarrow]{dd} {\displaystyle{T}} & \begin{tabular}{c} \includegraphics[scale=1.0]{T12} \end{tabular}  &  \\
\begin{tabular}{c} \includegraphics[trim=0cm  1cm 0 0cm,scale=1.4]{III} \end{tabular}  \arrow[rightarrow]{ru}{\displaystyle{\widetilde{a}_4=0}} \arrow[rightarrow]{rd} {\displaystyle{T}}&   &    &  \\
 &\begin{tabular}{c} \includegraphics[scale=.5]{Tstar12-12-1} \end{tabular} \arrow[rightarrow]{r} {\displaystyle{\widetilde{a}_3=0}}
 & 
  \begin{tabular}{c} \includegraphics[trim=0cm  0cm 0 0cm,scale=.6]{T12321} \end{tabular} &   \\
\begin{tabular}{c} \includegraphics[scale=1.4]{IV}\end{tabular} \arrow[rightarrow]{ru}  {\displaystyle{S}}   \arrow[rightarrow]{rd} {\displaystyle{\widetilde{a}_3=0}} &     & &  \\
& \raisebox{-.5cm}{ \scalebox{.6}{ \includegraphics{Istar0}} }\arrow[rightarrow]{r} {\displaystyle{P_2=0}} \arrow[rightarrow]{ruu}  {\displaystyle{S}} &     \begin{tabular}{c}\includegraphics[scale=.6]{T12-1-2}\end{tabular}  & 
\end{tikzcd} 
\caption{ III + IV$^{\text{s}}$,  Resolution I.  $P_1=\widetilde{a}_3^2+4\widetilde{a}_6 t$ and $P_2=\widetilde{a}_2^2 \widetilde{a}_4^2 s -4\widetilde{a}_4^3-4\widetilde{a}_2^3\widetilde{a}_6 s^2 +18\widetilde{a}_2\widetilde{a}_4\widetilde{a}_6 s -27\widetilde{a}_6^2 s$.}
\end{table}

\clearpage 
\begin{table}
\begin{tikzcd}[column sep=normal]
& \raisebox{-1cm}{\includegraphics[scale=1.4]{IV}} \arrow[rightarrow]{r}{\displaystyle{P_1=0}} \arrow[rightarrow]{dd} {\displaystyle{T}} & \begin{tabular}{c} \includegraphics[scale=1.0]{T12} \end{tabular}  &  \\
\begin{tabular}{c} \includegraphics[trim=0cm  1cm 0 0cm,scale=1.4]{III} \end{tabular}  \arrow[rightarrow]{ru}{\displaystyle{\widetilde{a}_4=0}} \arrow[rightarrow]{rd} {\displaystyle{T}}&   &    &  \\
 & \raisebox{-.5cm}{ \scalebox{.6}{ \includegraphics{Tstar12-12-1}} } \arrow[rightarrow]{r} {\displaystyle{\widetilde{a}_3=0}}
 & 
  \begin{tabular}{c} \includegraphics[trim=0cm  0cm 0 0cm,scale=.6]{Tstar123-2-1} 
  \end{tabular} &   \\
\begin{tabular}{c} \includegraphics[scale=1.4]{IV}\end{tabular} \arrow[rightarrow]{ru}  {\displaystyle{S}}   \arrow[rightarrow]{rd} {\displaystyle{\widetilde{a}_3=0}} &     & &  \\
& \raisebox{-.5cm}{ \scalebox{.6}{ \includegraphics{Istar0}} }\arrow[rightarrow]{r} {\displaystyle{P_2=0}} \arrow[rightarrow]{ruu}  {\displaystyle{S}} &     \begin{tabular}{c}\includegraphics[scale=.6]{T12-1-2}\end{tabular}  & 
\end{tikzcd} 
\caption{ III + IV$^{\text{s}}$ Resolution II. $P_1=\widetilde{a}_3^2+4\widetilde{a}_6 t$ and $P_2=\widetilde{a}_2^2 \widetilde{a}_4^2 s -4\widetilde{a}_4^3-4\widetilde{a}_2^3\widetilde{a}_6 s^2 +18\widetilde{a}_2\widetilde{a}_4\widetilde{a}_6 s -27\widetilde{a}_6^2 s$.}
\end{table}

 \newpage 
\begin{table}
\begin{tikzcd}[column sep=normal]
& \raisebox{-1cm}{\includegraphics[scale=1.4]{IV}} \arrow[rightarrow]{r}{\displaystyle{P_1=0}} \arrow[rightarrow]{dd} {\displaystyle{T}} & \begin{tabular}{c} \includegraphics[scale=1.0]{T12} \end{tabular}  &  \\
\begin{tabular}{c} \includegraphics[trim=0cm  1cm 0 0cm,scale=1.4]{III} \end{tabular}  \arrow[rightarrow]{ru}{\displaystyle{\widetilde{a}_4=0}} \arrow[rightarrow]{rd} {\displaystyle{T}}&   &    &  \\
 &\begin{tabular}{c} \includegraphics[scale=.6]{Tstar12-12-1} \end{tabular} \arrow[rightarrow]{r} {\displaystyle{\widetilde{a}_3=0}}
 & 
  \begin{tabular}{c} \includegraphics[trim=0cm  0cm 0 0cm,scale=.6]{Tstar12-12-2}\end{tabular} &   \\
\begin{tabular}{c} \includegraphics[scale=1.4]{IV}\end{tabular} \arrow[rightarrow]{ru}  {\displaystyle{S}}   \arrow[rightarrow]{rd} {\displaystyle{\widetilde{a}_3=0}} &     & &  \\
& \raisebox{-.5cm}{ \scalebox{.6}{ \includegraphics{Istar0}} }\arrow[rightarrow]{r} {\displaystyle{P_2=0}} \arrow[rightarrow]{ruu}  {\displaystyle{S}} &     \begin{tabular}{c}\includegraphics[scale=.6]{T12-1-2}\end{tabular}  & 
\end{tikzcd} 
\caption{ III + IV$^{\text{s}}$,  Resolution III.  $P_1=\widetilde{a}_3^2+4\widetilde{a}_6 t$ and $P_2=\widetilde{a}_2^2 \widetilde{a}_4^2 s -4\widetilde{a}_4^3-4\widetilde{a}_2^3\widetilde{a}_6 s^2 +18\widetilde{a}_2\widetilde{a}_4\widetilde{a}_6 s -27\widetilde{a}_6^2 s$. 
The non-Kodaira fiber in codimension-two is a contraction of a IV$^*$ and its specialization in codimension-three is a contraction of a III$^*$.}
\end{table}
\clearpage

 \newpage 
\begin{table}
\begin{tikzcd}[column sep=normal]
& \raisebox{-1cm}{\includegraphics[scale=1.4]{IV}} \arrow[rightarrow]{r}{\displaystyle{P_1=0}} \arrow[rightarrow]{dd} {\displaystyle{T}} & \begin{tabular}{c} \includegraphics[scale=1.0]{T12} \end{tabular}  &  \\
\begin{tabular}{c} \includegraphics[trim=0cm  1cm 0 0cm,scale=1.4]{III} \end{tabular}  \arrow[rightarrow]{ru}{\displaystyle{\widetilde{a}_4=0}} \arrow[rightarrow]{rd} {\displaystyle{T}}&   &    &  \\
 &\begin{tabular}{c} \includegraphics[scale=.6]{Tstar12-12-1} \end{tabular} \arrow[rightarrow]{r} {\displaystyle{\widetilde{a}_3=0}}
 & 
  \begin{tabular}{c} \includegraphics[trim=0cm  0cm 0 0cm,scale=.6]{T12321}\end{tabular} &   \\
\begin{tabular}{c} \includegraphics[scale=1.4]{IV}\end{tabular} \arrow[rightarrow]{ru}  {\displaystyle{S}}   \arrow[rightarrow]{rd} {\displaystyle{\widetilde{a}_3=0}} &     & &  \\
& \raisebox{-.5cm}{ \scalebox{.6}{ \includegraphics{Istar0}} }\arrow[rightarrow]{r} {\displaystyle{P_2=0}} \arrow[rightarrow]{ruu}  {\displaystyle{S}} &     \begin{tabular}{c}\includegraphics[scale=.6]{T12-1-2}\end{tabular}  & 
\end{tikzcd} 
\caption{ III + IV$^{\text{s}}$,  Resolution IV.  $P_1=\widetilde{a}_3^2+4\widetilde{a}_6 t$ and $P_2=\widetilde{a}_2^2 \widetilde{a}_4^2 s -4\widetilde{a}_4^3-4\widetilde{a}_2^3\widetilde{a}_6 s^2 +18\widetilde{a}_2\widetilde{a}_4\widetilde{a}_6 s -27\widetilde{a}_6^2 s$.  
 The non-Kodaira fiber in codimension-two is a contraction of a IV$^*$ and its specialization in codimension-three is a contraction of a III$^*$.
}
\end{table}
\clearpage

\bibliography{mboyoBib}

\begin{thebibliography}{10}




\bibitem{Aluffi_CBU}
P.~Aluffi.
\newblock Chern classes of blowups.
\newblock {\em Math. Proc. Cambridge Philos. Soc.}, 148(2):227--242, 2010.

\bibitem{AE1}
P.~Aluffi and M.~Esole,
\newblock {Chern class identities from tadpole matching in type IIB and F-theory}.
\newblock {\em JHEP}, 03:032, 2009.

\bibitem{AE2}
P.~Aluffi and M.~Esole,
\newblock {New Orientifold Weak Coupling Limits in F-theory}.
\newblock {\em JHEP}, 02:020, 2010.

\bibitem{Anderson:2017zfm}
  L.~B.~Anderson, M.~Esole, L.~Fredrickson and L.~P.~Schaposnik,
  Singular Geometry and Higgs Bundles in String Theory,
  SIGMA {\bf 14} (2018) 037.


\bibitem{Argyres:1995jj} 
  P.~C.~Argyres and M.~R.~Douglas,
 New phenomena in SU(3) supersymmetric gauge theory,
  Nucl.\ Phys.\ B {\bf 448}, 93 (1995)

\bibitem{Arras:2016evy} 
  P.~Arras, A.~Grassi and T.~Weigand,
 Terminal Singularities, Milnor Numbers, and Matter in F-theory,
  J.\ Geom.\ Phys.\  {\bf 123}, 71 (2018)
 
\bibitem{Avramis:2005hc} 
  S.~D.~Avramis and A.~Kehagias,
 A Systematic search for anomaly-free supergravities in six dimensions,
  JHEP {\bf 0510}, 052 (2005).

  
\bibitem{Baez:2009dj} 
  J.~C.~Baez and J.~Huerta,
  ``The Algebra of Grand Unified Theories,''
  Bull.\ Am.\ Math.\ Soc.\  {\bf 47}, 483 (2010)

\bibitem{Batyrev.Betti}
V.~V. Batyrev.
\newblock Birational {C}alabi-{Y}au {$n$}-folds have equal {B}etti numbers.
\newblock In {\em New trends in algebraic geometry ({W}arwick, 1996)}, volume
  264 of {\em London Math. Soc. Lecture Note Ser.}, pages 1--11. Cambridge
  Univ. Press, Cambridge, 1999.

\bibitem{Bernard} 
  C.~W.~Bernard, N.~H.~Christ, A.~H.~Guth and E.~J.~Weinberg,
   Instanton Parameters for Arbitrary Gauge Groups,
  Phys.\ Rev.\ D {\bf 16}, 2967 (1977).



\bibitem{Bershadsky:1996nh}
M.~Bershadsky, K.~A. Intriligator, S.~Kachru, D.~R. Morrison, V.~Sadov, and
  C.~Vafa,
\newblock {Geometric singularities and enhanced gauge symmetries}.
\newblock {\em Nucl. Phys.}, B481:215--252, 1996.

\bibitem{Bershadsky:1996nu} 
  M.~Bershadsky and A.~Johansen,
  Colliding singularities in F theory and phase transitions,
  Nucl.\ Phys.\ B {\bf 489}, 122 (1997)
  
\bibitem{Bershadsky:1997sb} 
  M.~Bershadsky and C.~Vafa,
  ``Global anomalies and geometric engineering of critical theories in six-dimensions,''
  hep-th/9703167.



\bibitem{BCHM}
 C.~Birkar, P.~Cascini, C.~D.~Hacon and J.~M$^c$Kernan, Existence of minimal models for varieties of log general type, 
J. Amer. Math. Soc. 23 (2010), 405-468 

\bibitem{Cadavid:1995bk} 
  A.~C.~Cadavid, A.~Ceresole, R.~D'Auria and S.~Ferrara,
  Eleven-dimensional supergravity compactified on Calabi--Yau threefolds,
  Phys.\ Lett.\ B {\bf 357}, 76 (1995)
  

\bibitem{Cattaneo}
A.~Cattaneo. Crepant resolutions of Weierstrass threefolds and non-Kodaira fibers.
\href{https://arxiv.org/abs/1307.7997}{\tt  arXiv:1307.7997v2   [math.AG]}.


\bibitem{CDE}
A.~Collinucci, F.~Denef, and M.~Esole.
\newblock {D-brane Deconstructions in IIB Orientifolds}.
\newblock {\em JHEP}, 02:005, 2009.

\bibitem{Deligne.Formulaire}
P.~Deligne, Courbes \'elliptiques: formulaire d'apr\`es J. Tate. (French) {\it Modular functions of one variable}, IV ({\it Proc. Internat. Summer Schoo}l, Univ. Antwerp, Antwerp, 1972), pp. 53--73. Lecture Notes in Math., Vol. 476, Springer, Berlin, 1975.




%
\bibitem{Esole:2017csj} 
  M.~Esole,
{\it  Introduction to Elliptic Fibrations}, 
  In: 
   Cardona A., Morales P., Ocampo H., Paycha S., Reyes Lega A. (eds) 
  Quantization, Geometry and Noncommutative Structures in Mathematics and Physics, pp 247--276,  Mathematical Physics Studies. Springer, Cham, 2017.
   
 
\bibitem{EFY}
M.~Esole, J.~Fullwood, and S.-T. Yau.
\newblock {$D_5$ elliptic fibrations: non-Kodaira fibers and new orientifold
  limits of F-theory}.
\newblock Commun.\ Num.\ Theor.\ Phys.\  {\bf 09}, no. 3, 583 (2015).

\bibitem{EJJN1}
M.~Esole, S.~G. Jackson, R.~Jagadeesan, and A.~G. No{\"e}l.
\newblock {Incidence Geometry in a Weyl Chamber I: GL$_n$}, 
\newblock arXiv:1508.03038 [math.RT].

\bibitem{EJJN2} 
  M.~Esole, S.~G.~Jackson, R.~Jagadeesan and A.~G.~No{\"e}l,
  Incidence Geometry in a Weyl Chamber II: $SL_n$,
  arXiv:1601.05070 [math.RT].

\bibitem{G2} 
  M.~Esole, R.~Jagadeesan and M.~J.~Kang,
  The Geometry of G$_2$, Spin(7), and Spin(8)-models,
  arXiv:1709.04913 [hep-th].

\bibitem{F4} 
  M.~Esole, P.~Jefferson and M.~J.~Kang,
  The Geometry of F$_4$-Models,
  arXiv:1704.08251 [hep-th].

\bibitem{Euler} 
  M.~Esole, P.~Jefferson and M.~J.~Kang,
  Euler Characteristics of Crepant Resolutions of Weierstrass Models,
  arXiv:1703.00905 [math.AG].

\bibitem{EP1} 
  M.~Esole and S.~Pasterski,
  D$_4$-flops of the E$_7$-model,
  arXiv:1901.00093 [hep-th].


\bibitem{CharMW}
  M.~Esole and M.~J.~Kang,
  Characteristic numbers of elliptic fibrations with non-trivial Mordell--Weil groups,
  arXiv:1808.07054 [hep-th].

\bibitem{Char}
  M.~Esole and M.~J.~Kang,
Characteristic numbers of crepant resolutions of Weierstrass models,
  arXiv:1807.08755 [hep-th].

\bibitem{SU2G2} 
  M.~Esole and M.~J.~Kang,
  The Geometry of the SU(2)$\times$ G$_2$-model,
  JHEP {\bf 1902}, 091 (2019)

\bibitem{SO4} 
  M.~Esole and M.~J.~Kang,
  Flopping and Slicing: SO(4) and Spin(4)-models,
  arXiv:1802.04802 [hep-th].

\bibitem{EKY2} 
  M.~Esole, M.~J.~Kang and S.~T.~Yau,
  Mordell--Weil Torsion, Anomalies, and Phase Transitions,
  arXiv:1712.02337 [hep-th].

\bibitem{EKY} 
  M.~Esole, M.~J.~Kang and S.~T.~Yau,
  A New Model for Elliptic Fibrations with a Rank One Mordell--Weil Group: I. Singular Fibers and Semi-Stable Degenerations,
  arXiv:1410.0003 [hep-th].

\bibitem{Esole:2012tf} 
  M.~Esole and R.~Savelli,
 Tate Form and Weak Coupling Limits in F-theory,
  JHEP {\bf 1306}, 027 (2013)

\bibitem{ES} 
  M.~Esole and S.~H.~Shao,
  M-theory on Elliptic Calabi--Yau Threefolds and 6d Anomalies,
  arXiv:1504.01387 [hep-th].

\bibitem{ESY1}
M.~Esole, S.-H. Shao, and S.-T. Yau.
\newblock {Singularities and Gauge Theory Phases}.
\newblock {\em Adv. Theor. Math. Phys.}, 19:1183--1247, 2015.

\bibitem{ESY2} 
  M.~Esole, S.~H.~Shao and S.~T.~Yau,
  Singularities and Gauge Theory Phases II,
  Adv.\ Theor.\ Math.\ Phys.\  {20}, 683 (2016)

\bibitem{EY} 
  M.~Esole and S.~T.~Yau,
  Small resolutions of SU(5)-models in F-theory,
  Adv.\ Theor.\ Math.\ Phys.\   {17}, no. 6, 1195 (2013)

\bibitem{Erler} 
  J.~Erler,
  Anomaly cancellation in six-dimensions, 
  J.\ Math.\ Phys.\  {\bf 35}, 1819 (1994)

\bibitem{Ferrara:1996hh} 
  S.~Ferrara, R.~R.~Khuri and R.~Minasian,
  M theory on a Calabi--Yau manifold,
  Phys.\ Lett.\ B {\bf 375}, 81 (1996)



\bibitem{Fullwood:SVW}
J.~Fullwood.
\newblock {On generalized Sethi-Vafa-Witten formulas}.
\newblock {\em J. Math. Phys.}, 52:082304, 2011.

\bibitem{Gimon:1996rq} 
  E.~G.~Gimon and J.~Polchinski,
   Consistency conditions for orientifolds and d manifolds, 
  Phys.\ Rev.\ D {\bf 54}, 1667 (1996)

\bibitem{Grassi:2018rva} 
  A.~Grassi and T.~Weigand,
  On topological invariants of algebraic threefolds with ($\mathbb Q$-factorial) singularities,
  arXiv:1804.02424 [math.AG].

\bibitem{Grassi:2014zxa} 
  A.~Grassi, J.~Halverson, J.~Shaneson and W.~Taylor,
  Non-Higgsable QCD and the Standard Model Spectrum in F-theory,
  JHEP {\bf 1501}, 086 (2015)
  
\bibitem{GM1}
A.~Grassi and D.~R. Morrison.
\newblock Group representations and the {E}uler characteristic of elliptically
  fibered Calabi--Yau threefolds.
\newblock {\em J. Algebraic Geom.}, 12(2):321--356, 2003.

\bibitem{Green:1984bx} 
  M.~B.~Green, J.~H.~Schwarz and P.~C.~West,
  Anomaly Free Chiral Theories in Six-Dimensions,
  Nucl.\ Phys.\ B {\bf 254}, 327 (1985).


\bibitem{Hayashi:2014kca} 
  H.~Hayashi, C.~Lawrie, D.~R.~Morrison and S.~Schafer-Nameki,
  Box Graphs and Singular Fibers,
  JHEP {\bf 1405}, 048 (2014)
  doi:10.1007/JHEP05(2014)048

\bibitem{IMS}
K.~A. Intriligator, D.~R. Morrison, and N.~Seiberg.
\newblock {Five-dimensional supersymmetric gauge theories and degenerations of
  Calabi--Yau spaces}.
\newblock {\em Nucl.Phys.}, B497:56--100, 1997.

\bibitem{Katz:2011qp}
S.~Katz, D.~R. Morrison, S.~Sch\"afer-Nameki, and J.~Sully.
\newblock {Tate's algorithm and F-theory}.
\newblock {\em JHEP}, 1108:094, 2011.

\bibitem{Katz:1996xe} 
  S.~H.~Katz and C.~Vafa,
   Matter from geometry,
  Nucl.\ Phys.\ B {\bf 497}, 146 (1997)

\bibitem{Kawamata.flops}
Y.~Kawamata. Flops connect minimal models, Pulb. RIMS, Kyoto Univ. 44 (2008), 419--423

\bibitem{Kawamata.CY}
Y. Kawamata, On the cone of divisors of Calabi-Yau fibre spaces, Internal J.
Math. 8 (1997), 665--687.


\bibitem{Kawamata} Y.~Kawamata. Crepant blowing-up of 3-dimensional canonical singularities and its application to degenerations of surfaces. Ann. of Math. (2), 127(1):93--163, 1988.

\bibitem{KMFinite}
Y.~Kawamata, and K.~Matsuki, The Number of the Minimal Models for a 3- Fold of General Type is Finite, Math.  Ann. 276, 595--598, 1987. 

\bibitem{Kodaira}
K.~Kodaira.
\newblock On compact analytic surfaces. {II}, {III}.
\newblock {\em Ann. of Math. (2) 77 (1963), 563--626; ibid.}, 78:1--40, 1963.
%

\bibitem{KM}
J.~Koll\'{a}r and S.~Mori, Birational geometry of algebraic varieties, Cambridge Tracts in Mathematics, vol. 134, Cambridge University Press, Cambridge,
1998. 


\bibitem{Kovacs}
S.~ Kov\'{a}c, Quotient singularities with no crepant resolution?, URL (version: 2011-06-03): https://mathoverflow.net/q/66702


\bibitem{Krause:2011xj} 
  S.~Krause, C.~Mayrhofer and T.~Weigand,
  $G_4$ flux, chiral matter and singularity resolution in F-theory compactifications,
  Nucl.\ Phys.\ B {\bf 858}, 1 (2012)
  doi:10.1016/j.nuclphysb.2011.12.013
  [arXiv:1109.3454 [hep-th]].


\bibitem{Kumar:2010ru} 
  V.~Kumar, D.~R.~Morrison and W.~Taylor,
  Global aspects of the space of 6D N = 1 supergravities,
  JHEP {\bf 1011}, 118 (2010)


\bibitem{Lawrie:2012gg} 
  C.~Lawrie and S.~Sch\"afer-Nameki,
  The Tate Form on Steroids: Resolution and Higher Codimension Fibers,'
  JHEP {\bf 1304}, 061 (2013).

\bibitem{Matsuki.Weyl}
K.~Matsuki,  Weyl groups and birational transformations among minimal
  models, \href{http://dx.doi.org/10.1090/memo/0557}{{\em Mem. Amer. Math.
  Soc.} {\bfseries 116} no.~557, (1995) vi+133}.
  
  \bibitem{Matsuki.book}
K.~Matsuki, Introduction to the Mori program, Universitext, Springer-Verlag, New
York, 2002

  
  
\bibitem{Mayrhofer:2014opa} 
  C.~Mayrhofer, D.~R.~Morrison, O.~Till and T.~Weigand,
   Mordell--Weil Torsion and the Global Structure of Gauge Groups in F-theory, 
  JHEP {\bf 1410}, 16 (2014)
  
\bibitem{Miranda.smooth}
R.~Miranda.
\newblock Smooth models for elliptic threefolds.
\newblock In {\em The birational geometry of degenerations ({C}ambridge,
  {M}ass., 1981)}, volume~29 of {\em Progr. Math.}, pages 85--133.
  Birkh{\"a}user Boston, Mass., 1983.

\bibitem{Monnier:2017oqd} 
  S.~Monnier, G.~W.~Moore and D.~S.~Park,
  ``Quantization of anomaly coefficients in 6D $\mathcal{N}=(1,0)$ supergravity,''
  JHEP {\bf 1802}, 020 (2018)

\bibitem{Morrison:2011mb} 
  D.~R.~Morrison and W.~Taylor,
  Matter and singularities,
  JHEP {\bf 1201}, 022 (2012)
 
\bibitem{Morrison:2012np} 
  D.~R.~Morrison and W.~Taylor,
  Classifying bases for 6D F-theory models,
  Central Eur.\ J.\ Phys.\  {\bf 10}, 1072 (2012)

\bibitem{Morrison:1996pp}
D.~R. Morrison and C.~Vafa.
\newblock {Compactifications of F theory on Calabi--Yau threefolds. 2.}
\newblock {\em Nucl. Phys.}, B476:437--469, 1996.

\bibitem{MumfordSuominen} 
D.~Mumford and K.~Suominen, {\it Introduction to the theory of moduli}, in Algebraic Geometry, Oslo 1970, Proceedings of the 5th Nordic 
summer school in Math, Wolters-Noordhoff, 1972, 171-222. 


\bibitem{Park} 
  D.~S.~Park,
  Anomaly Equations and Intersection Theory,
  JHEP {\bf 1201}, 093 (2012)
  
 \bibitem{Salam} 
  S.~Randjbar-Daemi, A.~Salam, E.~Sezgin and J.~A.~Strathdee,
  An Anomaly Free Model in Six-Dimensions,
  Phys.\ Lett.\  {\bf 151B}, 351 (1985).

\bibitem{Sadov:1996zm} 
  V.~Sadov,
  Generalized Green-Schwarz mechanism in F theory,
  Phys.\ Lett.\ B {\bf 388}, 45 (1996).

\bibitem{Sagnotti:1992qw} 
  A.~Sagnotti,
  A Note on the Green-Schwarz mechanism in open string theories,
  Phys.\ Lett.\ B {\bf 294}, 196 (1992)

\bibitem{Schwarz:1995zw} 
  J.~H.~Schwarz,
  Anomaly - free supersymmetric models in six-dimensions,
  Phys.\ Lett.\ B {\bf 371}, 223 (1996)
 
\bibitem{Sethi:1996es} 
  S.~Sethi, C.~Vafa and E.~Witten,
  Constraints on low dimensional string compactifications, 
  Nucl.\ Phys.\ B {\bf 480}, 213 (1996)

\bibitem{Szydlo.Thesis}
M.~G. Szydlo.
\newblock {\em Flat regular models of elliptic schemes}.
\newblock ProQuest LLC, Ann Arbor, MI, 1999.
\newblock Thesis (Ph.D.)--Harvard University.

\bibitem{Tatar:2012tm} 
  R.~Tatar and W.~Walters,
  GUT theories from Calabi--Yau 4-folds with SO(10) Singularities,
  JHEP {\bf 1212}, 092 (2012)
  doi:10.1007/JHEP12(2012)092
  [arXiv:1206.5090 [hep-th]].


\bibitem{Tate.Alg}
J.~Tate, 
Algorithm for determining the type of a singular fiber in an elliptic pencil. {\it Modular functions of one variable}, IV ({\it Proc. Internat. Summer School}, Univ. Antwerp, Antwerp, 1972), pp. 33--52. Lecture Notes in Mathematics 476, Springer-Verlag, Berlin, 1975.

\bibitem{vanRitbergen:1998pn} 
  T.~van Ritbergen, A.~N.~Schellekens and J.~A.~M.~Vermaseren,
  Group theory factors for Feynman diagrams,
  Int.\ J.\ Mod.\ Phys.\ A {\bf 14}, 41 (1999)

\bibitem{Witten}
E.~Witten, Phase transitions in M theory and F theory, Nucl. Phys. B {\bf 471}, 195 (1996).

\bibitem{Witten:1995ex} 
  E.~Witten,
   String theory dynamics in various dimensions,
  Nucl.\ Phys.\ B {\bf 443}, 85 (1995)


\bibitem{Witten:1982fp} 
  E.~Witten,
  An SU(2) Anomaly,
  Phys.\ Lett.\ B {\bf 117}, 324 (1982)
  [Phys.\ Lett.\  {\bf 117B}, 324 (1982)].

\end{thebibliography}

\end{document}